\documentclass[11pt,twoside]{article}
\newcommand{\ifims}[2]{#1} 
\newcommand{\ifAMS}[2]{#1}   
\newcommand{\ifau}[3]{#1}  

\newcommand{\ifbook}[2]{#1}   

\ifbook{
    \newcommand{\Chapter}[1]{\section{#1}}
    \newcommand{\Section}[1]{\subsection{#1}}
    \newcommand{\Subsection}[1]{\subsubsection{#1}}

  }
  {
    \newcommand{\Chapter}[1]{\chapter{#1}}
    \newcommand{\Section}[1]{\section{#1}}
    \newcommand{\Subsection}[1]{\subsection{#1}}

  }

\def\thetitle{Penalized maximum likelihood estimation and effective dimension}
\def\thanksa
{The author is partially supported by
Laboratory for Structural Methods of Data Analysis in Predictive Modeling, MIPT, 
RF government grant, ag. 11.G34.31.0073.
Financial support by the German Research Foundation (DFG) through the Collaborative 
Research Center 649 ``Economic Risk'' is gratefully acknowledged} 
\def\theruntitle{Penalized MLE and effective dimension}

\def\theabstract{
This paper extends some prominent statistical results including \emph{Fisher Theorem and 
Wilks phenomenon} to the penalized maximum likelihood estimation with a quadratic 
penalization. 
It appears that sharp expansions for the penalized MLE \( \tilde{\thetav}_{\GP} \) and for the penalized maximum likelihood can be obtained without involving any asymptotic arguments,
the results only rely on smoothness and regularity properties of the of the considered log-likelihood function. 
The error of estimation is specified in terms of the effective dimension \( \dimG \)
of the parameter set which can be much smaller than the true parameter dimension and even 
allows an infinite dimensional functional parameter. 
In the i.i.d. case, the Fisher expansion for the 
penalized MLE can be established under the constraint ``\( \dimG^{2}/\nsize \) is 
small'' while the remainder in the Wilks result is of order 
\( \dimG^{3}/\nsize \). 
}

\def\kwdp{62F10}
\def\kwds{62J12,62F25,62H12}

\def\thekeywords{penalty, Wilks and Fisher expansions}

\def\authora{Vladimir Spokoiny}
\def\runauthora{spokoiny, v.}
\def\addressa{
    Weierstrass-Institute and Humboldt University Berlin, \\ Moscow Institute of
    Physics and Technology,
    \\
    Mohrenstr. 39, 10117 Berlin, Germany,    \\
    }
\def\emaila{spokoiny@wias-berlin.de}
\def\affiliationa{Weierstrass-Institute and Humboldt University Berlin}

\renewcommand{\(}{$\,}
\renewcommand{\)}{\,$}

\def\nquad{\hspace{-1cm}}
\def\eqdef{\stackrel{\operatorname{def}}{=}}

\def\tow{\stackrel{w}{\longrightarrow}}
\def\toP{\stackrel{\P}{\longrightarrow}}

\def\ND{\mathcal{N}}

\newcommand{\cc}[1]{\mathscr{#1}}
\newcommand{\bb}[1]{\boldsymbol{#1}}

\renewcommand{\bar}[1]{\overline{#1}}

\renewcommand{\tilde}[1]{\widetilde{#1}}

\renewcommand{\Gamma}{\varGamma}
\renewcommand{\Pi}{\varPi}
\renewcommand{\Sigma}{\varSigma}
\renewcommand{\Delta}{\varDelta}
\renewcommand{\Lambda}{\varLambda}
\renewcommand{\Psi}{\varPsi}
\renewcommand{\Phi}{\varPhi}
\renewcommand{\Theta}{\varTheta}
\renewcommand{\Omega}{\varOmega}
\renewcommand{\Xi}{\varXi}
\renewcommand{\Upsilon}{\varUpsilon}

\def\Var{\operatorname{Var}}

\def\argmax{\operatornamewithlimits{argmax}}

\def\tr{\operatorname{tr}}

\def\R{I\!\!R}
\def\E{I\!\!E}
\def\P{I\!\!P}

\def\kappa{\varkappa}

\def\T{\top}
\def\diag{\operatorname{diag}}
\def\diam{\operatorname{diam}}

\def\uv{\bb{u}}

\def\Uv{\bb{U}}

\def\Yv{\bb{Y}}

\def\gammav{\bb{\gamma}}

\def\varepsilonv{\bb{\varepsilon}}

\def\gammav{\bb{\gamma}}

\def\xiv{\bb{\xi}}

\def\Psiv{\bb{\Psi}}
\def\CONST{\mathtt{C}}

\newcommand{\tobedone}[1]{\par\textbf{\color{red}To be done:} {\color{magenta}#1}}

\usepackage{color}

\definecolor{blue(pigment)}{rgb}{0.2, 0.2, 0.6}
\definecolor{ultramarine}{rgb}{0.07, 0.04, 0.56}
\definecolor{darkspringgreen}{rgb}{0.09, 0.45, 0.27}
\definecolor{hookersgreen}{rgb}{0.0, 0.44, 0.0}
\definecolor{plum(traditional)}{rgb}{0.56, 0.27, 0.52}
\definecolor{purple(html/css)}{rgb}{0.5, 0.0, 0.5}
\definecolor{magenta(dye)}{rgb}{0.79, 0.08, 0.48}

\def\block{\operatorname{block}}

\def\diam{\operatorname{diam}}

\def\oper{\operatorname{op}}

\def\vol{\operatorname{vol}}

\def\CONST{\mathtt{C} \hspace{0.1em}}

\def\nsize{{n}}

\def\ex{\mathrm{e}}

\def\ND{\cc{N}}
\def\gaussv{\bb{\gauss}}
\def\gauss{\gamma}

\def\Id{I\!\!\!I}
\def\Ind{\operatorname{1}\hspace{-4.3pt}\operatorname{I}}

\def\alp{\alpha}

\def\aexpzeta{a_{\expzeta}}
\def\aGLMlink{a_{\GLMlink}}

\def\AA{A}

\def\AssId{\mathcal{I}}

\def\assId{\iota}

\def\AnGP{A_{\nul,\GP}}

\def\bias{b}

\def\biasGP{\bias_{\GP}}

\def\BB{I\!\!B}
\def\B{\cc{B}}

\def\BU{\mathcal{B}}

\def\BB{B}

\def\BBGP{\BB_{\GP}}

\def\cdim{\mathfrak{c}}

\def\cdimb{\cdim_{1}}

\def\cdimg{\cdim_{2}}

\def\Covm{\Sigma}			
\def\Covm{\mathbb{V}}
\def\covm{\sigma}

\def\CS{\cc{E}}
\def\DP{D}
\def\DPc{\DP}

\def\DPcc{\DP_{\nul}}
\def\DPnGP{\DP_{1,\GP}}
\def\DPnGPr{\DRr_{1,\GP}}

\def\DPGP{\DP_{\GP}}

\def\dist{d}

\def\dimp{p}

\def\dimA{\mathtt{p}}
\def\dimB{\dimA}

\def\dimAg{\dimA^{*}}
\def\dime{\dimA_{e}}
\def\dimG{\dimA_{\GP}}

\def\dimh{p_{1}}
\def\dimq{q}

\def\dimn{\dimp_{\nsize}}

\def\dPsi{\td_{\Psi}}

\def\Ellips{\cc{E}}

\def\Excgr{\diamondsuit}

\def\ExcGP{\err_{\GP}}

\def\expzeta{\mathtt{s}}

\def\err{\diamondsuit}
\def\errm{\err_{\rdm}}				
\def\errb{\err_{\rdb}}				

\def\eps{\epsilon}			
\def\eps{\varepsilon}

\def\entrl{\mathbb{Q}}

\def\entrlq{\entrl_{1}}		
\def\entrlg{\entrl_{2}}

\def\entrg{\mathbb{G}}
\def\entrgq{\entrg_{1}}		
\def\entrgg{\entrg_{2}}

\def\elli{\bar{\ell}}

\def\fis{\mathfrak{a}}

\def\fisGP{\mathtt{w}_{\GP}}		
\def\fisGP{\fis_{\GP}}

\def\gm{\mathtt{g}}
\def\gmc{\gm_{c}}
\def\gmb{\gm}
\def\gmbm{\gmb_{1}}
\def\gmd{\gm_{0}}
\def\gmi{\mathtt{b}}
\def\gmiid{\gm_{1}}
\def\gmiGP{\gmi_{\GP}}

\def\gp{g}

\def\GLMlink{g}
\def \GLMLINK{A}

\def\GP{G}

\def\GPn{\GP_{1}}

\def\GQF{Q}

\def\GVS{S}

\def\HG{H}

\def\HGLsum{L}
\def\hg{h}

\def\IF{\Bbb{F}}
\def\IFGP{\IF_{\GP}}

\def\IFon{\mathsf{H}}

\def\IFonec{\IFon_{0}}

\def\kullb{\cc{K}} 

\def\kb{k^{*}}

\def\LT{L}
\def\LGP{\LT_{\GP}}

\def\La{\mathbb{L}}
\def\Lab{\La_{\rdb}}
\def\Lam{\La_{\rdm}}

\def\LL{\cc{L}}

\def\lambdam{\gm_{1}}
\def\lambdaB{{\lambda}^{*}}		
\def\lambdaB{\lambda_{\BB}}		
\def\lambdaB{\supA}

\def\lambdaGP{\lambda_{\GP}}

\def\muc{\mu_{c}}

\def\mes{\pi}
\def\mesd{\mes^{\circ}}

\def\MM{\cc{M}}

\def\nablan{\nabla_{\nul}}
\def\nablaGLM{S}

\def\normc{\delta}

\def\nunu{\nu_{0}}

\def\nul{\mathrm{o}}

\def\NN{\mathbb{N}}

\def\penr{\operatorname{pen}}

\def\Pdom{\mu_{0}}
\def\PDOM{\bb{\mu}_{0}}

\def\Proj{\Pi}

\def\partition{\cc{F}}

\def\QQ{\mathbb{H}}
\def\QQg{\QQ_{2}}
\def\QQq{\QQ_{1}}

\def\QL{W}

\def\reps{\epsilon}

\def\rdl{\epsilon}
\def\rd{\bb{\rdl}}
\def\rddelta{\delta}
\def\rdomega{\varrho}

\def\rddeltab{\rddelta^{*}}
\def\rdomegaGP{\rdomega_{\GP}}
\def\rddeltaGP{\rddelta_{\GP}}
\def\rdb{\rd}
\def\rdm{\underline{\rdb}}
\def\rderr{\chi}

\def\rhor{\omega}

\def\rhorb{\rhor^{*}}

\def\riskt{\cc{R}}

\def\risktGP{\riskt_{\GP}}

\def\rr{\mathtt{r}}
\def\rrb{\rr^{*}}

\def\rrf{\mathfrak{r}}

\def\rA{\rupd}

\def\rups{\rr_{0}}

\def\rupd{\rr_{\circ}}

\def\rupsGP{\rr_{\GP}}

\def\SPiS{\Upsilon}

\def\spread{\Delta}

\def\score{\nabla}
\def\scorer{\breve{\nabla}}
\def\scoren{\score_{\nul}}

\def\supA{\lambda}
\def\supAB{\supA^{*}}

\def\smooths{\mathbb{S}}

\def\thetav{\bb{\theta}}
\def\thetavs{\thetav^{*}}
\def\thetavc{\thetav'}
\def\thetavd{\thetav^{\circ}}

\def\thetavsGP{\thetavs_{\GP}}

\def\thetavn{\thetav_{\nul}}
\def\thetavsn{\thetavs_{\nul}}

\def\tilden#1{\tilde{#1}_{\nul}}

\def\testst{T}

\def\td{\delta}

\def\Thetas{\Theta_{0}}
\def\ThetasGP{\Theta_{0,\GP}}

\def\Thetan{\Theta_{\nul}}

\def\TGP{\testst_{\GP}}

\def\upsi{\upsilon}
\def\ups{\bb{\upsilon}}

\def\upsc{\ups^{\prime}}
\def\upsd{\ups^{\circ}}
\def\upsdc{\ups^{\sharp}}
\def\upsdu{\ups^{\flat}}
\def\upss{\ups^{*}}

\def\upsv{\bb{\varkappa}}

\def\upsdGP{\ups_{\GP}}

\def\UP{\cc{U}}

\def\UPb{\UP^{*}}

\def\UU{\cc{Y}}

\def\Ups{\varUpsilon}
\def\Upsd{\Ups^{\circ}}
\def\Upss{\Ups_{\circ}}

\def\vA{\mathtt{v}}

\def\VL{\mathscr{X}}

\def\VP{V}
\def\VPc{\VP_{0}} 	
\def\VPc{\VP}

\def\VPD{\VP_{2}}

\def\VV{\mathbb{V}}
\def\VVb{\VV^{*}}	
\def\VVc{\VV}

\def\VVb{\VVc}
\def\VVGP{\VV_{\GP}}
\def\vp{\mathbf{v}}	
\def\vpc{\vp_{0}}
\def\vp{\mathsf{v}}

\def\vpB{\vp}

\def\xivr{\breve{\xiv}}

\def\xis{\xi^{*}}

\def\xivGP{\xiv_{\GP}}

\def\xivn{\xiv_{\nul}}
\def\xivrGP{\xivr_{\GP}}

\def\xx{\mathtt{x}}
\def\xxc{\xx_{c}}

\def\yy{\mathtt{y}}

\def\yyd{\yy_{0}}		

\def\YY{\cc{Y}}

\def\zq{z}
\def\zqc{\zq_{c}}

\def\zz{\mathfrak{z}}

\def\zzQ{\zz_{0}}

\def\zzQ{\zz_{\QQ}}

\usepackage{amsmath,amssymb,amsthm}
\usepackage{natbib}
\usepackage{epsfig,graphicx}
\usepackage{comment}
\usepackage{color}
\usepackage{srcltx}
\usepackage[mathscr]{eucal}
\usepackage[math]{easyeqn}
\usepackage{etoolbox}
\usepackage{hyperref}
\hypersetup{
            colorlinks,
            linkcolor=hookersgreen,
            linktoc=blue,
            citecolor=ultramarine,
            urlcolor=black,
            filecolor=black
            }
\usepackage{nameref}

\ifims{
\textheight=23cm
\textwidth=14.8cm
\topmargin=0pt
\oddsidemargin=1.0cm
\evensidemargin=1.0cm
\linespread{1.3}
\renewenvironment{abstract}
    {\centerline{\textbf{Abstract}}\bigskip
      \begin{center}
       \begin{minipage}{11cm}
        \begin{small}
    }
    {   \end{small}
       \end{minipage}
      \end{center}
     \bigskip
    }

}{ 
}

\numberwithin{equation}{section}
\numberwithin{figure}{section}
\newcounter{example}[section]
\numberwithin{example}{section}
\newcounter{remark}[section]
\numberwithin{remark}{section}
\newtheorem{theorem}{Theorem}[section]
\newtheorem{proposition}[theorem]{Proposition}
\newtheorem{lemma}[theorem]{Lemma}
\newtheorem{corollary}[theorem]{Corollary}

\newtheorem{exmp}[example]{Example}
\newtheorem{rmrk}[remark]{Remark}
\newenvironment{example}{\begin{exmp}\rm}{\end{exmp}}
\newenvironment{remark}{\begin{rmrk}\rm}{\end{rmrk}}

\begin{document}
\thispagestyle{empty}
\ifims{
\title{\thetitle}
\ifau{ 
  \author{
    \authora
    \ifdef{\thanksa}{\thanks{\thanksa}}{}
    \\[5.pt]
    \addressa \\
    \texttt{ \emaila}
  }
}
{  
  \author{
    \authora
    \ifdef{\thanksa}{\thanks{\thanksa}}{}
    \\[5.pt]
    \addressa \\
    \texttt{ \emaila}
    \and
    \authorb
    \ifdef{\thanksb}{\thanks{\thanksb}}{}
    \\[5.pt]
    \addressb \\
    \texttt{ \emailb}
  }
}
{   
  \author{
    \authora
    \ifdef{\thanksa}{\thanks{\thanksa}}{}
    \\[5.pt]
    \addressa \\
    \texttt{ \emaila}
    \and
    \authorb
    \ifdef{\thanksb}{\thanks{\thanksb}}{}
    \\[5.pt]
    \addressb \\
    \texttt{ \emailb}
    \and
    \authorc
    \ifdef{\thanksc}{\thanks{\thanksc}}{}
    \\[5.pt]
    \addressc \\
    \texttt{ \emailc}
  }
}

\maketitle
\pagestyle{myheadings}
\markboth
 {\hfill \textsc{ \small \theruntitle} \hfill}
 {\hfill
 \textsc{ \small
 \ifau{\runauthora}
      {\runauthora and \runauthorb}
      {\runauthora, \runauthorb, and \runauthorc}
 }
 \hfill}
\begin{abstract}
\theabstract
\end{abstract}

\ifAMS
    {\par\noindent\emph{AMS 2000 Subject Classification:} Primary \kwdp. Secondary \kwds}
    {\par\noindent\emph{JEL codes}: \kwdp}

} 
{ 
\begin{frontmatter}
\title{\thetitle}


\runtitle{\theruntitle}

\ifau{ 
\begin{aug}
    \author{\authora\ead[label=e1]{\emaila}}
    \address{\addressa \\
     \printead{e1}}
\end{aug}

 \runauthor{\runauthora}
\affiliation{\affiliationa} }
{ 
\begin{aug}
    \author{\authora\ead[label=e1]{\emaila}\thanksref{t21}}
    \and
    \author{\authorb\ead[label=e2]{\emailb}\thanksref{t22}}
    
    \address{\addressa \\
     \printead{e1}}
    \address{\addressb \\
     \printead{e2}}
    \thankstext{t21}{\thanksa}
    \thankstext{t22}{\thanksb}
    \affiliation{\affiliationa, \affiliationb} 
    \runauthor{\runauthora and \runauthorb}
\end{aug}
} 
{ 
\begin{aug}
    \author{\authora\ead[label=e1]{\emaila}\thanksref{t21}}
    \and
    \author{\authorb\ead[label=e2]{\emailb}\thanksref{t22}}
    \and
    \author{\authorc\ead[label=e3]{\emailc}\thanksref{t23}}
    
    \address{\addressa \\
     \printead{e1}}
    \address{\addressb \\
     \printead{e2}}
    \address{\addressc \\
     \printead{e3}}
    \thankstext{t21}{\thanksa}
    \thankstext{t22}{\thanksb}
    \thankstext{t23}{\thanksc}
    \affiliation{\affiliationa, \affiliationb, \affiliationc} 
    \runauthor{\runauthora, \runauthorb, and \runauthorc}
\end{aug}}

\begin{abstract}
\theabstract
\end{abstract}



\end{frontmatter}
} 

\tableofcontents


\ifbook{
\Chapter{Introduction}
\label{Swilksint}
}{}
The Fisher and Wilks Theorems 
belong to the short list of most fascinating results in the statistical theory. 
In particular, the Wilks result 
in its simple form claims that the likelihood ratio test statistic is close in 
distribution to the \( \chi^{2}_{\dimp} \) distribution as the sample size increases,
where \( \dimp \) means the parameter dimension.
So, the limiting distribution of this test statistic only depends on the dimension 
of the parameter space whatever the parametric model is. 
This explains why this result is sometimes called the \emph{Wilks phenomenon}.
This paper aims at reconsidering the mentioned results from different viewpoints.
One important issue is that the presented results are 
stated for \emph{finite samples}.
There are only few general finite-sample results in statistical inference; 
see \cite{BoMa2011} and references therein in context of i.i.d. modeling. 
The novel approach from \cite{SP2011} offered a general framework for a \emph{finite 
sample theory}, and the present paper makes a further step in this direction: 
the classical large sample results are extended 
to the finite sample case with \emph{explicit and sharp} error bounds.

Another important point is a possible \emph{model misspecification}.
The classical parametric theory requires the parametric assumption to be exactly fulfilled.
Any violation of the parametric specification may destroy the Fisher and Wilks 
results; cf. \cite{huber1967}. 
This study admits from the very beginning that the parametric specification 
is probably wrong. 
This automatically extends the applicability of the proposed approach. 

The further issue is the use of \emph{penalization} for reducing the \emph{model 
complexity}. 
If the parameter dimension is too large, the classical statistical results become almost 
intractable because the corresponding error is proportional to the dimension of 
parameter space. 
Sieve parametric approach is often used to replace the an infinite dimensional problem with 
a finite dimensional one; see e.g. 
\cite{ShWo1994},
\cite{shen1997},
\cite{sara2000},
\cite{BiMa1998,BBiMa1999},
and references therein.
Some asymptotic results for generalized regression models are available in 
\cite{FaZh2001}. 

Another standard way of reducing the complexity of the model is by introducing some penalty 
in the likelihood function.
In this paper we focus on quadratic-type penalization.
Roughness penalty approach provides a popular example; cf. \cite{greensil1994}.
\cite{PoKo1994} explained how roughness penalty works in context of quantile regression.
Tikhonov regularization and ridge regression are the other examples which are often 
used in linear inverse problems.
It is well known that the use of a penalization in context of an inverse problem 
provides regularization and uncertainty reduction at the same time. 
Our results show that the use of penalization indeed leads to some improvement in the 
obtained error bounds. Namely, one can replace the original parameter dimension 
\( \dimp \) by the so called \emph{effective dimension} \( \dimG \) which can be much 
smaller than \( \dimp \). 
Even the case of a functional parameter \( \thetav \) with \( \dimp = \infty \) can be 
included. 
In this paper the penalty term is supposed to be given in advance. 
In general, a model selection procedure based on a proper choice of penalization is a high 
topic, one of the central in nonparametric statistics. 
We refer to \cite{shen1997}, \cite{BiMa1998}, \cite{vdG2002} for the general models and to 
\cite{BiMa2001,BiMa2007} for Gaussian model selection where one can find an 
extensive overview of the vast literature on this problem.

The final issue is the \emph{critical parameter dimension} which is measured by the
effective dimension \( \dimG \). 
The problem of statistical inference for models with growing parameter dimension is quite 
involved.
There are some specific issues even if a simple linear or exponential model 
is considered, the results from \cite{Po1984,Po1985} requires 
``\( \dimp^{2}/n \) small'' for asymptotic normality of the MLE.
Depending on the considered problem and the model at hand, the conditions on the 
critical parameter dimension \( \dimp \) may differ. 
For instance, \cite{Po1988} obtained the Fisher and Wilks results for a generalized linear model
under \( \dimp^{3/2}/n \to 0 \), \cite{Mammen1996} established similar results for high-dimensional 
linear models.
A general Wilks result can be stated under the condition 
that \( \dimp^{3}/n \) is small; see e.g. \cite{Cherno2009}.
Below we show that
the conditions on the critical dimension in penalized ML estimation can be given 
in terms of the effective dimension \( \dimG \) rather than the parameter dimension \( \dimp \).
In particular, in the i.i.d. case, the Fisher expansion can be stated under ``\( \dimG^{2}/n \) small''
and ``\( \dimG^{3}/n \) small'' is sufficient for the Wilks result.

\medskip

First we specify our set-up.
Let \( \Yv \) denote the observed data and \( \P \) mean their distribution.
A general parametric assumption (PA) means that 
\( \P \) belongs to \( \dimp \)-dimensional family 
\( (\P_{\thetav}, \thetav \in \Theta \subseteq \R^{\dimp}) \) dominated by 
a measure \( \PDOM \).
This family yields the log-likelihood function 
\( L(\thetav) = L(\Yv,\thetav) \eqdef \log \frac{d\P_{\thetav}}{d\PDOM}(\Yv) \).
The PA can be misspecified, so, in general, \( L(\thetav) \) is a 
\emph{quasi log-likelihood}. 
The classical likelihood  principle suggests to estimate  \( \thetav \) by 
maximizing the function \( L(\thetav) \):
\begin{EQA}[c]
    \tilde{\thetav}
    \eqdef
    \argmax_{\thetav \in \Theta} L(\thetav) .
\label{tthetamkGP} 
\end{EQA}
If \( \P \not\in \bigl( \P_{\thetav} \bigr) \), then 
the quasi MLE estimate \( \tilde{\thetav} \) from \eqref{tthetamkGP} is 
still meaningful and it can be viewed as an estimate of the value \( \thetavs \) 
defined by maximizing the expected value of \( L(\thetav) \): 
\begin{EQA}[c]
    \thetavs
    \eqdef
    \argmax_{\thetav \in \Theta} \E L(\thetav)
\label{thetavsdGP}
\end{EQA}
which is the true value in the parametric situation and can be viewed as the 
parameter of the best parametric fit in the general case. 

%
The classical \emph{Fisher Theorem} claims the expansion for the MLE 
\( \tilde{\thetav} \):
\begin{EQA}[c]
    \DPc \bigl( \tilde{\thetav} - \thetavs \bigr) - \xiv
    \toP
    0 ,
\label{DPcttxiv}
\end{EQA}    
where \( \DPc^{2} = - \nabla^{2} \E L(\thetavs) \) and 
\( \xiv \eqdef \DPc^{-1} \nabla L(\thetavs) \).
Under the correct model specification, \( \DPc^{2} \) is the total Fisher information 
matrix and the vector \( \xiv \) is centered and standardized. 
So, it is asymptotically standard normal under general CLT conditions. 

It is well known 
that many important properties of the quasi MLE 
\( \tilde{\thetav} \) like concentration or coverage probability can be described 
in terms of the \emph{excess} or \emph{quasi maximum likelihood} 
\( L(\tilde{\thetav},\thetavs) \eqdef
    L(\tilde{\thetav}) - L(\thetavs)
    =
    \max_{\thetav \in \Theta} L(\thetav) - L(\thetavs)
\), which is the difference between the maximum of the process \( L(\thetav) \) and 
its value at the ``true'' point \( \thetavs \). 
The \emph{Wilks phenomenon} claims that the distribution of the twice excess 
\( 2 L(\tilde{\thetav},\thetavs) \) can be approximated by 
\( \| \xiv \|^{2} \) which is asymptotically \( \chi^{2}_{\dimp} \), 
where \( \dimp \) is the dimension of the parameter space:
\begin{EQA}[c]
    2 L(\tilde{\thetav},\thetavs) - \| \xiv \|^{2}
    \toP 
    0,
    \qquad 
    \| \xiv \|^{2} 
    \tow
    \chi^{2}_{\dimp} \, .
\label{Wilksintr}
\end{EQA}    
This fact is very attractive and yields asymptotic confidence and concentration sets 
as well as the limiting critical values for the likelihood ratio tests.
However, practical applications of all mentioned results are limited: 
they require true parametric distribution, large samples and a fixed parameter dimension. 

Modern applications stimulate a further extension of the classical theory beyond the 
classical parametric assumptions.
\cite{SP2011} offers a general approach which appears to be very useful for such an 
extension.
The whole approach is based on the following local bracketing result:
\begin{EQA}[c]
    \Lam(\thetav,\thetavs) - \errm
    \le 
    L(\thetav) - L(\thetavs)
    \le 
    \Lab(\thetav,\thetavs) + \errb,
    \qquad 
    \thetav \in \Thetas .
\label{LLLmbinGP}
\end{EQA}
Here \( \Lab(\thetav,\thetavs) \) and \( \Lam(\thetav,\thetavs) \) are quadratic in 
\( \thetav - \thetavs \) expressions and \( \Thetas \) is a local vicinity of the central point 
\( \thetavs \).
This result can be viewed as an extension of the famous Le Cam 
\emph{local asymptotic normality} (LAN) condition. 
The LAN condition considers just one quadratic process for approximating the 
log-likelihood \( L(\thetav) \).
The use of bracketing with two different quadratic expressions allows one to keep 
control of the error terms \( \errb, \errm \) even for relatively large neighborhoods 
\( \Thetas \) of \( \thetavs \) while the LAN approach is essentially restricted to a root-n 
vicinity of \( \thetavs \).
It also allows to incorporate a large parameter dimension and a model 
misspecification.
However, the approach from \cite{SP2011} has natural limitations: 
the parameter dimension \( \dimp \) cannot be too large. 
For instance, in the i.i.d. case, the error terms 
\( \errb \) and \( \errm \) are of order \( \sqrt{\dimp^{3}/\nsize} \) 
which destroys the Wilks result if \( \dimp > \nsize^{1/3} \).

A standard way of overcoming this difficulty is to impose a kind of smoothness
assumption on the unknown parameter value \( \thetavs \).
Here we discuss one general way to deal with such smoothness assumptions using
a quadratic penalization. 
Section~\ref{Chgrough} offers a new approach to studying the properties of the penalized MLE
which is based on a linear approximation of the gradient of the log-likelihood process.
Compared to the bracketing approach \eqref{LLLmbinGP}, it allows to establish a Fisher type expansion for the penalized MLE under weaker conditions on the critical dimension of the problem. 
Another important novelty of the approach is the systematic use of the \emph{effective dimension} 
\( \dimG \) in place of the original dimension \( \dimp \) of the parameter space. 
Usually \( \dimG \) is much smaller than \( \dimp \).
It is even possible to treat the case of a functional parameter if the effective 
dimension of the parameter set remains finite. 
Our main results include the Fisher and Wilks expansions for the penalized MLE.
In the important special case of an i.i.d. model, 
the error term in the Wilks expansion is small if \( \dimG^{3}/\nsize \) is small, 
while the Fisher expansion requires \( \dimG^{2}/\nsize \) small.

Also we discuss an implication of these results to the bias-variance decomposition of 
the squared risk of the penalized MLE. 
In all our results, the error terms only depend on the effective dimension 
\( \dimG \). 

%

The \ifbook{paper}{chapter} 
is organized as follows.
Section~\ref{Chgrough} states the analog of 
Fisher and Wilks results for the penalized MLE procedure.
Section~\ref{SproofsWilks} collects the conditions and proofs of the main results.
Section~\ref{Chgempir} presents some results from the empirical process theory which are 
used in our proofs.

\ifbook{\section{Fisher and Wilks Theorems under quadratic penalization}}
	{\Section{Fisher and Wilks Theorems under quadratic penalization}}
\label{Chgrough}
\label{Srough}
Let \( \penr(\thetav) \) be a penalty function on \( \Theta \).
A big value of \( \penr(\thetav) \) corresponds to a large degree of roughness or
a small amount of smoothness of \( \thetav \).
The underlying assumption on the model is that the true value \( \thetavs \) is smooth
in the sense that \( \penr(\thetavs) \) is relatively small.
A penalized (quasi) MLE approach leads to maximizing the penalized log-likelihood:
\begin{EQA}[c]
    \tilde{\thetav}
    =
    \argmax_{\thetav \in \Theta} \bigl\{ L(\thetav) - \penr(\thetav) \bigr\}  .
\label{pMLEro}
\end{EQA}
%
Below we discuss an important special case of a quadratic penalty
\( \penr(\thetav) = \| \GP \thetav \|^{2}/2 \) for a given symmetric matrix \( \GP \);
see e.g. \cite{greensil1994} or \cite{PoKo1994} for particular examples.
Denote
\begin{EQA}
    \LGP(\thetav)
    & \eqdef &
    L(\thetav) - \| \GP \thetav \|^{2} / 2,
    \\
    \tilde{\thetav}_{\GP}
    & \eqdef &
    \argmax_{\thetav \in \Theta} \LGP(\thetav) .
\label{LGPro}
\end{EQA}
The use of a penalty changes the target of estimation which is now 
defined as
\begin{EQA}
    \thetavsGP
    & \eqdef &
    \argmax_{\thetav \in \Theta} \E \LGP(\thetav) .
\label{LsGPro}
\end{EQA}
So, introducing a penalty leads to some estimation bias: the new target
\( \thetavsGP \) may be different from \( \thetavs \).
At the same time, similarly to linear modeling, the use of penalization reduces the 
variability of the estimate \( \tilde{\thetav}_{\GP} \) and improves its 
concentration properties. 
An interesting question is the total impact and a possible gain of using the 
penalized procedure. 
A preliminary answer is that the penalty term \( \| \GP \thetavs \|^{2} \) at the 
true point should 
not be too large relative to the squared error of estimation for the penalized model.
This rule is known under the name ``bias-variance trade-off''.

Another important message of this study is that the use of penalization allows to reduce 
the parameter dimension to the \emph{effective dimension} which can be viewed as the 
entropy of the penalized parameter space. 
The resulting confidence and concentration sets depend on the effective dimension rather 
than on the real parameter dimension and they can be much more narrow than in the 
non-penalized case. 

The principle steps of the study are as follows.
The \emph{concentration} step claims that the penalized MLE 
\( \tilde{\thetav}_{\GP} \) is concentrated in a local vicinity \( \ThetasGP(\rupsGP) \) 
of the point \( \thetavsGP \).  
It is based on the upper function method which 
bounds the penalized log-likelihood \( \LGP(\thetav) \) from above by a deterministic function.
Theorem~\ref{TLDGP} states that \( \tilde{\thetav}_{\GP} \) belongs to the local set 
\( \ThetasGP(\rupsGP) \) with a dominating probability, 
and this local set can be much smaller than the similar set for the 
non-penalized results.
As the next step, \cite{SP2011} applied the \emph{bracketing} approach to bound from 
above and from below the log-likelihood process \( L(\thetav) \) by two quadratic in 
\( \thetav - \thetavs \) expressions. 
Here the bracketing step is changed essentially by using a local linear approximation of 
the vector gradient process \( \nabla L(\thetav) \). 
This helps to get a sharper bound on the error of approximation and improve the quality 
of the Fisher expansion. 
Similarly to \cite{SP2011}, the obtained results are stated for finite samples and do not involve 
any asymptotic arguments.
An advantage of the proposed approach is that it combines an accurate local approximation with 
rather rough large deviation arguments and allows one to obtain usual asymptotic statements 
including asymptotic normality of the penalized MLE.

As an important special case, Section~\ref{Siidro} considers the i.i.d. model 
and discusses the dimensional asymptotic. 
If \( \dimG^{2} = o(\nsize) \), then the Fisher expansion is meaningful.
The Wilks expansion requires \( \dimG^{3} = o(\nsize) \).

\Section{Effective dimension}
Let 
\( \VPc^{2} \) be the matrix shown in condition \nameref{ED0Gref} 
in Section~\ref{ScondroGP}.
Typically \( \VPc^{2} = \Var\bigl\{ \nabla L(\thetavsGP) \bigr\} \) and 
this matrix measures the local variability of the process \( \LGP(\cdot) \).
Let also \( \DPGP^{2} \) be a penalized information matrix defined as 
\begin{EQA}
	\DPGP^{2} 
	&=& 
	- \nabla^{2} \E \LGP(\thetavsGP) 
	=
	\DPc^{2} + \GP^{2}
\label{DPGPdef}
\end{EQA}
with \( \DPc^{2} = - \nabla^{2} \E L(\thetavsGP) \).
One can redefine \( \DPc^{2} = - \nabla^{2} \E L(\thetavs) \) 
under condition \nameref{LL0Gref} below and the so called small modeling bias condition; 
see Section~\ref{Sqrisk}. 
The \emph{effective dimension} \( \dimG \) is defined as the trace of the matrix
\( \BBGP \eqdef \DPGP^{-1} \VPc^{2} \DPGP^{-1} \):
\begin{EQA}
	\dimG
	& \eqdef &
	\tr \bigl( \BBGP \bigr)  .
\label{dimedef}
\end{EQA}
Below we show that the use of penalization enables us to replace the original dimension 
\( \dimp \) in our risk bounds with the effective dimension
\( \dimG \) which can be much smaller than \( \dimp \) depending on relations between
the matrices \( \DPc^{2} \), \( \VPc^{2} \), and \( \GP^{2} \).

In our results the value \( \dimG \) will be used via another quantity  
\( \zq(\BBGP,\xx) \) which also depends on a fixed constant \( \xx \) and 
for moderate values of \( \xx \) can be defined as 
\begin{EQA}
	\zq(\BBGP,\xx)
	&=&
	\sqrt{\dimG} + \sqrt{2 \xx \lambdaGP} ,
\label{zzxxdefro}
\end{EQA}
where \( \lambdaGP \eqdef \lambda_{\max}\bigl( \BBGP \bigr) \) is the largest 
eigenvalue of \( \BBGP \); see \eqref{zzxxppdBlroB} for a precise definition.

\bigskip

Now we present
a couple of typical examples of using the quadratic penalty:
blockwise penalization and estimation under a Sobolev smoothness constraint.
For simplicity of presentation we assume that \( \VPc^{2} = \DPc^{2} = n \Id_{\dimp} \),
while \( \GP^{2} \) is diagonal with non-decreasing eigenvalues \( \gp_{j}^{2} \).
Then \( \DPGP^{2} = \DPc^{2} + \GP^{2} 
= \diag\bigl\{ n + \gp_{1}^{2}, \ldots, n + \gp_{\dimp}^{2} \bigr\} \).
It holds that 
\( \BBGP = \diag\bigl\{ 
(1 + n^{-1} \gp_{1}^{2})^{-1}, \ldots, (1 + n^{-1} \gp_{\dimp}^{2})^{-1} \bigr\} \), and 
we apply \eqref{dimedef} for computing the effective dimension \( \dimG \).

\paragraph{Block penalization}
Consider the case when \( \GP \) is of a simple two-block structure: 
\( \GP = \diag\{ 0,\GPn \} \).
Many blocks can be considered in the similar way.
The first block of dimension \( \dimp_{0} \) corresponds to the unconstrained part
of the parameter vector while
the second block of dimension \( \dimh \) corresponds to the low energy component.
An interesting question is the minimal penalization \( \GPn \) making the impact of
the low energy part inessential.
Assume for simplicity that 
\( \GPn = \gp \Id_{\dimh} \).
Then 
\begin{EQA}[c]
    \dimG
    = 
    \tr \BBGP 
    = 
    \dimp_{0} + \dimh / \bigl( 1 + n^{-1} \gp^{2} \bigr) .
\label{GVSblk}
\end{EQA}
One can see that the impact of the second block \( \GP_{1} \) in the effective dimension
is inessential if \( \gp^{2} / n \gg \dimh/\dimp_{0} \).

\paragraph{Sobolev smoothness constraint}
Consider the case with \( \DPc^{2} = \VPc^{2} = n \Id_{\dimp} \) and
\( \GP^{2} = \diag\{ \gp_{1}^{2},\ldots,\gp_{\dimp}^{2} \} \) with
\( \gp_{j} = L j^{\beta} \) for \( \beta > 1/2 \).
The value \( \beta \) is usually considered as the Sobolev smoothness parameter.
It holds 
\begin{EQA}[c]
    \dimG
    =
    \sum_{j=1}^{\dimp} \frac{1}{1 + L^{2} j^{2\beta} / n} \, .
\label{dimGsobo}
\end{EQA}    
Define also the index \( \dime \) as the largest \( j \) satisfying 
\( L^{2} j^{2\beta} \le n \).
It is straightforward to see that \( \beta > 1/2 \) yields 
\( \dimG \le \CONST(\beta,L) \dime \) 
for some constant \( \CONST(\beta,L) \) depending on \( \beta,L \) only.

\paragraph{Linear inverse problem}
The next example corresponds to the case of a linear inverse problem.
Assume for simplicity of notation the sequence space representation, 
the noise is inhomogeneous  with increasing eigenvalues 
\( \VPc^{2} = \diag\bigl\{ \vp_{1}^{2},\ldots,\vp_{\dimp}^{2} \bigr\} \) and 
the information matrix \( \DPc^{2} \) is proportional to identity, that is,
\( \DPc^{2} = n \Id_{\dimp} \). 
Then the effective dimension is given by the sum
\begin{EQA}[c]
    \dimG
    = 
    \sum_{j=1}^{\dimp} \frac{\vp_{j}^{2}}{n + \gp_{j}^{2}} \, .
\label{dimGinv}
\end{EQA}
To keep the effective dimension small, one has to compensate the increase of the eigenvalues 
\( \vp_{j}^{2} \) by the penalization \( \gp_{j}^{2} \).

\Section{Conditions}
\label{ScondroGP}
This section presents the list of conditions 
which are similar to ones from the 
non-penalized case in \cite{SP2011}. 
However, the use of penalization leads to some change in each condition. 
Most important fact is that the use of penalization helps to state the large deviation result 
for much smaller local neighborhoods than in the non-penalized case.
\cite{SP2011} 
presented the LD result for local sets of the form
\( \Thetas(\rr) = \bigl\{ \thetav: \| \VPc (\thetav - \thetavs) \| \le \rr \bigr\} \)
with a proper \( \rr \asymp \dimp^{1/2} \).
Now we redefine this set by using \( \DPGP^{2} \) in place of \( \VPc^{2} \)
and \( \thetavsGP \) in place of \( \thetavs \):
\begin{EQA}
    \ThetasGP(\rr)
    & \eqdef &
    \bigl\{
        \thetav: \| \DPGP(\thetav - \thetavsGP) \| \le \rr \bigr\} .
\label{ThetaRro}
\end{EQA}
Moreover, the radius \( \rr \) can be selected of order \( \dimG^{1/2} \),
which can be very useful for large or infinite \( \dimp \).
%
%

Our conditions mainly assume some regularity and smoothness of the 
penalized log-likelihood process \( \LGP(\thetav) \).
The first condition states some smoothness properties of the expected log-likelihood 
\( \E \LGP(\thetav) \) as a function 
of \( \thetav \) in a vicinity \( \ThetasGP(\rr) \) of \( \thetavsGP \).
More precisely, it 
effectively means that the expected log-likelihood \( \E L(\thetav) \) is twice continuously differentiable on the local set \( \ThetasGP(\rr) \).

Below each condition is given in penalized and non-penalized form for the sake of comparison.
Already now it is worth saying that the use of penalization helps to relax most of conditions.
%
%
%
Define 
\begin{EQA}[c]
	\IFGP(\thetav)
	\eqdef
	- \nabla^{2} \E \LGP(\thetav)
	=
	- \nabla^{2} \E L(\thetav) + \GP^{2} .
\end{EQA}
Then \( \DPGP^{2} = \IFGP(\thetavsGP) \).
The conditions involve a radius \( \rupsGP \) which separates the local zone and the zone of large 
deviations. 
This value will be made precise in Theorem~\ref{TLDGP}.

Here and below \( \| A \|_{\oper} \) means the operator norm of a matrix \( A \).

\begin{description}
    \item[\label{LL0Gref}\( \bb{(\LL_{0}\GP)} \)]
    \textit{
    For each \( \rr \leq \rupsGP \), 
    there is a constant \( \rddeltaGP(\rr) \leq 1/2 \) such that
    }
\begin{EQA}[c]
\label{LmgfquadELGP}
    \bigl\|
		\DPGP^{-1} \IFGP(\thetav) \DPGP^{-1} - \Id_{\dimp} 
    \bigr\|_{\oper}
    \le
    \rddeltaGP(\rr)  ,
    \qquad
    \thetav \in \ThetasGP(\rr).
\end{EQA}
\end{description}
Under condition \nameref{LL0Gref}, it follows from the second order Taylor
expansion at \( \thetavsGP \):
\begin{EQA}
    \bigl|
        - 2 \E \LGP(\thetav,\thetavsGP)
        - \| \DPGP (\thetav - \thetavsGP) \|^{2}
    \bigr|
    & \le &
    \rddeltaGP(\rr) \| \DPGP (\thetav - \thetavsGP) \|^{2} ,
    \quad
    \thetav \in \ThetasGP(\rr).
\label{EdeltGP}
\end{EQA}
A non-penalized version of \eqref{LmgfquadELGP} claims a similar approximation 
for the matrix function \( \IF(\thetav) = - \nabla^{2} \E L(\thetav) \) in the vicinity 
\( \Thetas(\rups) \) centered at \( \thetavs \) instead of \( \thetavsGP \):
with \( \DPc^{2} \eqdef \IF(\thetavs) \)
\begin{description}
    \item[\label{nLL0ref}\( \bb{(\LL_{0})} \)]
\( \qquad
    \bigl\|
		\DPc^{-1} \IF(\thetav) \DPc^{-1} - \Id_{\dimp} 
    \bigr\|_{\oper}
    \le
    \rddelta(\rups)  ,
    \qquad
    \thetav \in \Thetas(\rr)
\).
\end{description}

As the quadratic penalty \( \| \GP \thetav \|^{2} \) does not change the smoothness properties 
of the expected contrast \( \E \LGP(\thetav) \), the conditions \nameref{LL0Gref} and 
\nameref{nLL0ref} are essentially equivalent provided that the points 
\( \thetavs \) and \( \thetavsGP \) are too far from each others.

\medskip
Now we consider the stochastic component of the log-likelihood process \( \LGP(\thetav) \) which is the same as in the non-penalized case:
\begin{EQA}[c]
	\zeta(\thetav)
	\eqdef
	\LGP(\thetav) - \E \LGP(\thetav)
	=
	L(\thetav) - \E L(\thetav) .
\end{EQA}
We assume that it is twice differentiable and denote by \( \nabla \zeta(\thetav) \) its gradient and by 
\( \nabla^{2} \zeta(\thetav) \) its Hessian matrix. 
The next two conditions are to ensure that the random vector \( \nabla \zeta(\thetavsGP) \) and 
the random processes \( \nabla^{2} \zeta(\thetav) \) are stochastically bounded with exponential moments.
The conditions involve a \( \dimp \times \dimp \)-matrix \( \VPc \) which normalizes the vector 
\( \nabla \zeta(\thetavsGP) \), and a similar matrix \( \VPD \) normalizing 
\( \nabla^{2} \zeta(\thetav) \).

\begin{description}
\item[\label{ED0Gref}\( \bb{(E_{0}\GP)} \)]
    \emph{ There exist a positively semi-definite symmetric matrix \( \VPc^{2} \),
    and constants \( \gmb > 0 \), \( \nunu \ge 1 \) such that
    \( \Var\bigl\{ \nabla \zeta(\thetavsGP) \bigr\} \le \VPc^{2} \) and 
    }
\begin{EQA}[c]
\label{expzetaclocGP}
	\sup_{\gammav \in \R^{\dimp}} 
    \log \E \exp\biggl\{
        \lambda \frac{\gammav^{\T} \nabla \zeta(\thetavsGP)}
                     {\| \VPc \gammav \|}
    \biggr\} \le
    \frac{\nunu^{2} \lambda^{2}}{2}, \qquad 
    |\lambda| \le \gmb .
\end{EQA}

\item[\label{ED2Gref}\( \bb{(E_{2}\GP)} \)]
    \emph{There exist a positively semi-definite symmetric matrix \( \VPD^{2} \), a value \( \rhor > 0 \) 
    and for each \( \rr > 0 \), a constant \( \gm(\rr) > 0 \) such that
    it holds for any \( \thetav \in \ThetasGP(\rr) \):    }
\begin{EQA}[c]
    \sup_{\gammav_{1},\gammav_{2} \in \R^{\dimp} }
    \log \E \exp\biggl\{ 
    	\frac{\lambda}{\rhor} \,\,
        \frac{\gammav_{1}^{\T} \nabla^{2} \zeta(\thetav) \gammav_{2}} 
        	 {\| \VPD \gammav_{1} \| \cdot \| \VPD \gammav_{2} \|}
	\biggr\} 
    \le 
    \frac{\nunu^{2} \lambda^{2}}{2} \, ,
    \qquad 
    |\lambda| \leq \gm(\rr).
\label{gUUgem2}
\end{EQA}
Below we only need that  the constant \( \gm(\rr) \) 
is larger than \( \CONST \, \dimG \) for a fixed constant \( \CONST \).
This allows to reduce the condition to the case with a fixed \( \gm \) which 
does not depend on the distance \( \rr \). 
\end{description}

Their non-penalized versions are almost identical: one has to replace \( \thetavsGP \)
with \( \thetavs \) and \( \ThetasGP(\rr) \) with \( \Thetas(\rr) \).
\begin{description}
	\item[\label{nED0ref}\( \bb{(E_{0})} \)]
\( \qquad
	\sup\limits_{\gammav \in \R^{\dimp}} 
    \log \E \exp\biggl\{
        \lambda \dfrac{\gammav^{\T} \nabla \zeta(\thetavsGP)}
                     {\| \VPc \gammav \|}
    \biggr\} 
    \leq 
    \dfrac{\nunu^{2} \lambda^{2}}{2}, \qquad 
    |\lambda| \le \gmb .
\)

	\item[\label{nED2ref}\( \bb{(E_{2})} \)]
\( \qquad
    \sup\limits_{\gammav_{1},\gammav_{2} \in \R^{\dimp} }
    \log \E \exp\biggl\{ 
    	\dfrac{\lambda}{\rhor} \,\,
        \dfrac{\gammav_{1}^{\T} \nabla^{2} \zeta(\thetav) \gammav_{2}} 
        	 {\| \VPD \gammav_{1} \| \cdot \| \VPD \gammav_{2} \|}
	\biggr\} 
    \leq
    \dfrac{\nunu^{2} \lambda^{2}}{2} \, ,
    \qquad 
    |\lambda| \leq \gm(\rr).
\)
\end{description}
The conditions \nameref{nED0ref} and \nameref{ED0Gref} are very similar while 
\nameref{ED2Gref} is restricted to the  vicinity \( \ThetasGP(\rr) \) which can be much smaller
than \( \Thetas(\rr) \).

\medskip
The \emph{identifiability condition} relates the matrices \( \VPc^{2} \) and \( \VPD^{2} \) 
and to \( \DPGP^{2} \).
\begin{description}
  \item[\label{AssIdGref}\( \bb{(\AssId\GP)} \)] 
      There is a constant 
      \( \fisGP > 0 \) such that 
\begin{EQA}
	\fisGP^{2} \DPGP^{2} \ge \VPc^{2} ,
	& \quad &
	\fisGP^{2} \DPGP^{2} \ge \VPD^{2} . 
\label{lamGPDPVPfis}
\end{EQA}
\end{description}

In the non-penalized case of 
\ifbook{\cite{SP2011}}{Chapter~\ref{Ch_FW}}, this condition reads as
\begin{description}
  \item[\label{nAssIdref}\( \bb{(\AssId)} \)] 
\( \qquad \fis^{2} \DPc^{2} \ge \VPc^{2} \) with 
\( \DPc^{2} = - \nabla^{2} \E L(\thetavs) \). 
\end{description}

\noindent
Therefore, the use of regularization helps to improve the identifiability in the 
regularized problem relative to the non-penalized one as 
\( \DPc^{2} \leq \DPGP^{2} \).

Finally, for \( \rr > \rupsGP \),  we need a global identification property which ensures that the deterministic component 
\( \E \LGP(\thetav,\thetavsGP) \) of the penalized log-likelihood is competitive with the
variation of the stochastic component.

\begin{description}
    \item[\label{LLGref}\( \bb{(\LL\GP)} \)]
    \textit{
    There exists \( \gmiGP(\rr) > 0 \) such that
    \( \rr \, \gmiGP(\rr) \to \infty \) as \( \rr \to \infty \) and}
\begin{EQA}
    \frac{2 \E \LGP(\thetavsGP) - 2 \E \LGP(\thetav)}{\| \VPD (\thetav - \thetavsGP) \|^{2}}
    & \ge &
    \gmiGP(\rr) ,
    \qquad
    \thetav \in \ThetasGP(\rr) .
\label{xxentrttGP}
\end{EQA}    
\end{description}

A non-penalized version reads as follows: for \( \rr \, \gmi(\rr) \to \infty \) 
as \( \rr \to \infty \)
\begin{description}
    \item[\label{nLLref}\( \bb{(\LL)} \)]
\( \qquad    \dfrac{2 \E L(\thetavs) - 2 \E L(\thetav)}{\| \VPD (\thetav - \thetavs) \|^{2}}
    \ge 
    \gmi(\rr) ,
    \qquad
    \thetav \in \Thetas(\rr) 
\).
\end{description}

Obviously \( \E \LGP(\thetav) \leq \E L(\thetav) \) yielding 
\( \gmiGP(\rr) \geq \gmi(\rr) \) in typical situations; 
therefore the \nameref{LLGref} is less restrictive than 
\nameref{nLLref}.

\ifbook{
We  briefly comment on examples for which the conditions can be easily verified.
\cite{SP2011}, Section 5.1, considered in details the i.i.d. case and presented some mild 
sufficient conditions on the parametric family which imply the above general conditions.
Another class of examples is built by
generalized linear models which includes the cases of Gaussian, Poissonian, binary, 
regression and exponential type models among others.
Condition \nameref{ED0Gref} requires some exponential moments of the observations 
(errors). 
Usually one only assumes some finite moments of the normalized increments of the likelihood function; cf. \cite{IH1981}, Chapter~2.
Our conditions \nameref{ED0Gref} and \nameref{ED2Gref} a bit more restrictive but it allows one to obtain some finite sample 
bounds. 
Note that majority of finite samples results are stated under gaussian or sub-gaussian 
stochastic errors. 
The sub-gaussian case is included in \nameref{ED0Gref} and \nameref{ED2Gref}  and it corresponds to 
\( \gm = \infty \) which slightly simplifies the formulation of the results. 
However, our results apply for sub-exponential errors with \( \gm < \infty \) as well. 
%
Condition \nameref{LL0Gref} only requires some regularity of the considered parametric 
family and is not restrictive.
Conditions \nameref{ED2Gref} with \( \gm(\rr) \equiv \gm > 0 \) and \nameref{LLGref} 
with \( \gmi(\rr) \equiv \gmi > 0 \)
are easy to verify if the parameter set 
\( \Theta \) is compact and the sample size \( \nsize \) is sufficiently large.
It suffices to check a usual identifiability condition that the value 
\( \E \LGP(\thetav,\thetavs) \) does not vanish for \( \thetav \neq \thetavs \).

The regression and generalized regression models are included as well; cf. 
\cite{Gh1999,Gh2000} or \cite{Kim2006}.
\cite{SP2011}, Section 5.2, argued that the \nameref{ED2Gref} is automatically fulfilled for 
generalized linear models, while \nameref{ED0Gref} requires that 
regression errors have to fulfill some 
exponential moments condition. 
If this condition is too restrictive and a more stable (robust) estimation procedure is 
desirable, one can apply the LAD-type contrast leading to median regression. 
\cite{SP2011}, Section 5.3, showed for the case of linear median regression that 
all the required conditions are fulfilled automatically if the sample size \( \nsize \) 
exceeds \( \CONST \dimp \) for a fixed constant \( \CONST \).
\cite{SpWe2012} applied this approach for local polynomial quantile regression.
\cite{zaitsev2013properties} applied the approach to the problem of regression with Gaussian 
process where the unknown parameters enter in the likelihood function in a rather complicated 
way. 
We conclude that the imposed conditions are quite general and can be verified
for many classical examples met in the statistical literature.
}{}

\Section{Concentration and a large deviation bound}
This section demonstrates that the use of the penalty term helps to strengthen the
concentration properties of the penalized quasi maximum likelihood estimator (qMLE) 
\( \tilde{\thetav}_{\GP} \).
Namely, we show that \( \tilde{\thetav}_{\GP} \) belongs with a dominating probability to a set
\( \ThetasGP(\rupsGP) \) which can be much smaller than a similar set 
from the non-penalized case; see Remark~\ref{RTLDGP}.
%
%

All our results involve a value \( \xx \).
We say that a generic random set \( \Omega(\xx) \) is of a \emph{dominating probability} if
\(    \P\bigl( \Omega(\xx) \bigr)
    \ge 
    1 - \CONST \ex^{-\xx}  
\)
for a fixed constant \( \CONST \) like 1 or 2.
%
We also use two growing functions \( \zq(\BBGP,\xx) \) and 
\( \zzQ(\xx) \) of the argument \( \xx \).
The functions \( \zq(\BBGP,\xx) \) already mentioned in \eqref{zzxxdefro} and 
\( \zzQ(\xx) \) are given analytically and only depend on the parameters of the model.
The function \( \zq(\BBGP,\xx) \) describes the quantiles of the norm of the 
normalized score vector \( \xivGP \); see \eqref{xivGP} below.
The formal definition is given in \eqref{zzxxppdBlroB}.
The function \( \zzQ(\xx) \) is related to the penalized entropy of the parameter space and it is
given by \eqref{zzxxgfin}.
In typical situations one can use the upper bounds
\( \zq^{2}(\BBGP,\xx) \leq \CONST (\dimG + \xx) \) and \( \zzQ^{2}(\xx) \leq \CONST (\dimG + \xx) \) 
for both functions.

\begin{theorem}
\label{TLDGP}
Suppose \nameref{ED0Gref}, \nameref{ED2Gref}, and
\nameref{AssIdGref}. 
Let \nameref{LLGref} hold with the function \( \gmiGP(\rr) \) satisfying for a fixed \( \rupsGP \)
\begin{EQA}
    \gmiGP(\rr) \, \rr
    & \ge &
    2 \bigl\{ 
    \zq(\BBGP,\xx) + \rdomegaGP(\rr,\xx) \bigr\} \, ,
    \quad
    \rr > \rupsGP,
\label{cgmibrrGP}
\end{EQA}
where \( \zq(\BBGP,\xx) \) is from \eqref{zzxxppdBlroB} and
\begin{EQA}[c]
    \rdomegaGP(\rr,\xx)
    \eqdef
    \nunu \, \fisGP \, \zzQ(\xx + \log(2 \rr/\rupsGP)) \, \rhor 
\label{ExceqrrrhdefGP}
\end{EQA}  
with the function \( \zzQ(\xx) \) given by \eqref{zzxxgfin}.
Then 
\begin{EQA}[c]
    \P\bigl( \tilde{\thetav}_{\GP} \not\in \ThetasGP(\rupsGP) \bigr)
    \le 
    3 \ex^{-\xx} .
\label{PCAxxglrrsGP}
\end{EQA}    
\end{theorem}

\begin{remark}
\label{RTLDGP}
%
This result helps to fix a proper \( \rupsGP \) ensuring \eqref{PCAxxglrrsGP}.
The concentration result applies if the lower bound \eqref{cgmibrrGP}
on the negative expectation of the penalized log-likelihood process holds.
Condition \eqref{cgmibrrGP} can be made more detailed by separating the region 
\( \ThetasGP(\rr) \) of moderate deviations in which the condition \nameref{LL0Gref} 
applies with \( \rddeltaGP(\rr) \) small and the remaining set 
\( \Theta \setminus \ThetasGP(\rr) \).
On the set \( \ThetasGP(\rr) \) one can use \( \gmiGP(\rr) \geq 1 - \rddeltaGP(\rr) \), 
that is, \( \gmiGP(\rr) \, \rr \approx \rr \).
In addition, the remainder \( \rdomegaGP(\rr,\xx) \) in the right hand-side of \eqref{cgmibrrGP} is proportional to \( \rhor \) and this value is typically small.
For instance, in the i.i.d. case it is of order \( \nsize^{-1/2} \).
Therefore, the condition \eqref{cgmibrrGP} together with \nameref{LL0Gref} requires that 
\( \rupsGP \) fulfills 
\( \rupsGP \ge 2 \zq(\BBGP,\xx) = 2 \bigl( \sqrt{\dimG} + \sqrt{2 \xx} \bigr) \). 
In the non-penalized case of \cite{SP2011}, a similar condition reads as
\( \rups^{2} \ge \CONST (\dimp + \xx) \), so the use of penalization helps to improve
the concentration properties of the penalized MLE.
We conclude that the use of penalization leads to weaker conditions and 
to a stronger concentration property.
The only problem is that the corresponding estimate \( \tilde{\thetav}_{\GP} \) 
concentrates around \( \thetavsGP \) instead of \( \thetavs \).
This can yield a bias effect; see Section~\ref{Sqrisk} below.
\end{remark}

\begin{proof}
By definition \( \sup_{\thetav \in \ThetasGP(\rupsGP)} \LGP(\thetav,\thetavsGP) \geq 0 \).
So, it suffices to check that 
\( \LGP(\thetav,\thetavsGP) < 0 \) for all \( \thetav \in \Theta \setminus \ThetasGP(\rupsGP) \).
The proof is based on the following bound:
for each \( \rr \)
\begin{EQA}
	\P \biggl( \sup_{\thetav \in \ThetasGP(\rr)}
		\bigl| 
			\zeta(\thetav,\thetavsGP)   
			- (\thetav - \thetavsGP)^{\T} \nabla \zeta(\thetavsGP)
		\bigr|
		\geq 
		\nunu \, \fisGP \, \zzQ(\xx) \, \rhor \, \rr
	\biggr)
	& \leq &
	\ex^{-\xx} .
\label{PsuptsrrrhorGP}
\end{EQA}
This bound is a special case of the general result from Theorem~\ref{Texpro}.
It implies by Theorem~\ref{TsuprUP} with \( \rho = 1/2 \) on a set of dominating probability at least \( 1 - \ex^{-\xx} \) that for all \( \rr \geq \rupsGP \)
and all \( \thetav \) with \( \| \DPGP (\thetav - \thetavsGP) \| \leq \rr \)
\begin{EQA}
	\bigl| 
		\zeta(\thetav,\thetavsGP) - (\thetav - \thetavsGP)^{\T} \nabla \zeta(\thetavsGP)
	\bigr|
	& \leq &
	\rdomegaGP(\rr,\xx) \, \rr,
\label{zetattsnaGP}
\end{EQA}
where 
\begin{EQA}
	\rdomegaGP(\rr,\xx)
	&=&
	\nunu \, \fisGP \, \zzQ\bigl( \xx + \log(2\rr/\rupsGP) \bigr) \, \rhor \, .
\label{rdomdefGP}
\end{EQA}
The use of 
\( \nabla \E \LGP(\thetavsGP) = 0 \) yields
\begin{EQA}
\label{LGPLaGPrr}
	\sup_{\thetav \in \ThetasGP(\rr)}
		\bigl| 
			\LGP(\thetav,\thetavsGP) - \E \LGP(\thetav,\thetavsGP)  
			- (\thetav - \thetavsGP)^{\T} \nabla \LGP(\thetavsGP)
		\bigr|
	& \leq &
	\rdomegaGP(\rr,\xx) \, \rr .
\end{EQA}
Also the vector 
\( \xivGP = \DPGP^{-1} \nabla \LGP(\thetavsGP) = \DPGP^{-1} \nabla \zeta(\thetavsGP) \) can be bounded with a dominating probability: 
by Theorem~\ref{LLbrevelocro} 
\( \P\bigl( \| \xivGP \| \geq \zq(\BBGP,\xx) \bigr) \leq 2 \ex^{-\xx} \).
We ignore here the negligible term \( \CONST \ex^{-\xxc} \).
The condition \( \| \xivGP \| \leq \zq(\BBGP,\xx) \) implies
for each \( \rr \geq \rupsGP \) 
\begin{EQA}
	&& \nquad
	\sup_{\thetav \in \ThetasGP(\rr)} 
		\bigl| (\thetav - \thetavsGP)^{\T} \nabla \LGP(\thetavsGP) \bigr|
	\\
	& \leq &
	\sup_{\thetav \in \ThetasGP(\rr)} 
		\| \DPGP (\thetav - \thetavsGP) \| \times \| \DPGP^{-1} \nabla \zeta(\thetavsGP) \|
	= 	
	\rr \| \xivGP \| 
	\leq 
	\zq(\BBGP,\xx) \, \rr.
\label{supDLGPsGP}
\end{EQA}
Condition \nameref{LLGref} implies 
\( - 2 \E \LGP(\thetav,\thetavsGP) \geq \rr^{2} \gmiGP(\rr) \) for each \( \thetav \) with 
\( \| \DPGP (\thetav - \thetavsGP) \| = \rr \).
We conclude that the condition  
\begin{EQA}
	\rr \gmiGP(\rr)
	& \geq &
	2 \bigl\{ \zq(\BBGP,\xx) + \rdomegaGP(\rr,\xx) \bigr\},
	\quad \rr > \rupsGP ,
\label{rr2gmirrGP}
\end{EQA}
ensure \( \LGP(\thetav,\thetavsGP) < 0 \) for all \( \thetav \not\in \ThetasGP(\rupsGP) \) with a dominating probability.
\end{proof}

\Section{Wilks and Fisher expansions}
This section collects the main results of the paper.
Let \( \thetavsGP \) be the point of concentration from \eqref{LsGPro} and let
\( \zeta(\thetav) = \LGP(\thetav) - \E \LGP(\thetav) = L(\thetav) - \E L(\thetav) \).
Define a random \( \dimp \)-vector
\begin{EQA}
    \xivGP
    & \eqdef &
    \DPGP^{-1} \nabla \zeta(\thetavsGP)
    =
    \DPGP^{-1} \bigl\{ \nabla L(\thetavsGP) - \GP^{2} \thetavsGP \bigr\} .
\label{xivGP}
\end{EQA}

\begin{theorem}
\label{TconflocroGP}
Suppose that \( \rupsGP \) is selected to ensure \eqref{cgmibrrGP}.
Suppose also that the conditions \nameref{ED0Gref}, \nameref{ED2Gref}, \nameref{AssIdGref} hold.
On a random set 
\( \Omega(\xx) \) of a dominating probability at least \( 1 - 4 \ex^{-\xx} \), it holds 
\begin{EQA}[c]
    \| \DPGP (\tilde{\thetav}_{\GP} - \thetavsGP) - \xivGP \|
	\leq 
    \ExcGP(\xx) ,
\label{ttusmxivFGP}
\end{EQA}
where \( \ExcGP(\xx) \) is given by 
\begin{EQA}[c]
    \ExcGP(\xx)
    \eqdef
    \bigl\{ \rddeltaGP(\rupsGP) 
    + \sqrt{8} \, \nunu \, \fisGP \, \zzQ(\xx) \, \rhor \bigr\}\, \rupsGP
\label{Exc8defGP}
\end{EQA}  
for \( \zzQ(\xx) \) given by \eqref{zzxxgfin}. 
\end{theorem}

The proof of this and the next result is based on a linear expansion of the gradient 
\( \nabla \LGP(\thetav) \) and will be given in Section~\ref{SproofsWilks}.

Now we present a result on the excess 
\( \LGP(\tilde{\thetav}_{\GP},\thetavsGP) 
= \LGP(\tilde{\thetav}_{\GP}) - \LGP(\thetavsGP) \).
The classical Wilks result claims that the twice excess is nearly \( \chi^{2}_{\dimp} \).
Our result describes the quality of its approximation by a quadratic form 
\( \| \xivGP \|^{2} \).

\begin{theorem}
\label{TWilks2rGP}
Suppose that \nameref{LL0Gref}, \nameref{ED0Gref}, and \nameref{ED2Gref} hold.
Suppose also that \( \rupsGP \) is selected to ensure \eqref{cgmibrrGP}.
On a random set 
\( \Omega(\xx) \) of a dominating probability at least \( 1 - 5 \ex^{-\xx} \), it holds with 
\( \ExcGP(\xx) \) from \eqref{Exc8defGP}
\begin{EQA}
\label{WilLLeGP}
    \bigl| 2 \LGP(\tilde{\thetav}_{\GP},\thetavsGP) - \| \xivGP \|^{2} \bigr|
    & \le &
    2 \rupsGP \, \ExcGP(\xx) + \ExcGP^{2}(\xx),
    \\
    \Bigl| 
    	\sqrt{ 2\LGP(\tilde{\thetav}_{\GP},\thetavsGP) } 
		- \| \xivGP \| 
	\Bigr|
    & \le &
    3 \, \ExcGP(\xx).
\label{sqWilLLGP}
\end{EQA}    
\end{theorem}

One can see that the Fisher expansion \eqref{ttusmxivFGP} and the square root Wilks expansion
\eqref{sqWilLLGP} require \( \ExcGP(\xx) \) small, while the standard Wilks expansion 
\eqref{WilLLeGP} is accurate if \( \rupsGP \, \ExcGP(\xx) \) is small. 
This makes some difference if the parameter dimension is large. 
Below we address this question for the important special case of an i.i.d. likelihood.

The classical Fisher and Wilks results include some statements about the limiting 
behavior of the vector \( \xivGP \) and of the quadratic form \( \| \xivGP \|^{2} \).
In the i.i.d. case, one can easily show that the vector \( \xivGP \) is asymptotically standard 
normal as \( \nsize \to \infty \); see Section~\ref{Siidro} below.
However, it is well known that the convergence of \( \| \xivGP \|^{2} \) 
to the \( \chi^{2} \)-distribution is quite slow even in the case of a fixed dimension 
\( \dimp \). 
For finite sample inference, we recommend to combine the approximations \eqref{ttusmxivFGP} to 
\eqref{sqWilLLGP} 
with any resampling technique which mimics the specific behavior of the quadratic form
\( \| \xivGP \|^{2} \); see \cite{SpZh2014}. 


\Section{Quadratic risk bound and modeling bias}
\label{Sqrisk}
This section demonstrates the applicability of the obtained general results to bounding 
the quadratic risk of estimation. 
For the penalized MLE \( \tilde{\thetav}_{\GP} \) of the parameter \( \thetav \), 
consider the quadratic loss of estimation
\( \| \QL \bigl( \tilde{\thetav}_{\GP} - \thetavs \bigr) \|^{2} \)
for a given non-negative symmetric matrix \( \QL \).
A special case includes the usual quadratic loss 
\( \bigl\| \tilde{\thetav}_{\GP} - \thetavs \bigr\|^{2} \).
Here the point \( \thetavs \in \Theta \) is a proxy for the true parameter value which 
describes the best parametric fit of the true measure \( \P \) by the family 
\( (\P_{\thetav}) \):
\begin{EQA}[c]
	\thetavs
	\eqdef
	\argmax_{\thetav \in \Theta} \E L(\thetav).
\end{EQA}
The use of penalization \( \| \GP \thetav \|^{2}/2 \) introduces some estimation bias:
the penalized MLE \( \tilde{\thetav}_{\GP} \) estimates \( \thetavsGP \) from 
\eqref{LsGPro} rather than \( \thetavs \). 
The value \( \| \QL (\thetavs - \thetavsGP) \|^{2} \) is called the modeling bias and it 
describes the modeling error caused by using the penalization.
The variance term \( \| \QL \bigl( \tilde{\thetav}_{\GP} - \thetavsGP \bigr) \|^{2} \) describes the 
error \emph{within the penalized model}, and it can be studied with the help of 
the Fisher expansion of Theorem~\ref{TconflocroGP}: 
\( \bigl\| \DPGP \bigl( \tilde{\thetav}_{\GP} - \thetavsGP \bigr) - \xivGP \bigr\| 
\le \ExcGP(\xx) \) on a set \( \Omega(\xx) \) of dominating probability
for \( \xivGP = \DPGP^{-1} \nabla \zeta (\thetavsGP) \).
This yields the following result on \( \Omega(\xx) \):
\begin{EQA}[c]
    \bigl\| 
     	\DPGP (\tilde{\thetav}_{\GP} - \thetavs - \biasGP) - \xivGP 
    \bigr\|
    \le 
    \ExcGP(\xx) 
\label{tvtQLl}
\end{EQA}
with the \emph{bias} \( \biasGP = \thetavsGP - \thetavs \).
For any positive symmetric \( \dimp \times \dimp \) matrix \( \QL \) satisfying 
\( \QL^{2} \leq \DPGP^{2} \), it implies the probability bound for the squared loss 
\begin{EQ}[rcl]
    \| \QL (\tilde{\thetav}_{\GP} - \thetavs) \|
    & = &
    	\| \QL \biasGP + \QL \DPGP^{-1} \xivGP \| \pm \ExcGP(\xx) .
\label{tvtQL2l}
\end{EQ}    
One can see that analysis of the quadratic risk of the penalized MLE \( \tilde{\thetav}_{\GP} \) 
can be reduced to the analysis of \( \| \QL \biasGP + \QL \DPGP^{-1} \xivGP \|^{2} \).
Now we consider an implication of this bound to   
the squared risk \( \E \| \QL (\tilde{\thetav}_{\GP} - \thetavs) \|^{2} \). 
The use of the identity \( \E \nabla \zeta (\thetavsGP) = 0 \) and 
\( \Var(\nabla \zeta (\thetavsGP)) \leq \VPc^{2} \) yields 
\begin{EQA}
    \E \| \QL \biasGP + \QL \DPGP^{-1} \xivGP \|^{2}
    &=&
    \| \QL \biasGP \|^{2} + \E \| \QL \DPGP^{-2} \nabla \zeta (\thetavsGP) \|^{2}
    \\
    &=&
    \| \QL \biasGP \|^{2} 
    + \tr \bigl( \QL \DPGP^{-2} \Var\bigl\{ \nabla \zeta (\thetavsGP) \bigr\} \DPGP^{-2} \QL  \bigr) 
    \\
    & \leq &
    \| \QL \biasGP \|^{2} 
    + \tr \bigl( \QL \DPGP^{-2} \VPc^{2} \DPGP^{-2} \QL  \bigr) .
\label{biasvarQL}
\end{EQA}   
Denote \( \VL_{\GP} \eqdef \tr \bigl( \QL \DPGP^{-2} \VPc^{2} \DPGP^{-2} \QL \bigr) \) and
\begin{EQA}
	\risktGP
    & \eqdef &
	\| \QL \biasGP \|^{2} + \VL_{\GP}
	=
	\| \QL \biasGP \|^{2} 
    + \tr \bigl( \QL \DPGP^{-2} \VPc^{2} \DPGP^{-2} \QL \bigr) .
\label{EbiasvarQL}
\end{EQA}

\begin{theorem}
\label{Tq2loss}
Let \nameref{ED0Gref}, \nameref{ED2Gref}, \nameref{LL0Gref},
\nameref{AssIdGref}, and \nameref{LLGref} hold.
If \( \QL^{2} \leq \DPGP^{2} \), then
it holds with \( \risktGP \) from \eqref{EbiasvarQL}
\begin{EQA}[c]
    \E \| \QL (\tilde{\thetav}_{\GP} - \thetavs) \|^{2} 
    \le 
    \bigl\{ \risktGP^{1/2} + \ExcGP^{*} \bigr\}^{2} ,
\label{Ecriskt2QL}
\end{EQA}
where 
\begin{EQA}
	\ExcGP^{*} 
	&=& 
	4 \Bigl\{ \rddeltaGP(\rupsGP) \, \rupsGP
    	+ 2 \, \nunu \, \fisGP \, \rupsGP \, (\QQq + \QQg/\gm + 4) \rhor 
	\Bigr\} .
\label{ExcGPsdef}
\end{EQA}
\end{theorem}

\begin{remark}

If the error term \( \ExcGP^{*} \) in \eqref{Ecriskt2QL} is relatively small, this result implies 
\( \E \| \QL (\tilde{\thetav}_{\GP} - \thetavs) \|^{2} \approx \risktGP 
= \| \DPGP \biasGP \|^{2} + \VL_{\GP} \).
This is the usual decomposition of the quadratic risk 
in term of the squared bias \( \| \QL (\thetavsGP - \thetavs) \|^{2} \) and 
the variance term \( \VL_{\GP} \).
%
The condition ``\( \| \QL \biasGP \|^{2} / \VL_{\GP} \) is small'' yields 
\( \risktGP \approx \VL_{\GP} \).
This condition can be naturally called the \emph{small modeling bias} (SMB) condition, 
often it is referred to as \emph{undersmoothing}. 
The bias-variance trade-off corresponds to the situation with 
\( \| \QL \biasGP \|^{2} \asymp \VL_{\GP} \).
\emph{Oversmoothing} means that the bias terms \( \| \QL \biasGP \|^{2} \) dominates. 
\end{remark}

\begin{remark}
As already mentioned, the result \eqref{Ecriskt2QL} is informative if the remainder \( \ExcGP^{*} \) is relatively small and can be ignored. 
For the special case \( \QL^{2} = \DPGP^{2} \), it holds \( \VL_{\GP} = \dimG \asymp \rupsGP^{2} \).
In the i.i.d. situation (see Section~\ref{Siidro} below) 
\begin{EQA}
	\rupsGP^{-1} \ExcGP^{*}
	& \leq &
	\CONST \sqrt{\dimG/\nsize}
\label{rupsGPoExcGP}
\end{EQA}
which yields a sharp risk bound 
\( \E \| \QL (\tilde{\thetav}_{\GP} - \thetavs) \|^{2} = \risktGP \bigl( 1 + o(1) \bigr) \) under ``\( \dimG/\nsize \) small''.
\end{remark}

\begin{remark}
The bias induced by penalization can be measured in terms of the value 
\( \| \GP \thetavs \|^{2} \).
To be more precise, consider the case with \( \QL^{2} = \DPc^{2} \), where 
\( \DPc^{2} = - \nabla^{2} \E L(\thetavs) \) is the non-penalized Fisher information matrix.
The definition of \( \thetavs \) and \( \thetavsGP \) implies 
\begin{EQA}
	\E L(\thetavs) - \| \GP \thetavs \|^{2}/2
	& \leq &
	\E L(\thetavsGP) - \| \GP \thetavsGP \|^{2}/2
	\leq
	\E L(\thetavsGP) .
\label{LGPtsLs}
\end{EQA}
Condition \nameref{LL0Gref} implies 
\( \E L(\thetavs) - \E L(\thetavsGP) \approx \| \DPc (\thetavs - \thetavsGP) \|^{2}/2 \)
and 
\begin{EQA}[c]
	\| \DPc (\thetavs - \thetavsGP) \|^{2}
	\leq 
	\| \GP \thetavs \|^{2} - \| \GP \thetavsGP \|^{2}
	\leq 
	\| \GP \thetavs \|^{2} .
\end{EQA}
So, if the true point is ``smooth'' in there sense that \( \| \GP \thetavs \|^{2} \) is small,
then the squared bias \( \| \DPc (\thetavs - \thetavsGP) \|^{2} \) caused by penalization 
is small as well.
\end{remark}

\begin{proof}
The Fisher expansion from Theorem~\ref{TconflocroGP} can be written as 
\begin{EQA}
	\P\Bigl( 
		\bigl\| \DPGP (\tilde{\thetav}_{\GP} - \thetavs) - \DPGP \biasGP - \xivGP \bigr\| 
		\geq 
		\ExcGP(\xx) 
	\Bigr)
	& \leq &
	4 \ex^{-\xx} .
\label{FisherGPxx}
\end{EQA}
The definition \eqref{Exc8defGP} of \( \ExcGP(\xx) \) and \eqref{ESES234ch} 
of Theorem~\ref{TUPUpsdch} imply
\begin{EQA}
	\E^{1/2} \bigl\| \DPGP (\tilde{\thetav}_{\GP} - \thetavs) - \DPGP \biasGP - \xivGP \bigr\|^{2}
	& \leq &
	4 \Bigl\{ \rddeltaGP(\rupsGP) \rupsGP
    + 2 \, \nunu \, \fisGP \, \rupsGP \, (\QQq + \QQg/\gm + 4) \rhor \Bigr\} .
\label{E12DPGPQQ4}
\end{EQA}
By the result follows by the triangle inequality
\begin{EQA}
	\E^{1/2} \bigl\| \DPGP (\tilde{\thetav}_{\GP} - \thetavs) \bigr\|^{2}
	& \leq &
	\E^{1/2} \bigl\| \DPGP (\tilde{\thetav}_{\GP} - \thetavs) - \DPGP \biasGP - \xivGP \bigr\|^{2}
	+ \E^{1/2} \bigl\| \DPGP \biasGP + \xivGP \bigr\|^{2} .
\label{E12GPE1212}
\end{EQA}
This yields the assertion of the theorem.
\end{proof}

\ifbook{\Section{Proofs of the Fisher and Wilks expansions}} 
{\Section{Proofs of the Fisher and Wilks expansions}}
\label{SproofsWilks}
This section presents the proofs of the main results and some additional statements
which can be of independent interest.
The principle step of the proof is a bound on the local linear approximation of the gradient
\( \nabla \LGP(\thetav) \). 
Below we study separately its stochastic and deterministic components coming from the decomposition \( L(\thetav) = \E L(\thetav) + \zeta(\thetav) \).
With \( \DPGP^{2} = - \nabla^{2} \E \LGP(\thetavsGP) \),
this leads to the decomposition
\begin{EQA}
	\rderr(\thetav,\thetavsGP)
	& \eqdef & 
	\DPGP^{-1} \bigl\{ \nabla \LGP(\thetav) - \nabla \LGP(\thetavsGP) \bigr\} 
	+ \DPGP \, (\thetav - \thetavsGP)
	\\
	&=&
	\DPGP^{-1} \bigl\{ \nabla \zeta(\thetav) - \nabla \zeta(\thetavsGP) \bigr\} 
	\\
	&&
	+ \, 
	\DPGP^{-1} \bigl\{ \nabla \E \LGP(\thetav) - \nabla \E \LGP(\thetavsGP) \bigr\} 
	+ \DPGP \, (\thetav - \thetavsGP) .
\end{EQA}
First we check the deterministic part.
For any \( \thetav \) with \( \| \DPGP (\thetav - \thetavsGP) \| \leq \rr \)
and any unit vector \( \uv \in \R^{\dimp} \), it holds
\begin{EQA}
	\uv^{\T} \E \rderr(\thetav,\thetavsGP) 
	& = & 
	\uv^{\T} \DPGP^{-1} \bigl\{ \nabla \E \LGP(\thetav) - \nabla \E \LGP(\thetavsGP)
    	+ \DPGP^{2} (\thetav - \thetavsGP) \bigr\}
	\\
	&=&
	\uv^{\T} \bigl\{ \Id_{\dimp} - \DPGP^{-1} \IFGP(\thetavd) \DPGP^{-1} \bigr\} \, 
    	\DPGP (\thetav - \thetavsGP) ,
\end{EQA}
where \( \thetavd = \thetavd(\uv) \) is a point on the line connecting 
\( \thetavsGP \) and \( \thetav \).
This implies by \nameref{LL0Gref}
\begin{EQA}[c]
	\bigl\| \E \rderr(\thetav,\thetavsGP) \bigr\|
	\leq 
    \| \Id_{\dimp} - \DPGP^{-1} \IFGP(\thetavd) \DPGP^{-1} \|_{\oper} \, 
    \rr
    \le 
	\rddeltaGP(\rr) \rr .
\label{EUPupsrrGP}
\end{EQA}
Now we study the stochastic part. 
Consider the vector process 
\begin{EQA}
	\UP(\thetav,\thetavsGP)
	& \eqdef &
	\DPGP^{-1} \bigl\{ 
		\nabla \zeta(\thetav) - \nabla \zeta(\thetavsGP)  
	\bigr\} .
\label{UPupsnmGP}
\end{EQA}
Further, define \( \ups = \VPD (\thetav - \thetavsGP) \) and introduce 
a vector process \( \UU(\ups) \) with
\begin{EQA}[c]
    \UU(\ups) 
    \eqdef 
    \VPD^{-1} 
    \bigl[ \nabla \zeta(\thetav) - \nabla \zeta(\thetavsGP) \bigr] .
\label{UUupsnmGP}
\end{EQA}    
It obviously holds 
\( \nabla \UU(\ups) = \VPD^{-1} \nabla^{2} \zeta(\thetav) \VPD^{-1} \). 
Moreover, for any \( \gammav_{1}, \gammav_{2} \in \R^{\dimp} \) with 
\( \| \gammav_{1} \| = \| \gammav_{2} \| = 1 \), condition \nameref{ED2Gref} implies
for \( | \lambda| \leq \gm(\rr) \)
\begin{EQA}
    \log \E \exp\biggl\{ 
    	\frac{\lambda}{\rhor} \gammav_{1}^{\T} \nabla \UU(\ups) \gammav_{2} 
	\biggr\} 
    &=&
    \log \E \exp\biggl\{ 
       \frac{\lambda}{\rhor} 
       \gammav_{1}^{\T} \VPD^{-1} \nabla^{2} \zeta(\thetav) \VPD^{-1} \gammav_{2}  
    \biggr\} 
    \le 
    \frac{\nunu^{2} \lambda^{2}}{2} .
    \qquad
\label{gUUemgGP}
\end{EQA}
Define
\( \Upss(\rr) \eqdef \{ \ups \colon \| \ups \| \leq \rr, \, \| \GVS \ups \| \leq \rr \} \)
for \( \GVS^{-2} = \fisGP^{-2} \DPGP^{-1} \VPD^{2} \DPGP^{-1} \).
Then
\begin{EQA}
	\sup_{\thetav \in \ThetasGP(\rr)} \| \UP(\thetav,\thetavsGP) \|
	& \leq &
	\sup_{\ups \in \Upss(\rr)} \| \AA \UU(\ups) \| 
\label{bouuvupsrupsdxGP}
\end{EQA}
for \( \AA = \fisGP^{-1} \DPGP^{-1} \VPD \).
Theorem~\ref{TexproA} yields
\begin{EQA}[c]
	\sup_{\ups \in \Upss(\rr)} \| \AA \UU(\ups) \|
	\leq 
	\sqrt{8} \, \nunu \, \zzQ(\xx) \, \fisGP \, \rhor \, \rr
\end{EQA}
on a set of a dominating probability at least \( 1 - \ex^{-\xx} \), where
the function \( \zzQ(\xx) \) is given by \eqref{zzxxgfin}.
%

Putting together the bounds \eqref{EUPupsrrGP} and \eqref{bouuvupsrupsdxGP} imply 
the following result.

\begin{theorem}
\label{TliapprLLGP}
Suppose that the matrix \( \IFGP(\thetav) \eqdef - \nabla^{2} \E \LGP(\thetav) \) 
fulfills the condition \nameref{LL0Gref} and 
let also \nameref{ED0Gref} and \nameref{ED2Gref} be fulfilled on \( \ThetasGP(\rr) \) 
for any fixed \( \rr \leq \rrb \).
Then
\begin{EQA}[c]
	\P\biggl\{ 
		\sup_{\thetav \in \ThetasGP(\rr)} \bigl\| \DPGP^{-1} \bigl\{ 
		\nabla \LGP(\thetav) - \nabla \LGP(\thetavsGP) 
	\bigr\}
    + \DPGP (\thetav - \thetavsGP) \bigr\|
    	\geq 
    	\ExcGP(\rr,\xx) 
    \biggr\} 
    \leq
    \ex^{-\xx},
\label{supupsUPdxxGP}
\end{EQA}
where 
\begin{EQA}[c]
    \ExcGP(\rr,\xx)
    \eqdef
    \bigl\{ \rddeltaGP(\rr) + \sqrt{8} \, \nunu \, \zzQ(\xx) \, \fisGP \, \rhor \bigr\} \rr .
\label{ExceqrrrhGP}
\end{EQA}  
\end{theorem}

The result of Theorem~\ref{TliapprLLGP} can be extended to the increments of the process 
\( \UP(\thetav) \):
on a random set of probability at least \( 1 - \ex^{-\xx} \), it holds for any 
\( \thetav,\thetavd \in \ThetasGP(\rr) \) and 
\( \rderr(\thetav,\thetavd) 
= \DPGP^{-1} \bigl\{ \nabla \LGP(\thetav) - \nabla \LGP(\thetavd) \bigr\}
+ \DPGP \, (\thetav - \thetavd) \)
\begin{EQA}
	\E \bigl[ \rderr(\thetav,\thetavd) \bigr]
	& \leq &
	\rddeltaGP(\rr) \, \| \DPGP (\thetav - \thetavd) \|
	\leq
	2 \rr \, \rddeltaGP(\rr),
	\\
	\bigl\| \rderr(\thetav,\thetavd) \bigr\|
	& \leq &
    2 \, \ExcGP(\rr,\xx) .
\label{supupsUPdGP}
\end{EQA}

Now we present the proof of Theorem~\ref{TconflocroGP} about the Fisher expansion for the qMLE 
\( \tilde{\thetav}_{\GP} \) defined by maximization of \( \LGP(\thetav) \).
Let \( \rupsGP \) be selected to ensure that 
\( \P\bigl\{ \tilde{\thetav}_{\GP} \not\in \ThetasGP(\rupsGP) \bigr\} \leq \ex^{-\xx} \).
Furthermore, the definition of \( \tilde{\thetav}_{\GP} \) yields 
\( \nabla \LGP(\tilde{\thetav}_{\GP}) = 0 \) and 
\begin{EQA}[c]
	\rderr(\tilde{\thetav}_{\GP},\thetavsGP)
	=
	- \DPGP^{-1} \nabla \LGP(\thetavsGP) + \DPGP (\tilde{\thetav}_{\GP} - \thetavsGP) .
\end{EQA}
By Theorem~\ref{TliapprLLGP}, it holds on a set of a dominating probability
\begin{EQA}[c]
	\| \DPGP (\tilde{\thetav}_{\GP} - \thetavsGP) - \xivGP \|
	\leq
	\ExcGP(\xx) 
\label{supupsUPdxxtGP}
\end{EQA}
as required.

\medskip

As the next step, we apply the obtained results to evaluate the quality of the Wilks expansion 
\( 2 \LGP(\tilde{\thetav},\thetavsGP) \approx \| \xivGP \|^{2} \).
%
%
For this we derive a uniform deviation bound on the error of a quadratic approximation 
\begin{EQA}[c]
	\alp(\thetav,\thetavd)
	\eqdef
	\LGP(\thetav) - \LGP(\thetavd) 
    - (\thetav - \thetavd)^{\T} \nabla \LGP(\thetavd) 
    + \frac{1}{2} \| \DPGP (\thetav - \thetavd) \|^{2} 
\label{zetatuuv}
\end{EQA}
in all \( \thetav, \thetavd \in \Thetas \), where \( \Thetas \) is some vicinity 
of a fixed point \( \thetavsGP \).
With \( \thetavd \) fixed, the gradient \( \nabla \alp(\thetav,\thetavd) 
\eqdef \frac{d}{d\thetav} \alp(\thetav, \thetavd) \) fulfills
\begin{EQA}[c]
	\nabla \alp(\thetav,\thetavd)
	=
	\nabla \LGP(\thetav) - \nabla \LGP(\thetavd) 
	+ \DPGP^{2} (\thetav - \thetavd)
    =
    \DPGP \, \rderr(\thetav,\thetavd) ;
\end{EQA}
cf. \eqref{UPupsnmGP}.
This implies 
\begin{EQA}[c]
    \alp(\thetav,\thetavd) 
    = 
    (\thetav - \thetavd)^{\T} \nabla \alp(\thetavc,\thetavd)  ,
\label{alptts}
\end{EQA}    
where \( \thetavc \) is a point on the line connecting \( \thetav \) and \( \thetavd \).
Further, 
\begin{EQA}[c]
	\bigl| \alp(\thetav,\thetavd) \bigr|
    = 
    \bigl| (\thetav - \thetavd)^{\T} \DPGP \DPGP^{-1} \nabla \alp(\thetavc,\thetavd) \bigr|
    \le
    \| \DPGP (\thetav - \thetavd) \| 
    \sup_{\thetavc \in \ThetasGP(\rr)} \bigl| \rderr(\thetavc,\thetavd) \bigr| \, ,
\end{EQA}
and one can apply \eqref{supupsUPdGP}.
This yields the following result.

\begin{theorem}
\label{TqapprbrGP}
Suppose \nameref{LL0Gref}, \nameref{ED0Gref}, and \nameref{ED2Gref}.
For each \( \rr \), it holds on a random set 
\( \Omega(\xx) \) of a dominating probability at least \( 1 - \ex^{-\xx} \), 
it holds with any \( \thetav, \thetavd \in \ThetasGP(\rr) \)
\begin{EQA}[rclcrcl]
        \frac{\bigl| \alp(\thetav,\thetavsGP) \bigr|}{\| \DPGP (\thetav - \thetavsGP) \|}
    & \le &
    \ExcGP(\rr,\xx) ,
    & \quad &
    \bigl| \alp(\thetav,\thetavsGP) \bigr|
    & \le &
    \rr \, \ExcGP(\rr,\xx) ,
\label{supalp12s}
    \\
        \frac{\bigl| \alp(\thetavsGP,\thetav) \bigr|}{\| \DPGP (\thetav - \thetavsGP) \|}
    & \le &
    2 \ExcGP(\rr,\xx) ,
    & \quad &
    \bigl| \alp(\thetavsGP,\thetav) \bigr|
    & \le &
    2 \rr \, \ExcGP(\rr,\xx) ,
\label{supalp12st}
    \\
        \frac{\bigl| \alp(\thetav,\thetavd) \bigr|}{\| \DPGP (\thetav - \thetavd) \|}
    & \le &
    2 \ExcGP(\rr,\xx) ,
    & \quad &
	\bigl| \alp(\thetav,\thetavd) \bigr|
    & \le &
    4 \rr \, \ExcGP(\rr,\xx) ,
\label{supalp12}
\end{EQA}    
where \( \ExcGP(\rr,\xx) \) is from \eqref{ExceqrrrhGP}.
\end{theorem}

The result of Theorem~\ref{TqapprbrGP} for the special case with 
\( \thetav = \thetavsGP \) and \( \thetavd = \tilde{\thetav}_{\GP} \) yields in view of 
\( \nabla \LGP(\tilde{\thetav}_{\GP}) = 0 \) for \( \rr = \rupsGP \) 
and \( \ExcGP(\xx) = \ExcGP(\rupsGP,\xx) \)
under the condition 
\( \tilde{\thetav}_{\GP} \in \ThetasGP(\rupsGP) \)
\begin{EQA}[c]
    \Bigl| 
        \LGP(\tilde{\thetav}_{\GP},\thetavsGP) 
        - \| \DPGP (\tilde{\thetav}_{\GP} - \thetavsGP) \|^{2} / 2
    \Bigr|
    =
    \bigl| \alp(\thetavsGP,\tilde{\thetav}_{\GP}) \bigr|
    \le 
    2 \rupsGP \, \ExcGP(\xx) .
\label{LLttsxiGP}
\end{EQA}    
Furthermore, with \( \thetav = \tilde{\thetav}_{\GP} \) and \( \thetavd = \thetavsGP \)
\begin{EQA}
    \Bigl| 
        \LGP(\tilde{\thetav}_{\GP},\thetavsGP) 
        - \xivGP^{\T} \DPGP (\tilde{\thetav}_{\GP} - \thetavsGP)
        + \| \DPGP (\tilde{\thetav}_{\GP} - \thetavsGP) \|^{2} / 2
    \Bigr|
    &=&
    \bigl| \alp(\tilde{\thetav}_{\GP},\thetavsGP) \bigr|
    \\
    & \leq &
    \rupsGP \, \ExcGP(\xx) 
\label{LLttsxiGP}
\end{EQA}    
which implies
\begin{EQA}
	\Bigl| 
        L(\tilde{\thetav}_{\GP},\thetavsGP) - \| \xivGP \|^{2} 
        + \| \DPGP (\tilde{\thetav}_{\GP} - \thetavsGP) - \xivGP \|^{2} 
    \Bigr|
    & \leq & 
    2 \rupsGP \, \ExcGP(\xx) .
\label{LLttsxiGP}
\end{EQA}
Now 
it follows by \eqref{supupsUPdxxtGP} that
\begin{EQA}[c]
    \bigl| L(\tilde{\thetav}_{\GP},\thetavsGP) - \| \xivGP \|^{2}/2 \bigr|
    \le 
    \rupsGP \, \ExcGP(\xx) 
    + \ExcGP^{2}(\xx) /2.
\label{LLttsxi2GP}
\end{EQA}    
%
The error term can be improved if the squared root of the excess is 
considered.
Indeed, if \( \tilde{\thetav}_{\GP} \in \ThetasGP(\rupsGP) \)
\begin{EQA}
    && \nquad 
    \Bigl| 
    	\bigl\{ 2\LGP(\tilde{\thetav}_{\GP},\thetavsGP) \bigr\}^{1/2} 
		- \| \DPGP (\tilde{\thetav}_{\GP} - \thetavsGP) \| 
	\Bigr|
    \leq
    \frac{\bigl| 2\LGP(\tilde{\thetav}_{\GP},\thetavsGP) 
            - \| \DPGP (\tilde{\thetav}_{\GP} - \thetavsGP) \|^{2} \bigr|}
         {\| \DPGP (\tilde{\thetav}_{\GP} - \thetavsGP) \|}
    \\
    & \le & 
    \frac{2 \bigl| \alp(\tilde{\thetav}_{\GP},\thetavsGP) \bigr|}
         {\| \DPGP (\tilde{\thetav}_{\GP} - \thetavsGP) \|}
    \le 
    \sup_{\thetav \in \ThetasGP(\rupsGP)} 
    \frac{2 \bigl| \alp(\thetav,\thetavsGP) \bigr|}{\| \DPGP (\thetav - \thetavsGP) \|}
    \le 
    2 \, \ExcGP(\xx).
\label{sqLLttustuGP}
\end{EQA}    
The Fisher expansion \eqref{supupsUPdxxtGP} allows to replace here the norm 
of the standardized error \( \DPGP (\tilde{\thetav}_{\GP} - \thetavsGP) \) with 
the norm of the normalized score \( \xivGP \).
This completes the proof of Theorem~\ref{TWilks2rGP}.



\section{Examples}
This section illustrates the general results for two particularly important cases 
of i.i.d. and generalized linear models.
The primary focus of the study is to compare the penalized and non-penalized cases
and to quantify the impact of penalization. 
\subsection{I.i.d. case}

The model with independent identically distributed (i.i.d.) observations 
is one of the most popular setups in statistical literature and 
in statistical applications. 
The essential and the most developed part of the statistical theory is designed 
for the i.i.d. modeling. 
Especially, the classical asymptotic parametric theory is almost complete including 
asymptotic root-n normality and efficiency of the MLE 
and Bayes estimators under rather mild assumptions; see e.g. Chapter 2 and 3 in 
\cite{IH1981}.
So, the i.i.d. model can naturally serve as a benchmark for any extension of the statistical 
theory: being applied to the i.i.d. setup, the new approach should lead to 
essentially the same conclusions as in the classical theory.
Similar reasons apply to the regression model and its extensions.
Below we try demonstrate that the proposed non-asymptotic viewpoint 
is able to reproduce the existing brilliant and well established results of the 
classical parametric theory. 
With some surprise, 
the majority of classical efficiency results can be easily derived from 
the obtained general non-asymptotic bounds.

\Section{Quasi MLE in an i.i.d. model}
\label{SqMLEiid}
The basic i.i.d. parametric model means that the observations 
\( \Yv = (Y_{1},\ldots,Y_{\nsize}) \) are independent identically distributed from a 
distribution \( P \) from a given parametric family 
\( (P_{\thetav}, \thetav \in \Theta) \) on the observation space \( \YY_{1} \). 
Each \( \thetav \in \Theta \) 
clearly yields the product data distribution \( \P_{\thetav} = P_{\thetav}^{\otimes \nsize} \) 
on the product space \( \YY = \YY_{1}^{\nsize} \).
This section illustrates how the obtained general results can be applied to this 
type of modeling under possible model misspecification.
Different types of misspecification can be considered. Each of the assumptions, namely, 
data independence, identical distribution, parametric form of the marginal 
distribution can be violated. 
To be specific, we assume the observations \( Y_{i} \) independent and 
identically distributed. 
However, we admit that
the distribution of each \( Y_{i} \) does not necessarily 
belong to the parametric family \( (P_{\thetav}) \).
The case of non-identically distributed observations can be done similarly at cost 
of more complicated notation.

In what follows the parametric family \( (P_{\thetav}) \) is supposed to be dominated 
by a measure \( \Pdom \), and each density \( p(y,\thetav) = dP_{\thetav}/d\Pdom(y) \) 
is two times continuously differentiable in \( \thetav \) for all \( y \). 
Denote \( \ell(y,\thetav) = \log p(y,\thetav) \).
The parametric assumption \( Y_{i} \sim P_{\thetavs} \in (P_{\thetav}) \) leads to 
the log-likelihood 
\begin{EQA}[c]
    L(\thetav)
    =
    \sum \ell(Y_{i},\thetav) ,
\label{Ltiid}
\end{EQA}   
where the summation is taken over \( i=1,\ldots,\nsize \).
The quasi MLE \( \tilde{\thetav} \) maximizes this sum over \( \thetav \in \Theta 
\):
\begin{EQA}[c]
    \tilde{\thetav}
    \eqdef
    \argmax_{\thetav \in \Theta} L(\thetav)
    =
    \argmax_{\thetav \in \Theta} \sum \ell(Y_{i},\thetav).
\label{tttiid}
\end{EQA}    
The target of estimation \( \thetavs \) maximizes the expectation of \( L(\thetav) 
\):
\begin{EQA}[c]
    \thetavs
    \eqdef
    \argmax_{\thetav \in \Theta} \E L(\thetav)
    =
    \argmax_{\thetav \in \Theta} \sum \E \ell(Y_{i},\thetav).
\label{tsiid}
\end{EQA}
Let \( \zeta_{i}(\thetav) \eqdef \ell(Y_{i},\thetav) - \E \ell(Y_{i},\thetav) \).
Then \( \zeta(\thetav) = \sum \zeta_{i}(\thetav) \).
The equation 
\( \E \nabla L(\thetavs) = 0 \) implies
\begin{EQA}[c]
    \nabla \zeta(\thetavs)
    =
    \sum \nabla \zeta_{i}(\thetavs) 
    =
    \sum \nabla \ell_{i}(\thetavs) .
\label{nztiid}
\end{EQA}    

\Section{Conditions in the i.i.d. case}
\label{Scondiid}
I.i.d. structure of the  \( Y_{i} \)'s allows for rewriting the conditions 
\nameref{nED0ref}, \nameref{nED2ref}, \nameref{nAssIdref}, \nameref{nLL0ref}, and \nameref{nLLref} 
in terms of the marginal distribution.
In the following conditions the index \( i \) runs from \( 1 \) to \( \nsize \).
\begin{description}


\item[\( \bb{(ed_{0})} \)\label{ed0ref}] 
    \emph{ There exists a positive symmetric matrix \( \vpc \),
    such that for all \( |\lambda| \le \gmbm \)     
    } 
\begin{EQA}[c]
\label{expzetaciid} 
    \sup_{\gammav \in \cc{S}^{\dimp}} 
    \log \E \exp\biggl\{ 
        \lambda \frac{\gammav^{\T} \nabla \zeta_{i}(\thetavs)}
                     {\| \vpc \gammav \|} 
    \biggr\} \le 
    \nunu^{2} \lambda^{2} / 2. 
\end{EQA}
\end{description}
A natural candidate on \( \vpc^{2} \) is given by the variance of the gradient
\( \nabla \ell(Y_{1},\thetavs) \), that is, 
\( \vpc^{2} = \Var \nabla \ell(Y_{1},\thetav) = \Var \nabla \zeta_{1}(\thetav) \).
Note that \nameref{ed0ref} is automatically fulfilled 
if the the model is correctly specified and \( P = P_{\thetavs} \) because 
\( E_{\thetavs} \exp\bigl\{ \ell(Y_{1},\thetav) - \ell(Y_{1},\thetavs) \bigr\} \equiv 1 \).

Next consider the local sets 
\begin{EQA}[c]
    \Thetas(\rr) 
    = 
    \{ \thetav: \| \vpc(\thetav - \thetavs) \| \le \rr/\nsize^{1/2} \} .
\label{Theta0riid}
\end{EQA}    
The local smoothness conditions \nameref{nED2ref} and \nameref{nLL0ref} 
require to specify the functions \( \rddelta(\rr) \) and \( \rdomega(\rr) \).
If the log-likelihood function \( \ell(y,\thetav) \) is sufficiently smooth in 
\( \thetav \), these functions can be selected proportional to \( \rr \).

\begin{description}
\item[\( \bb{(ed_{2})} \)\label{ed2ref}] 
    \emph{There exist a value \( \rhorb > 0 \) 
    and for each \( \rr > 0 \), a constant \( \gm(\rr) > 0 \) such that}
\begin{EQA}[c]
    \sup_{\gammav_{1},\gammav_{2} \in \R^{\dimp} }
    \log \E \exp\biggl\{ 
    	\frac{\lambda}{\rhorb} \,\,
        \frac{\gammav_{1}^{\T} \nabla^{2} \zeta_{i}(\thetav) \gammav_{2}} 
        	 {\| \vpc \gammav_{1} \| \cdot \| \vpc \gammav_{2} \|}
	\biggr\} 
    \le 
    \frac{\nunu^{2} \lambda^{2}}{2} \, ,
    \qquad 
    |\lambda| \leq \gm(\rr).
\label{expzetac0iid}
\end{EQA}
\end{description}

Further we restate the local regularity condition \nameref{nLL0ref} in terms of
the expected value 
\( \elli(\thetav) \eqdef \E \ell(Y_{i},\thetav) \) of each \( \ell(Y_{i},\thetav) \). 
We suppose that \( \elli(\thetav) \) is two times 
differentiable and define the matrix function 
\( \IFon(\thetav) \eqdef - \nabla^{2} \elli(\thetav) \). 

\begin{description}
    \item[\( \bb{(\ell_{0})} \)\label{ell0ref}]
    \textit{The function \( \elli(\thetav) \) is two times differentiable
    and the matrix function \( \IFon(\thetav) = - \nabla^{2} \E \ell(Y_{1},\thetav) \)
    fulfills with \( \IFonec \eqdef \IFon(\thetavs) \) for some constant \( \rddeltab \):}
\begin{EQA}[c]
\label{LmgfquadELiid}
    \sup_{\thetav \in \Thetas(\rr)}
    \bigl\|
		\IFonec^{-1/2} \IFon(\thetav) \, \IFonec^{-1/2} - \Id_{\dimp} 
    \bigr\|_{\oper}
    \le
    \frac{\rddeltab \, \rr}{\sqrt{n}}  .
\end{EQA}
   
\end{description}

In the regular parametric case with \( P \in (P_{\thetav}) \), the matrices 
\( \vpc^{2} \) and \( \IFonec \) coincide with the
Fisher information matrix \( \IFonec = \IFon(\thetavs) \) of the family \( (P_{\thetav}) \) at the 
point \( \thetavs \).

\medskip
The consistency result for \( \tilde{\thetav} \) requires certain growth of the value 
\( \elli(\thetavs,\thetav) = \elli(\thetavs) - \elli(\thetav) \) as 
\( \| \thetav - \thetavs \| \) grows.
The marginal version of the global condition \nameref{nLLref} reads as follows:
\begin{description}
    \item[\( \bb{(\elli)} \)\label{elliref}]
    \textit{
    There exists \( \gmi(\rr) > 0 \) 
    such that
    \( \rr \gmi(\rr) \) is non-decreasing and
	}
\begin{EQA}
    \frac{2 \elli(\thetavs,\thetav)}{\| \IFonec^{1/2} (\thetav - \thetavs) \|^{2}}
    & \ge &
    \gmi(\rr) ,
    \qquad
    \forall \rr \geq \rups, \, \, \thetav \in \Thetas(\rr) .
\label{xxentrttiid}
\end{EQA}    
\end{description}

\begin{remark}
If the parametric i.i.d. model is correct, then 
\begin{EQA}
	\elli(\thetavs,\thetav) 
	&=& 
	\kullb(\thetavs,\thetav) 
	=
	E_{\thetavs} \log \frac{dP_{\thetavs}}{dP_{\thetav}}(Y_{1}) 
\label{ellittskKLiid}
\end{EQA}
is the Kullback-Leibler divergence for the family \( (P_{\thetav}) \).
Condition \nameref{elliref} is fulfilled automatically if 
\( \elli(\thetavs,\thetav) > 0 \) for \( \thetav \neq \thetavs \) and 
\( \Theta \) is a compact set.
Then 
\begin{EQA}
	\inf_{\thetav \in \Theta} 
		\frac{\elli(\thetavs,\thetav)}{\| \IFonec^{1/2} (\thetav - \thetavs) \|^{2}}
    & \ge &
    \gmi > 0 .
\label{inftTtT}
\end{EQA}
Based on this remark, one can verify \nameref{elliref} with \( \gmi(\rr) \geq \gmi > 0 \)
for all \( \rr \).
\end{remark}

The \emph{identifiability condition} relates the matrices \( \vpc^{2} \) and \( \IFonec \).
\begin{description}
  \item[\( (\bb{\assId}) \)\label{assIdref}] 
      There is a constant 
      \( \fis > 0 \) such that 
\begin{EQA}
	\fis^{2} \IFonec 
	& \ge & 
	\vpc^{2} . 
\label{lamGPDPVPfis}
\end{EQA}
\end{description}

\begin{lemma}
\label{Lcondiid}
Let \( Y_{1},\ldots,Y_{\nsize} \) be i.i.d. 
Then \nameref{ed0ref}, \nameref{ed2ref}, \nameref{ell0ref}, \nameref{elliref},
and \nameref{assIdref} imply
\nameref{nED0ref}, \nameref{nED2ref}, \nameref{nLL0ref}, \nameref{nLLref}, \nameref{nAssIdref}, with 
\( \VPc^{2} = \nsize \vpc^{2} \),
\( \DPc^{2} = \nsize \IFonec \),
\( \rhor = \rhorb \nsize^{-1/2} \),
\( \rddelta(\rr) = \rddeltab \rr / \sqrt{\nsize} \),
\( \gmi(\rr) \) from \nameref{elliref}, and
the same constants \( \nunu \), \( \fis \),  
\( \gmb \eqdef \lambdam \sqrt{\nsize} \).
\end{lemma}

\begin{proof}
The identities 
\( \VPc^{2} = \nsize \vpc^{2} \),
\( \DPc^{2} = \nsize \IFonec \)
follow from the i.i.d. structure of the observations \( Y_{i} \).
We briefly comment on condition \nameref{nED0ref}.
The use once again the i.i.d. structure yields by \eqref{nztiid}
in view of \( \VPc^{2} = \nsize \vpc^{2} \)
\begin{EQA}[c]
    \log \E \exp\Bigl\{ 
        \lambda \frac{\gammav^{\T} \nabla \zeta(\thetavs)}{\| \VPc \gammav \|} 
    \Bigr\} 
    =
    \nsize \E \exp\Bigl\{ 
        \frac{\lambda}{\nsize^{1/2}} 
        \frac{\gammav^{\T} \nabla \zeta_{1}(\thetavs)}{\| \vp \gammav \|} 
    \Bigr\}
    \le 
    \nunu^{2} \lambda^{2}/2
\label{Elognsizenul}
\end{EQA}    
as long as \( \lambda \le \nsize^{1/2} \gmbm \le \gmb \).
Similarly one can check \nameref{nED2ref}.
The conditions \nameref{nLL0ref}, \nameref{nLLref}, and \nameref{nAssIdref} follow from 
\nameref{ell0ref} and \nameref{elliref}, and \nameref{assIdref} due to 
\( \DPc^{2} = \nsize \IFonec \) and 
\( \E L(\thetav) = \nsize \elli(\thetav) \).
\end{proof}

Below we specify the obtained general results to the i.i.d. setup.

\Section{Results in the non-penalized i.i.d. case}
Here we specify the general results of previous chapters to the i.i.d. case.
In particular, we explicitly state the large deviation bound and show that it yields 
a root-n consistency of the qMLE \( \tilde{\thetav} \).
Then we comment on the Fisher, Wilks, and the BvM theorems.

First we describe the large deviation probability for the event 
\( \{ \tilde{\thetav} \not\in \Thetas(\rups) \} \) for a fixed \( \rups \). 
The next result specifies the general large deviation statement of 
\ifbook{Theorem~\ref{TLDGP}}{Theorem~\ref{TMLE}} 
to the finite dimensional non-penalized i.i.d. case
and states the inference results.

\begin{theorem}
\label{TLDiid}
Suppose \nameref{ed0ref}, \nameref{ed2ref}, \nameref{ell0ref}, and \nameref{assIdref}. 
Let also \nameref{elliref} hold with the function  \( \gmi(\rr) \) satisfying 
\begin{EQA}
    \gmi(\rr) \, \rr
    & \ge &
    2 \zq(\BB,\xx) + 2 \rdomega(\rr,\xx), 
    \quad
    \rr > \rups,
\label{cgmibrriid}
\end{EQA}
where \( \BB = \IFonec^{-1/2} \, \vpc^{2} \, \IFonec^{-1/2} = \DPc^{-1} \VPc^{2} \DPc^{-1} \),
\( \zq(\BB,\xx) \) is given by \eqref{zzxxppdBlroB}, and 
\begin{EQA}[c]
    \rdomega(\rr,\xx)
    \eqdef
    \nunu \, \zzQ\bigl(\xx + \log(2\rr/\rups) \bigr) \, \rhorb / \sqrt{\nsize}  
\label{Exceqrrrhiid}
\end{EQA}  
with \( \zzQ(\xx) \leq \CONST \sqrt{\dimp + \xx} \).
Then it holds on a set \( \Omega(\xx) \) with \( \P\bigl( \Omega(\xx) \bigr) \geq 1 - 5 \ex^{-\xx} \)
\begin{EQA}[c]
    \sqrt{\nsize} \| \IFonec^{1/2} (\tilde{\thetav} - \thetavs) \| 
    \leq 
    \rups .
\label{PLDboundiid}
\end{EQA}    
Furthermore, on this set \( \Omega(\xx) \), it holds 
\begin{EQA}
\label{Wilksiid}
    \bigl\|
        \sqrt{\nsize \IFonec} \bigl( \tilde{\thetav} - \thetavs \bigr)
        - \xiv
    \bigr\|
    & \le &
    \CONST \sqrt{(\dimp + \xx)^{2} / \nsize} \, ,
    \\
    \Bigl| \sqrt{2 L(\tilde{\thetav},\thetavs)} - \| \xiv \| \Bigr|
    & \le &
    \CONST \sqrt{(\dimp + \xx)^{2} / \nsize},
    \\
    \Bigl| 2 L(\tilde{\thetav},\thetavs) - \| \xiv \|^{2} \Bigr|
    & \le &
    \CONST \sqrt{(\dimp + \xx)^{3} / \nsize} .
\label{DPtttiid}
\end{EQA}
The constant \( \CONST \) here depends in an explicit way on the constants 
\( \fisGP \), \( \gmbm \), and \( \nunu \) from our conditions, and
\begin{EQA}
	\xiv
	& \eqdef &
	(\nsize \IFonec)^{-1/2} \sum_{i=1}^{\nsize} \nabla \ell(Y_{i},\thetavs) .
\label{xiviid}
\end{EQA}
\end{theorem}


\begin{proof}
Condition \nameref{assIdref} implies 
\( \BB = \IFonec^{-1/2} \, \vpc^{2} \, \IFonec^{-1/2} \leq \fis^{2} \Id_{\dimp} \) and thus,
\( \tr(\BB) \leq \fis^{2} \dimp \).
Therefore, the value \( \zq(\BB,\xx) \) fulfills
\( \zq^{2}(\BB,\xx) \leq \CONST (\dimp + \xx) \).
The same bound holds for \( \zzQ^{2}(\xx) \).
Condition \eqref{cgmibrriid} with \( \gmi(\rups) \approx 1 \) yields
\( \rups^{2} \approx 4 \zq^{2}(\BB,\xx) \approx \CONST (\dimp + \xx) \).
%
This yields in view of \( \rddelta(\rups) \leq \rddeltab \rups / \sqrt{\nsize} \) and
\( \rhor = \rhorb \nsize^{-1/2} \)
\begin{EQA}[c]
    \Excgr(\rups,\xx)
    \leq
    \bigl\{ \rddelta(\rups) + \nunu \, \zzQ(\xx) \, \rhor \bigr\} \rups
    \leq 
    \CONST (\dimp + \xx) /  \sqrt{\nsize} .
\label{taurddimpn}
\end{EQA}
Similarly
\begin{EQA}[c]
	\spread(\rups,\xx)
    \leq 
    \bigl\{ \rddelta(\rups) + \nunu \, \zzQ(\xx) \, \rhor \bigr\} \rups^{2} 
    \leq
    \CONST \sqrt{(\dimp+\xx)^{3} / \nsize} .
\label{spreadasm}
\end{EQA}
The results follow now from general theorems of Section~\ref{Srough}.
\end{proof}



%

For the classical asymptotic setup when \( \nsize \) tends to infinity,
the random vector \( \xiv \) from \eqref{xiviid} fulfills
\( \Var (\xiv) \le \IFonec^{-1/2} \, \vpc^{2} \, \IFonec^{-1/2} = \BB \) 
and by the central limit theorem \( \xiv \) is asymptotically normal 
\( \ND(0,\BB) \).
This yields by Theorem~\ref{TLDiid} that
\( \sqrt{\nsize \IFonec} \bigl( \tilde{\thetav} - \thetavs \bigr) \) is asymptotically 
normal \( \ND(0,\BB) \) as well.
The correct model specification implies \( \BB \equiv \Id_{\dimp} \) and hence 
\( \tilde{\thetav} \) is asymptotically efficient; see \cite{IH1981}.
Also \( 2 L(\tilde{\thetav},\thetavs) \approx \| \xiv \|^{2} \) which is nearly 
\( \chi^{2} \) r.v. with \( \dimp \) degrees of freedom.
This result is known as asymptotic Wilks theorem.

In the non-asymptotic framework of this paper, 
the error terms still depend on \( \nsize \) and they can only be small if \( \nsize \) is large.
However, we show in explicit way how these error terms depend on the parameter dimension. 
It appears that the root-n consistency result \eqref{PLDboundiid} requires 
``\( \dimp/\nsize \) small''.
The Fisher and square root Wilks results apply if ``\( \dimp^{2}/\nsize \) is small''.
Finally, the Wilks expansion is valid under ``\( \dimp^{3}/\nsize \) small''.
Existing statistical literature addresses the issue of a growing parameter dimension 
in different set-ups. 
The classical results by \cite{Portnoy1984,Portnoy1985,Portnoy1986} 
provide some constraints on parameter dimension for consistency and asymptotic normality 
of the M-estimator for regression models. 
Our results are consistent with the conclusion of that papers. 
We refer to \cite{AASP2012} for a version of such result in context of semiparametric 
profile estimation.
That paper also provides an example of an i.i.d. model in which the Fisher expansion of Theorem~\ref{TLDiid} fails for \( \dimp^{2} \geq \nsize \).
The next section demonstrates how these constraints on the parameter dimension can be relaxed 
by using a penalization.

\Section{Roughness penalization for an i.i.d. sample}
\label{Siidro}
This section discusses the impact of penalization in the case of an i.i.d. model with \( \nsize \) observations.
For penalty term \( \penr(\thetav) = \| \GP \thetav \|^{2}/2 \), 
the penalized log-likelihood is given by \( \LGP(\thetav) = L(\thetav) + \| \GP \thetav \|^{2}/2 \),
where \( L(\thetav) \) is from \eqref{Ltiid}.
With \( \thetavsGP = \argmax_{\thetav \in \Theta} \E \LGP(\thetav) \), define 
\begin{EQA}[c]
    \DPGP^{2}
    = 
    \nsize \IFon(\thetavsGP) + \GP^{2},
    \quad 
    \VPc^{2}
    =
    \nsize \, \vpc^{2},
    \quad
    \xivGP
    =
    \DPGP^{-1} \sum_{i=1}^{\nsize} \nabla \ell(Y_{i},\thetavsGP) ,
\label{DPcVpciid}
\end{EQA}    
where \( \IFon(\thetav) = - \nabla^{2} \E \ell(Y_{1},\thetav) \),
\( \vpc^{2} = \Var\bigl\{ \ell(Y_{1},\thetavsGP) \bigr\} \).
The value \( \dimG \) is defined as previously by \eqref{dimedef}.

Note that all the introduced quantities including the parameter set \( \Theta \), 
the parameter dimension \( \dimp \), and the effective dimension \( \dimG \),
may depend on \( \nsize \).
Here we also allow a functional parameter \( \thetav \) with \( \dimp = \infty \).
The main goal is to show that the presented general approach yields sharp results 
in this special case. 

Suppose that the conditions of Section~\ref{Scondiid} are fulfilled. 
One can easily check the conditions from Section~\ref{ScondroGP} with 
\( \rddeltaGP(\rr) = \CONST \, \rr / \sqrt{\nsize} \) and 
\( \rhor = \CONST / \sqrt{\nsize} \);
cf. Lemma~\ref{Lcondiid}. 
The large deviation bound of Theorem~\ref{TLDGP} applies 
for \( \rupsGP \approx 2 \zq(\BBGP,\xx) \asymp \sqrt{\dimG + \xx} \).
The general statements of 
Theorems~\ref{TconflocroGP} and \ref{TWilks2rGP} 
apply with \( \ExcGP(\xx) \leq \CONST (\dimG + \xx) / \sqrt{\nsize} \)
yielding the following expansions.

\begin{theorem}
\label{TWilksiid}
Suppose also that the conditions \nameref{ed0ref}, \nameref{ed2ref}, \nameref{ell0ref}, 
\nameref{elliref},
and \nameref{assIdref} are fulfilled.
If \( \gmi(\rr) \) fulfills
\begin{EQA}
    \gmi(\rr) \, \rr
    & \ge &
    2 \zq(\BBGP,\xx) + 2 \rdomega(\rr,\xx), 
    \quad
    \rr > \rups,
\label{cgmibrrGPiid}
\end{EQA}
with \( \BBGP = \DPGP^{-1} \, \VPc^{2} \, \DPGP^{-1/2} \), 
then on a set of dominating probability \( 1 - 5 \ex^{-\xx} \), it holds
\begin{EQA}
\label{Wilksroiid}
    \bigl\|
        \DPGP \bigl( \tilde{\thetav}_{\GP} - \thetavsGP \bigr)
        - \xivGP 
    \bigr\|
    & \le &
    \CONST \sqrt{(\dimG + \xx)^{2} / \nsize} \, ,
    \\
    \Bigl| \sqrt{2 \LGP(\tilde{\thetav}_{\GP},\thetavsGP)} - \| \xivGP \| \Bigr|
    & \le &
    \CONST \sqrt{(\dimG + \xx)^{2} / \nsize},
    \\
    \Bigl| 2 \LGP(\tilde{\thetav}_{\GP},\thetavsGP) - \| \xivGP \|^{2} \Bigr|
    & \le &
    \CONST \sqrt{(\dimG + \xx)^{3} / \nsize} .
\label{DPbtttGPiid}
\end{EQA}
The constant \( \CONST \) here depends in an explicit way on the constants 
\( \fisGP \), \( \gmbm \), and \( \nunu \) from our conditions.
\end{theorem}

A short look at the results for non-penalized and penalized estimates 
indicates that the quality of the penalized MLE \( \tilde{\thetav}_{\GP} \) 
improves relative to the non-penalized case because the matrix \( \DPGP^{2} \) 
can be much larger than \( \DPc^{2} \),
the variance of the stochastic term \( \xivGP \) is of order \( \dimG \) instead of 
\( \dimp \) for the variance of \( \xiv \),
and, simultaneously, the error terms in the Fisher and Wilks expansions become smaller
due to reduction of the effective dimension \( \dimG \) in place of 
the full dimension \( \dimp \).

\ifbook{}
{
\Section{BvM Theorem for the i.i.d. data}

Another constraint in the BvM Theorem on the dimension growth \( \dimn \) can be found in \cite{Gh1999} 
for linear regression models; see the condition (2.6) 
\( \dimn^{3/2} (\log \dimn)^{1/2} \, \eta_{\nsize} \to 0 \) there, 
in  which \( \eta_{\nsize} \) is of order \( (\dimn/\nsize)^{-1/2} \) in regular situations 
yielding a suboptimal constraint \( \nsize^{-1} \dimn^{4} \log \dimp \to 0 \).
\cite{Gh2000} obtained a version of the BvM result under the condition 
\( \nsize^{-1} \dimn^{3} (\log \dimn) \to 0 \) for a class of exponential models. 
A forthcoming paper \cite{PaSp2013} presents an example illustrating that the condition 
\( \dimp_{\nsize}^{3}/\nsize \to 0 \) cannot be dropped or relaxed.

The setup with growing parameter dimension is naturally used in sieve nonparametric estimation
when a nonparametric model is approximated by a sequence of parametric ones. 
We mention papers by 
\cite{ShWo1994, shen1997}, 
\cite{BiMa1993}, 
\cite{vdG1993,vdG2002}.
Some minimal smoothness assumptions are normally imposed on the underlying nonparametric function 
which ensure that the parameter dimension of a sieve is smaller in order than the sample size.
}

\Section{Generalized linear models (GLM)}
\label{SFWGLM}
Generalized linear models (GLM) are frequently used for modeling 
the data with special structure:
categorical data, binary data, 
Poissonian and exponential data, volatility models, etc.
All these examples can be treated in a unified way by a GLM approach.
This section specifies the results and conditions to this case.

Let \( \Yv = (Y_{1},\ldots, Y_{\nsize})^{\T} \sim \P \) be a sample of independent r.v.s.
The parametric GLM model is given by 
\( Y_{i} \sim P_{\Psi_{i}^{\T} \thetav} \in (P_{\upsi}) \),
where \( \Psi_{i} \) are given factors in \( \R^{\dimp} \),
\( \thetav \in \R^{\dimp} \) is the unknown parameter in \( \R^{\dimp} \), and
\( (P_{\upsi}) \) is an exponential family with canonical parametrization yielding the log-density
\( \ell(y,\upsi) = y \upsi - \GLMlink(\upsi) \) for a convex function \( \GLMlink(\upsi) \).
Below we suppose that the function \( \GLMlink(\upsi) \) is sufficiently smooth,
in particular, three times differentiable.

The (quasi) log-likelihood \( L(\thetav) \) can be represented in the form
\begin{EQA}
	L(\thetav)
	&=&
	\sum_{i=1}^{\nsize} 
		\bigl\{ Y_{i} \Psi_{i}^{\T} \thetav - \GLMlink(\Psi_{i}^{\T} \thetav) \bigr\}
	=
	\nablaGLM^{\T} \thetav - \GLMLINK(\thetav)
\label{LtGLM}
\end{EQA}
with a random \( \dimp \)-vector 
\begin{EQA}
	\nablaGLM
	& \eqdef &
	\sum_{i=1}^{\nsize} Y_{i} \Psi_{i} 
\label{nablGLMdef}
\end{EQA}
and a function 
\begin{EQA}
	\GLMLINK(\thetav)
	& \eqdef &
	\sum_{i} \GLMlink(\Psi_{i}^{\T} \thetav) .
\label{LINKGLMdef}
\end{EQA}
The MLE \( \tilde{\thetav} \) and the target \( \thetavs \) for this GLM read as
\begin{EQA}
	\tilde{\thetav}
	&=&
	\argmax_{\thetav} L(\thetav)
	=
	\argmax_{\thetav} \bigl\{ \nablaGLM^{\T} \thetav - \GLMLINK(\thetav) \bigr\} ,
	\\
	\thetavs
	&=&
	\argmax_{\thetav} \E L(\thetav)
	=
	\argmax_{\thetav} \bigl\{ \E \nablaGLM^{\T} \thetav - \GLMLINK(\thetav) \bigr\} ,
\label{ttstsGLM}
\end{EQA}
where
\begin{EQA}
	\E \nablaGLM
	&=&
	\sum_{i=1}^{\nsize} \E Y_{i} \, \Psi_{i} \, .
\label{EnablGLM}
\end{EQA}
The definition of \( \thetavs \) implies the identity \( \nabla \E L(\thetavs) = 0 \)
which yields
\begin{EQA}
	\E \nablaGLM
	&=&
	\nabla \GLMLINK(\thetavs) .
\label{EsnablGLM}
\end{EQA}
An important feature of a GLM is that the stochastic component \( \zeta(\thetav) \) of 
\( L(\thetav) \) is \emph{linear in} \( \thetav \): 
with \( \varepsilon_{i} = Y_{i} - \E Y_{i} \)
\begin{EQA}
	\zeta(\thetav)
	&=&
	L(\thetav) - \E L(\thetav)
	=
	\sum_{i=1}^{\nsize} \varepsilon_{i} \Psi_{i}^{\T} \thetav ,
	\\
	\nabla \zeta(\thetav)
	&=&
	\nablaGLM - \E \nablaGLM
	= 
	\sum_{i=1}^{\nsize} \varepsilon_{i} \Psi_{i}.
\label{nablaGLMdef}
\end{EQA}
%
In the contrary to the linear case, 
the Fisher information matrix \( \DPc^{2} = \IF(\thetavs) \) for
\begin{EQA}[c]
	\IF(\thetav)
	\eqdef 
	- \nabla^{2} \E L(\thetav)
	=
	\sum_{i=1}^{\nsize} \Psi_{i} \Psi_{i}^{\T} \GLMlink''(\Psi_{i}^{\T} \thetav) 
\label{IFGLMdef}
\end{EQA}
depends on the true data distribution via the target \( \thetavs \).
As \( \GLMlink(\cdot) \) is convex, it holds \( \GLMlink(u) \geq 0 \) for any \( u \) and thus
\( \IF(\thetav) \geq 0 \).

Linearity in \( \thetav \) of the stochastic component \( \zeta(\thetav) \) and 
concavity of the deterministic part \( \E L(\thetav) \) allow for
a simple and straightforward proof of the result
about concentration of the MLE \( \tilde{\thetav} \).
Recall the definition of the local vicinity \( \Thetas(\rr) \) of \( \thetavs \):
\begin{EQA}
	\Thetas(\rr)
	& \eqdef &
	\bigl\{ \thetav \colon \| \DPc (\thetav - \thetavs) \| \leq \rr \bigr\}.
\label{ThrGLM}
\end{EQA}

\begin{theorem}
\label{TGLMsolution}
If for some \( \rups > 0 \), \( \IF(\thetav) \) from \eqref{IFGLMdef} fulfill 
for \( \DPc^{2} = \IF(\thetavs) \)
\begin{EQA}
	\sup_{\thetav \in \Thetas(\rups)} 
		\| \DPc^{-1} \, \IF(\thetav) \, \DPc^{-1} - \Id_{\dimp} \|_{\oper}
	& \leq &
	\rddelta(\rups) 
\label{rddeGLM}
\end{EQA}
with \( \rddelta(\rups) < 1 \), 
and if \( \nablaGLM \) from \eqref{nablaGLMdef} follows for \( \xx > 0 \) the probability bound
\begin{EQA}
	\P\Bigl( 
		\| \DPc^{-1} (\nablaGLM - \E \nablaGLM) \| > \frac{1 - \rddelta(\rups)}{2} \rups  
	\Bigr)
	& \leq &
	2 \ex^{-\xx} ,
\label{PxizxGLM}
\end{EQA}
then the solution \( \tilde{\thetav} \) of \eqref{ttstsGLM} satisfies
\begin{EQA}
	\P\bigl( \tilde{\thetav} \not\in \Thetas(\rups) \bigr)
	& \leq &
	2 \ex^{-\xx} .
\label{PttsTsGLM}
\end{EQA}
\end{theorem}

\begin{proof}
The function \( L(\thetav) \) is concave in \( \thetav \)
because 
\begin{EQA}
	- \nabla^{2} L(\thetav)
	&=&
	\IF(\thetav) 
	\geq 
	0 .
\label{nab2LthIF}
\end{EQA}
If \( \tilde{\thetav} \not\in \Thetas(\rups) \), denote by
\( \check{\thetav} \) the point at which the line connecting \( \thetavs \) and 
\( \tilde{\thetav} \) crosses the boundary of \( \Thetas(\rups) \).
It is easy to see that
\begin{EQA}
	\check{\thetav} - \thetavs
	&=&
	\frac{\| \DPc (\check{\thetav} - \thetavs) \|}{\| \DPc (\tilde{\thetav} - \thetavs) \|}
		\bigl( \tilde{\thetav} - \thetavs \bigr) 
	=
	\frac{\rups}{\| \DPc (\tilde{\thetav} - \thetavs) \|} \,
		\bigl( \tilde{\thetav} - \thetavs \bigr) .
\label{chttsmGLM}
\end{EQA}
Concavity of \( L(\thetav) \) implies for the point of maximum \( \tilde{\thetav} \) that 
\begin{EQA}
	L(\tilde{\thetav}) - L(\thetavs)
	& \geq &
	L(\check{\thetav}) - L(\thetavs) .
\label{LttLtsGLM}
\end{EQA}
Therefore, it suffices to check that for each \( \thetav \) with 
\( \| \DPc (\thetav - \thetavs) \| = \rups \) that
\begin{EQA}
	L(\thetavs) - L(\thetav)
	& > &
	0
\label{Ltr0sLtGLM}
\end{EQA}
on a set \( \Omega(\xx) \) of probability \( 1 - 2 \ex^{-\xx} \).
Then the event \( \tilde{\thetav} \not\in \Thetas(\rups) \) is impossible on \( \Omega(\xx) \).
For any such \( \thetav \),
we apply the second order Taylor expansion of \( L(\thetav) \) at \( \thetavs \).
By definition of \( \thetavs \), it holds
\( \nabla \E L(\thetavs) = 0 \) and thus
\( \nabla L(\thetavs) = \nabla \zeta(\thetavs) = (\nablaGLM - \E \nablaGLM) \).
The use of \eqref{nab2LthIF}, \eqref{rddeGLM} yields now for 
\( \xiv = \DPc^{-1} (\nablaGLM - \E \nablaGLM) \) and for \( \thetav \) with 
\( \| \DPc (\thetav - \thetavs) \| = \rups \)
\begin{EQA}
	L(\thetavs) - L(\thetav)
	&=&
	(\thetav - \thetavs)^{\T} \nabla L(\thetavs)
	+ \frac{1}{2} \bigl\| \sqrt{\IF(\thetavd)} (\thetav - \thetavs) \bigr\|^{2}
	\\
	& \geq &
	(\nablaGLM - \E \nablaGLM)^{\T} (\thetav - \thetavs)
	+ \frac{1 - \rddelta(\rups)}{2} \| \DPc (\thetav - \thetavs) \|^{2}
	\\
	&=&
	\xiv^{\T} \DPc (\thetav - \thetavs)
	+ \frac{1 - \rddelta(\rups)}{2} \rups^{2} 
	\\
	& \geq &
	- \| \xiv \| \, \rups
	+ \frac{1 - \rddelta(\rups)}{2} \rups^{2}.
\label{LtsLtTaGLM}
\end{EQA}
Here \( \thetavd \) is a point from \( \Omega(\xx) \) on the interval connecting 
\( \thetav \) and \( \thetavs \).
If \( \| \xiv \| \leq \rups \bigl\{ 1 - \rddelta(\rups) \bigr\}/2 \), then this implies
\( L(\thetavs) - L(\thetav) > 0 \), and the result follows. 
\end{proof}

As a corollary, we obtain Fisher and Wilks expansions for the quasi MLE \( \tilde{\thetav} \)
in a generalized linear model.

\begin{theorem}
\label{TFWGLM}
Suppose the conditions of Theorem~\ref{TGLMsolution} for some \( \rups \).
Then it holds on a set \( \Omega(\xx) \) with 
\( \P\bigl( \Omega(\xx) \bigr) \geq 1 - 2 \ex^{-\xx} \)
\begin{EQA}
	\bigl\| \DPc \bigl( \tilde{\thetav} - \thetavs \bigr) - \xiv \bigr\|
	& \leq &
	\rups \, \rddelta(\rups),
	\\
    \bigl| 2 L(\tilde{\thetav},\thetavs) - \| \xiv \|^{2} \bigr|
    & \le &
    2 \rups^{2} \, \rddelta(\rups) +  \rups^{2} \, \rddelta^{2}(\rups) .
    \\
    \Bigl| 
    	\sqrt{ 2L(\tilde{\thetav},\thetavs) } 
		- \| \xiv \| 
	\Bigr|
    & \le &
    3 \rups \, \rddelta(\rups) .
\label{CorolFWGLM}
\end{EQA}
\end{theorem}

\begin{proof}
The large deviation bound of Theorem~\ref{TGLMsolution} allows to restrict the whole parameter 
space to the local vicinity \( \Thetas(\rups) \).
In this vicinity, the log-likelihood 
\( L(\thetav) = \nablaGLM^{\T} \thetav - \GLMLINK(\thetav) \) can be well approximated 
by the quadratic expansion \( \La(\thetav) \): 
\begin{EQA}
	L(\thetav)
	&=&
	(\nablaGLM - \E \nablaGLM)^{\T} \thetav + \E \nablaGLM^{\T} \thetav - \GLMLINK(\thetav),
	\\
	\La(\thetav)
	& \eqdef &
	(\nablaGLM - \E \nablaGLM)^{\T} \thetav 
		- \frac{1}{2} \| \DPc (\thetav - \thetavs) \|^{2} .
\label{LatGLMdef}
\end{EQA}

\begin{lemma}
\label{LLLaGLM}
Suppose \eqref{rddeGLM} for some \( \rups \). 
The difference \( L(\thetav) - \La(\thetav) \) is deterministic and 
it holds for each \( \thetav \in \Thetas(\rups) \)
\begin{EQA}
	\bigl| L(\thetav) - \La(\thetav) \bigr|
	& \leq &
	\frac{\rddelta(\rups)}{2} \bigl\| \DPc (\thetav - \thetavs) \bigr\|^{2}
	\leq 
	\frac{\rddelta(\rups)}{2} \rups^{2} ,
	\\
	\bigl\| \DPc^{-1} \bigl\{ \nabla L(\thetav) - \nabla \La(\thetav) \bigr\} \bigr\|
	& \leq &
	\rups \rddelta(\rups).
\label{supr0LLaGLM}
\end{EQA}
\end{lemma}

\begin{proof}
The linear stochastic terms \( (\nablaGLM - \E \nablaGLM)^{\T} \thetav \) are 
the same for 
\( L(\thetav) \) and \( \La(\thetav) \).
For the deterministic terms 
\( \E \nablaGLM^{\T} \thetav - \GLMLINK(\thetav) \) we use the Taylor formula of the second order 
at \( \thetavs \), 
the extreme point equation \( \nabla \GLMLINK(\thetavs) = \E \nablaGLM \),
and the definition \( \DPc^{2} = \IF(\thetavs) \):
\begin{EQA}
	\bigl| \E L(\thetav) - \E \La(\thetav) \bigr|
	&=&
	\bigl| 
		\GLMLINK(\thetav) - \GLMLINK(\thetavs) 
		- (\thetav - \thetavs)^{\T} \nabla \GLMLINK(\thetavs) 
		- \| \DPc (\thetav - \thetavs) \|^{2}/2 
	\bigr|
	\\
	&=&
	\frac{1}{2} \bigl| (\thetav - \thetavs)^{\T} \bigl\{ \IF(\thetavs) - \IF(\thetavd) \bigr\} 
		(\thetav - \thetavs) \bigr| ,
\label{ElElaeGLM}
\end{EQA}
where \( \thetavd \) is a point on the interval between \( \thetav \) and \( \thetavs \).
Now the condition \eqref{rddeGLM} implies 
\begin{EQA}
	\bigl| \E L(\thetav) - \E \La(\thetav) \bigr|
	& \leq &
	\frac{\rddelta(\rups)}{2} (\thetav - \thetavs)^{\T} \DPc^{2} (\thetav - \thetavs)
	=
	\frac{\rddelta(\rups)}{2} \bigl\| \DPc (\thetav - \thetavs) \bigr\|^{2}
	\leq 
	\frac{\rddelta(\rups)}{2} \rups^{2}
\label{ELELa22GLM}
\end{EQA}
and the first assertion follows. 
The second one can be proved similarly.
\end{proof}
With the approximation \eqref{supr0LLaGLM}, all the statements of the theorem 
follow from the general results of Theorem~\ref{TqapprbrGP}.
\end{proof}

To complete the study of a generalized linear model, we translate the general conditions
of Theorem~\ref{TGLMsolution} into conditions on the design \( \Psi \) and on 
individual errors \( \varepsilon_{i} \).

\begin{itemize}
	\item
	\textbf{Design regularity} is measured by the value
\begin{EQA}
	\dPsi
	& \eqdef &
	\max_{i} \| \DPc^{-1} \Psi_{i} \| .
\label{aPsiGLM}
\end{EQA}
In the case of a regular or random design, the Fisher design matrix \( \DPc^{2} = \IF(\thetavs) \)
is proportional to the sample size and thus, 
the value \( \dPsi \) is of order \( \nsize^{-1/2} \).
Our results only apply if this value is small, in particular, 
the condition \( \dPsi < 1/2 \) has to be fulfilled.

	\item 
	\textbf{Exponential moments of the errors}
Suppose that 
for some values \( \expzeta_{i} \) and fixed constants \( \CONST_{0}, \lambda_{0} > 0 \)
\begin{EQA}[c]
	\E \exp\bigl\{ \lambda_{0} \varepsilon_{i} / \expzeta_{i} \bigr\} 
	\leq 
	\CONST_{0} ,
	\qquad 
	i=1,\ldots,\nsize .
\label{Eexpl0zC0}
\end{EQA}
This condition means that the errors \( \varepsilon_{i} = Y_{i} - \E Y_{i} \) 
have exponential moments. 
In most of cases one can use \( \expzeta_{i}^{2} = \Var(Y_{i}) \).
Condition \eqref{Eexpl0zC0} implies that 
there are another constants \( \gmiid \leq \lambda_{0} \) and \( \nunu \)
such that the following condition is fulfilled:
\begin{EQA}[c]
	\E \exp\bigl\{ \lambda \varepsilon_{i} / \expzeta_{i} \bigr\} 
	\leq 
	\frac{1}{2} \nunu^{2} \lambda^{2} ,
	\qquad 
	i=1,\ldots,\nsize,
	\quad
	|\lambda| \leq \gmiid .
\label{EexpleiGLM}
\end{EQA}
This follows from the fact that each function 
\( \log \E \exp\bigl\{ \lambda \varepsilon_{i} / \expzeta_{i} \bigr\} \) analytic 
in \( \lambda \) in a vicinity of the point zero and can be well approximated 
by \( \lambda^{2}/2 \);
see \cite{GolSpo2009} for more details.
	\item 
	\textbf{Noise homogeneity} is measured by the variability of the values 
	\( \expzeta_{i} \):
\begin{EQA}
	\aexpzeta
	& \eqdef &
	\max_{i,j=1,\ldots,n} \, {\expzeta_{i}}/{\expzeta_{j}} \, .
\label{aexpzetaGLM}
\end{EQA}

	\item
	\textbf{Smoothness of the link function \( \GLMlink(\upsi) \)} can be measured 
	by its third derivative.
	It will be assumed that given \( \rr \), there is a constant \( \aGLMlink(\rr) \) 
\begin{EQA}
	\frac{|\GLMlink'''(\Psi_{i}^{\T} \thetav)|}{\GLMlink''(\Psi_{i}^{\T} \thetavs)}
	& \leq &
	\aGLMlink(\rr),
	\quad
	\thetav \in \Thetas(\rr), \quad i=1,\ldots,\nsize.
\label{aGLMlinkdef}
\end{EQA}	

	\item \textbf{Identifiability} is measured by relationship between 
	the matrices \( \DPc^{2} \) and \( \VPc^{2} \), where the matrix \( \VPc^{2} \) defined as
\begin{EQA}
	\VPc^{2}
	& \eqdef &
	\sum_{i=1}^{\nsize} \expzeta_{i}^{2} \, \Psi_{i} \Psi_{i}^{\T} .
\label{VPc2GLM}
\end{EQA}
If the the observation \( Y_{i} \) follow the GLM assumption 
\( P_{\ups_{i}} \) for \( \upsi_{i} = \Psi^{\T} \thetavs \),
that is, the model is correctly specified, then 
\( \Var(Y_{i}) = \GLMlink''(\upsi_{i}) \) and the matrices \( \VPc^{2} \) and \( \DPc^{2} \)
coincide.
In the general case under a possible model misspecification, 
the matrices \( \VPc^{2} \) and \( \DPc^{2} \) may be different.
In this case we need an identifiability condition 
\begin{EQA}[c]
	\VPc^{2}
	\leq 
	\fis^{2} \DPc^{2} 
\label{VPs2fiIFGLM}
\end{EQA}
for some \( \fis > 0 \).
This condition can be spelled out as
\begin{EQA}
	\sum_{i=1}^{\nsize} \expzeta_{i}^{2} \, \Psi_{i} \Psi_{i}^{\T}
	& \leq &
	\fis^{2} \sum_{i=1}^{\nsize} \GLMlink''(\Psi_{i}^{\T} \thetavs) \, \Psi_{i} \Psi_{i}^{\T} .
\label{sumi1nexplink}
\end{EQA}


\end{itemize}

First we discuss 
a deviation bound for the norm of the vector \( \xiv \) given by
\begin{EQA}[c]
	\xiv
	=
	\DPc^{-1} (\nablaGLM - \E \nablaGLM)
	=
	\DPc^{-1} \sum_{i=1}^{\nsize} \varepsilon_{i} \Psi_{i} \, .
\end{EQA}
The squared norm \( \| \xiv \|^{2} \) is a quadratic form of the \( \varepsilon_{i} \)'s
and one can directly apply general results for quadratic forms from 
Section~\ref{SdevboundnonGauss}.

\begin{theorem}
\label{TGLMcond}
Suppose \eqref{EexpleiGLM}, \eqref{aexpzetaGLM}, \eqref{aGLMlinkdef}, and \eqref{VPs2fiIFGLM}.
For \( \zq(\dimp,\xx) \) from \eqref{zzxxppdBlro}
with \( \zq(\dimp,\xx) \leq \sqrt{\dimp} + \sqrt{2 \xx} \), 
fix
\begin{EQA}
	\rups
	&=&
	4 \nunu \zq(\dimp,\xx) ,
\label{r04xGLM}
\end{EQA}
and suppose that \( \dPsi \) is small enough to ensure
\begin{EQA}
	\aGLMlink(\rups) \, \dPsi \, \rups
	& < &
	1/2 .
\label{delt12GLM}
\end{EQA}
Then the conditions of Theorem~\ref{TGLMsolution} are fulfilled with 
\( \rddelta(\rups) \leq \aGLMlink(\rups) \, \dPsi \, \rups \)
and the results of this theorem continue to apply.
\end{theorem}

\begin{proof}
Let \( \rups \) be fixed by \eqref{r04xGLM}. 
First we bound the value \( \rddelta(\rups) \).

\begin{lemma}
\label{LrdruaGLM}
The condition \eqref{rddeGLM} is fulfilled with 
\begin{EQA}
	\rddelta(\rups)
	&=&
	\aGLMlink(\rups) \, \dPsi \, \rups .
\label{rddeltGLM}
\end{EQA}
\end{lemma}

\begin{proof}
For each \( \thetav \in \Thetas(\rups) \) and \( i \leq \nsize \), it holds by \eqref{dPsiGLM}
\begin{EQA}
	\bigl| \Psi_{i}^{\T} \thetav - \Psi_{i}^{\T} \thetavs \bigr| 
	&=&
	\bigl| \bigl( \DPc^{-1} \Psi_{i} \bigr)^{\T} \DPc (\thetav - \thetavs) \bigr|
	\leq 
	\| \DPc^{-1} \Psi_{i} \| \, \rups 
	\leq 
	\dPsi \, \rups .
\label{PsiiTtsGLM}
\end{EQA}
This implies for the difference \( \IF(\thetav) - \IF(\thetavs) \)
\begin{EQA}
	\IF(\thetav) - \IF(\thetavs)
	&=&
	\sum_{i=1}^{\nsize} \bigl\{ 
		\GLMlink''(\Psi_{i}^{\T} \thetav) - \GLMlink''(\Psi_{i}^{\T} \thetavs) 
	\bigr\} \, \Psi_{i} \Psi_{i}^{\T}
	\\
	&=&
	\sum_{i=1}^{\nsize} 
		\frac{\GLMlink'''(\Psi_{i}^{\T} \thetavd)}{\GLMlink''(\Psi_{i}^{\T} \thetavs)} \,
		\bigl( \Psi_{i}^{\T} \thetav - \Psi_{i}^{\T} \thetav \bigr) \,
		\GLMlink''(\Psi_{i}^{\T} \thetavs) \,
	 	\Psi_{i} \Psi_{i}^{\T}
\label{IFtIfGLM}
\end{EQA}
for a point \( \thetavd \) on the interval between \( \thetavs \) and \( \thetav \).
Now \eqref{aGLMlinkdef} and \eqref{PsiiTtsGLM} imply
\begin{EQA}
	\biggl| \frac{\GLMlink'''(\Psi_{i}^{\T} \thetavd)}{\GLMlink''(\Psi_{i}^{\T} \thetavs)} \,
		\bigl( \Psi_{i}^{\T} \thetav - \Psi_{i}^{\T} \thetav \bigr)
	\biggr|
	& \leq &
	\aGLMlink(\rups) \, \dPsi \, \rups 
\label{IFtIfGLM2}
\end{EQA}
and 
\begin{EQA}
	\pm \bigl\{ \IF(\thetav) - \IF(\thetavs) \bigr\}
	& \leq &
	\aGLMlink(\rups) \, \dPsi \, \rups \, \DPc^{2}
\label{pmIFtIFaGLM}
\end{EQA}
Now the condition \eqref{rddeGLM} follows in an obvious way.
\end{proof}

This lemma and \eqref{delt12GLM} imply \( \rddelta(\rups) < 1/2 \).
Now we show that \eqref{EexpleiGLM} implies \eqref{PxizxGLM}. 

\begin{lemma}
\label{LexpED0GLM}
Let the errors \( \varepsilon_{i} = Y_{i} - \E Y_{i} \) be independent 
and follow \eqref{EexpleiGLM}. Then 
\begin{EQA}
	\log \E \exp\bigl\{\uv^{\T} \VPc^{-1} (\nablaGLM - \E \nablaGLM) \bigr\}
	& \leq &
	\frac{\nunu^{2}}{2} \| \uv \|^{2} \, ,
	\qquad
	\| \uv \| \leq \gm
\label{logEexpuV22GLM}
\end{EQA}
where \( \VPc^{2} \) is from \eqref{VPc2GLM} 
and \( \gm \) is given by  
\begin{EQA}
	\gm
	& \eqdef &
	\frac{\gmiid}{\dPsi \, \aexpzeta},
\label{dPsiGLM}
\end{EQA}
for \( \aexpzeta \) from \eqref{aexpzetaGLM}.
\end{lemma}

\begin{proof}
The formula \eqref{nablaGLMdef} and
independence of the \( \varepsilon_{i} \)'s imply for any vector \( \uv \in \R^{\dimp} \)
with \( \| \uv \| \leq \gm \)
\begin{EQA}
	\log \E \exp\bigl\{\uv^{\T} \VPc^{-1} (\nablaGLM - \E \nablaGLM) \bigr\}
	&=&
	\sum_{i=1}^{\nsize} \log \E \exp\bigl\{ \lambda_{i} \varepsilon_{i}/\expzeta_{i} \bigr\} ,
\label{logEexpGLM1n}
\end{EQA}
where the definitions \eqref{dPsiGLM} and \eqref{aexpzetaGLM} imply
for \( \lambda_{i} = \uv^{\T} \VPc^{-1} \Psi_{i} \, \expzeta_{i} \)
\begin{EQA}
	|\lambda_{i}|
	& = &
	|\uv^{\T} \VPc^{-1} \Psi_{i}| \, \expzeta_{i} 
	\leq 
	\gm \,\, \| \VPc^{-1} \Psi_{i} \| \, \expzeta_{i}  
	\leq 
	\gmiid .
\label{lamiGLM}
\end{EQA}
Therefore, by \eqref{EexpleiGLM} and the definition of \( \VPc^{2} \)
\begin{EQA}
	\log \E \exp\bigl\{\uv^{\T} \VPc^{-1} \nablaGLM \bigr\}
	& \leq &
	\frac{\nunu^{2}}{2} \sum_{i=1}^{\nsize} \lambda_{i}^{2}
	=
	\frac{\nunu^{2}}{2} \sum_{i=1}^{\nsize} 
		\uv^{\T} \VPc^{-1} \bigl( \Psi_{i} \Psi_{i}^{\T} \, \expzeta_{i}^{2} \bigr) 
		\, \VPc^{-1} \uv 
	=
	\frac{\nunu^{2}}{2} \| \uv \|^{2} \, ,
\label{logEnu2GLM}
\end{EQA}
and the assertion follows.
\end{proof}

The result of Lemma~\ref{LexpED0GLM} provides exponential moments of \( \xiv \)
and one can apply Theorem~\ref{LLbrevelocro} from Section~\ref{SdevboundnonGauss}
yielding the bound \eqref{PxizxGLM} under the condition
\begin{EQA}
	\frac{1 - \rddelta(\rups)}{2} \rups  
	& \geq &
	\nunu \, \zq(\dimp,\xx) 
\label{fr1deru2ru}
\end{EQA}
which is obviously fulfilled for our choice of \( \rups = 4 \nunu \, \zq(\dimp,\xx) \) 
in view of \( \rddelta(\rups) < 1/2 \).
This will also provide \eqref{PxizxGLM}.
All the conditions of Theorem~\ref{TGLMsolution} have been checked.
\end{proof}

\Section{Estimation for a penalized GLM}
This section briefly discusses what will be changed if the GLM \eqref{LtGLM} is penalized
by a roughness penalty term \( \| \GP \thetav \|^{2} \).
The corresponding penalized log-likelihood \( \LGP(\thetav) \) reads as 
\begin{EQA}
	\LGP(\thetav)
	&=&
	\nablaGLM^{\T} \thetav - \GLMLINK(\thetav) - \| \GP \thetav \|^{2} .
\label{LGPtGLM}
\end{EQA}
The penalized MLE and its target are defined by maximizing \( \LGP(\thetav) \) and 
its expectation:
\begin{EQA}
	\tilde{\thetav}_{\GP}
	& \eqdef &
	\argmax_{\thetav \in \Theta} 
		\bigl\{ \nablaGLM^{\T} \thetav - \GLMLINK(\thetav) - \| \GP \thetav \|^{2} \bigr\} ,
	\\
	\thetavsGP
	& \eqdef &
	\argmax_{\thetav \in \Theta} 
		\bigl\{ \E \nablaGLM^{\T} \thetav - \GLMLINK(\thetav) - \| \GP \thetav \|^{2} \bigr\} .
\label{ttGPdefGLM}
\end{EQA}
Further, define the matrix \( \DPGP \) by \( \DPGP^{2} = \IFGP(\thetavsGP) \) for
\begin{EQA}[c]
	\IFGP(\thetav)
	\eqdef
	\IF(\thetav) + \GP^{2}
	=
	\sum_{i=1}^{\nsize} \Psi_{i} \Psi_{i}^{\T} \GLMlink''(\Psi_{i}^{\T} \thetav) + \GP^{2}.
\label{IFGLMdefGP}
\end{EQA}
One can see that the use of penalization leads to a growth of the ``information matrix''
\( \DPGP^{2} \) relative to the non-penalized case.
The stochastic term \( (\nablaGLM - \E \nablaGLM)^{\T} \thetav \) of \( \LGP(\thetav) \)
remains the same as in the non-penalized case, thus, 
the matrix \( \VPc^{2} \) from \eqref{VPc2GLM} can be used here as well
and the identifiability condition \eqref{VPs2fiIFGLM} continues to hold. 

The local vicinity \( \ThetasGP(\rr) \) of \( \thetavsGP \)is now defined as
\begin{EQA}
	\ThetasGP(\rr)
	& \eqdef &
	\bigl\{ \thetav \colon \| \DPGP (\thetav - \thetavsGP) \| \leq \rr \bigr\}.
\label{ThrGLMGP}
\end{EQA}
The concentration result for \( \tilde{\thetav}_{\GP} \) can be easily extended to the 
penalized case.

\begin{theorem}
\label{TGLMsolutionGP}
Let, for some \( \rupsGP > 0 \), the matrix function \( \IFGP(\thetav) \) from \eqref{IFGLMdefGP} fulfill 
with \( \DPGP^{2} = \IFGP(\thetavsGP) \)
\begin{EQA}
	\sup_{\thetav \in \ThetasGP(\rups)} 
		\| \DPGP^{-1} \, \IFGP(\thetav) \, \DPGP^{-1} - \Id_{\dimp} \|_{\oper}
	& \leq &
	\rddelta(\rupsGP) 
\label{rddeGLMGP}
\end{EQA}
for \( \rddelta(\rupsGP) < 1 \).
Let also \( \nablaGLM \) from \eqref{nablaGLMdef} follow for \( \xx > 0 \) the probability bound
\begin{EQA}
	\P\Bigl( 
		\| \DPGP^{-1} (\nablaGLM - \E \nablaGLM) \| > \frac{1 - \rddelta(\rupsGP)}{2} \rupsGP  
	\Bigr)
	& \leq &
	2 \ex^{-\xx} .
\label{PxizxGLMGP}
\end{EQA}
Then the solution \( \tilde{\thetav}_{\GP} \) of \eqref{ttGPdefGLM} satisfies
\begin{EQA}
	\P\bigl( \tilde{\thetav} \not\in \ThetasGP(\rupsGP) \bigr)
	& \leq &
	2 \ex^{-\xx} .
\label{PttsTsGLMGP}
\end{EQA}
Moreover, under conditions 
\eqref{EexpleiGLM}, \eqref{aexpzetaGLM}, \eqref{aGLMlinkdef}, and \eqref{VPs2fiIFGLM},
one can fix
\begin{EQA}
	\rupsGP
	&=&
	4 \nunu \, \zq(\BBGP,\xx) 
	\qquad
	\text{for}
	\qquad
	\BBGP
	\eqdef 
	\DPGP^{-1} \VPc^{2} \DPGP^{-1} 
\label{r04xGLMGP}
\end{EQA}
with \( \zq(\BBGP,\xx) \leq \sqrt{\dimG} + \sqrt{2 \xx} \) from \eqref{zzxxppdBlroB}.
Then all the statements of Theorem~\ref{TFWGLM} hold for the pair 
\( \tilde{\thetav}_{\GP}, \thetavsGP \) 
with \( \xivGP \eqdef \DPGP^{-1} \bigl( \nablaGLM - \E \nablaGLM \bigr) \) 
in  place of \( \xiv \) and \( \rupsGP \) place of \( \rups \).
\end{theorem}

The proof of the non-penalized case applies here with obvious changes in notation.
However, at one place the difference is essential.
Namely, the radius \( \rupsGP \) can be much smaller and it depends on the effective dimension 
\( \dimG = \tr(\BBGP) = \tr(\DPGP^{-1} \VPc^{2} \DPGP^{-1}) \) 
rather than on the total dimension \( \dimp \).

\ifbook{}
{\Section{Random design}
Let \( \Psiv \) be a random design matrix.
This section extends the obtained results to 
}


\appendix

\Chapter{Deviation bounds for quadratic forms}
\label{Sprobabquad}

Here we collect some probability bounds for Gaussian and non-Gaussian quadratic forms.

\Section{Gaussian quadratic forms}
The next result explains the concentration effect of 
\( \gaussv^{\T} \BB \gaussv \)
for a standard Gaussian vector \( \gaussv \) and a symmetric matrix \( \BB \).
We use a version from \cite{laurentmassart2000}.


\begin{theorem}
\label{TexpbLGA}
\label{Lxiv2LD}
\label{Cuvepsuv0}
Let \( \gaussv \) be a standard normal Gaussian vector and \( \BB \) be symmetric positive.
Then with \( \dimA = \tr(\BB) \), \( \vA^{2} = \tr(\BB^{2}) \), and 
\( \supA = \| \BB \|_{\oper} \), it holds for each \( \xx \geq 0 \)
\begin{EQA}
	\P\bigl( \gaussv^{\T} \BB \gaussv > \dimA + 2 \vA \xx^{1/2} + 2 \supA \xx \bigr)
	& \leq &
	\ex^{-\xx} .
\label{Pxiv2dimAvp12}
\end{EQA}
This implies for any positive \( \BB \) 
\begin{EQA}
	\P\bigl( \| \BB^{1/2} \gaussv \| > \dimA^{1/2} + (2 \supA \xx)^{1/2} \bigr)
	& \leq &
	\ex^{-\xx} .
\label{Pxiv2dimAxx12}
\end{EQA}
Also
\begin{EQA}
	\P\bigl( \gaussv^{\T} \BB \gaussv < \dimA - 2 \vA \xx^{1/2} \bigr)
	& \leq &
	\ex^{-\xx} .
\label{Pxiv2dimAvp12m}
\end{EQA}
If \( \BB \) is symmetric but non necessarily positive then
\begin{EQA}
	\P\bigl( \bigl| \gaussv^{\T} \BB \gaussv - \dimA \bigr| > 2 \vA \xx^{1/2} + 2 \supA \xx \bigr)
	& \leq &
	2 \ex^{-\xx} .
\label{PxivTBBdimA2vp}
\end{EQA}
\end{theorem}

{
\begin{proof}
Normalisation by \( \supA \) reduces the statement to the case with \( \supA = 1 \).
Further, the standard rotating arguments allow to reduce the Gaussian quadratic form 
\( \| \gaussv \|^{2} \) to the chi-squared form:
\begin{EQA}
	\gaussv^{\T} \BB \gaussv
	&=&
	\sum_{j=1}^{\dimp} \lambda_{j} \nu_{j}^{2}
\label{xiv2sj1p}
\end{EQA}
with independent standard normal r.v.'s \( \nu_{j} \).
Here \( \lambda_{j} \in [0,1] \) are eigenvalues of \( \BB \), and 
\( \dimA = \lambda_{1} + \ldots + \lambda_{\dimp} \), 
\( \vA^{2} = \lambda_{1}^{2} + \ldots + \lambda_{\dimp}^{2} \).
One can easily 
compute the exponential moment of \( (\gaussv^{\T} \BB \gaussv - \dimA)/2 \):
for each positive \( \mu < 1 \)
\begin{EQA}
	\log \E \exp\bigl\{ \mu (\gaussv^{\T} \BB \gaussv - \dimA)/2 \bigr\}
	&=&
	\frac{1}{2} \sum_{j=1}^{\dimp} \bigl\{ - \mu \lambda_{j} - \log(1 - \mu \lambda_{j}) \bigr\} .
\label{lEemux2p2}
\end{EQA}

\begin{lemma}
Let \( \mu \lambda_{j} < 1 \) and \( \lambda_{j} \leq 1 \).
Then 
\begin{EQA}
	\frac{1}{2} \sum_{j=1}^{\dimp} \bigl\{ - \mu \lambda_{j} - \log(1 - \mu \lambda_{j}) \bigr\}
	& \leq &
	\frac{\mu^{2} \vA^{2}}{4 (1 - \mu)} \, .
\label{jmu2v221mu}
\end{EQA}
\end{lemma}

\begin{proof}
In view of \( \mu \lambda_{j} < 1 \), it holds for every \( j \)
\begin{EQA}
	- \mu \lambda_{j} - \log(1 - \mu \lambda_{j}) 
	&=&
	\sum_{k=2}^{\infty} \frac{(\mu \lambda_{j})^{k}}{k}
	\\
	& \leq &
	\frac{(\mu \lambda_{j})^{2}}{2}
	\sum_{k=0}^{\infty} (\mu \lambda_{j})^{k}
	\leq 
	\frac{(\mu \lambda_{j})^{2}}{2 (1 - \mu \lambda_{j})} 
	\leq 
	\frac{(\mu \lambda_{j})^{2}}{2 (1 - \mu)},
\label{jmu2v221mup}
\end{EQA}
and thus
\begin{EQA}
	\frac{1}{2} \sum_{j=1}^{\dimp} \bigl\{ - \mu \lambda_{j} - \log(1 - \mu \lambda_{j}) \bigr\}
	& \leq &
	\sum_{j=1}^{\dimp} \frac{(\mu \lambda_{j})^{2}}{4 (1 - \mu)} 
	\leq 
	\frac{\mu^{2} \vA^{2}}{4 (1 - \mu)} \, .
\label{sjmu2v221mu}
\end{EQA}
\end{proof}
The next technical lemma is helpful.

\begin{lemma}
\label{Lmuvpxx}
For each \( \vA > 0 \) and \( \xx > 0 \), it holds
\begin{EQA}
	\inf_{\mu > 0} \biggl\{ 
		- \mu \bigl( \vA \xx^{1/2} + \xx \bigr) + \frac{\mu^{2} \vA^{2}}{4 (1 - \mu)} 
	\biggr\}
	& \leq &
	- \xx .
\label{infmuxxvp}
\end{EQA}
\end{lemma}

\begin{proof}
Let pick up 
\begin{EQA}
	\mu 
	&=& 
	1 - \frac{1}{2\xx^{1/2}/\vA + 1} = \frac{\xx^{1/2}}{\xx^{1/2} + \vA/2} , 
\label{mu12xx12vp1m1}
\end{EQA}
so that \( \mu / (1 - \mu) = 2 \xx^{1/2}/\vA \). Then
\begin{EQA}
	&& \nquad
	- \mu \bigl( \vA \xx^{1/2} + \xx \bigr) + \frac{\mu^{2} \vA^{2}}{4 (1 - \mu)}
	\\
	&=&
	- \mu \bigl( \vA \xx^{1/2} + \xx + \vA^{2}/4 \bigr)
	+ \frac{\mu \vA^{2}}{4 (1 - \mu)}
	\\
	&=&
	- \frac{\xx^{1/2}}{\xx^{1/2} + \vA/2} \bigl( \xx^{1/2} + \vA/2 \bigr)^{2} 
	+ \frac{2 \xx^{1/2} \vA }{4}
	=
	- \xx 
\label{mux2xv4x12}
\end{EQA}
and the result follows.
\end{proof}

Now we apply the Markov inequality 
\begin{EQA}
	&& \nquad
	\log \P\bigl( \gaussv^{\T} \BB \gaussv > \dimA + 2 \vA \xx^{1/2} + 2 \xx \bigr)
	=
	\log \P\bigl( (\gaussv^{\T} \BB \gaussv - \dimA) / 2 > \vA \xx^{1/2} + \xx \bigr)
	\\
	& \leq &
	\inf_{\mu > 0} \biggl\{ 
		- \mu \bigl( \vA \xx^{1/2} + \xx \bigr) 
		+ \log\E \exp\bigl\{ \mu (\gaussv^{\T} \BB \gaussv - \dimA)/2 \bigr\}
	\biggr\}
	\\
	& \leq &
	\inf_{\mu > 0} \biggl\{ 
		- \mu \bigl( \vA \xx^{1/2} + \xx \bigr) + \frac{\mu^{2} \vA^{2}}{4 (1 - \mu)}
	\biggr\}
	\leq 
	- \xx
\label{x2xv4x12}
\end{EQA}
and the first assertion \eqref{Pxiv2dimAvp12} follows.
The second statement follows from the first one by 
\( \tr(\BB^{2}) \leq \| \BB \|_{\oper} \tr(\BB) = \supA \, \dimA \).

Similarly for any \( \mu > 0 \)
\begin{EQA}
	\P\bigl( \gaussv^{\T} \BB \gaussv - \dimA < - 2 \vA \sqrt{\xx} \bigr)
	& \leq &
	\exp\bigl( - \mu \vA \sqrt{\xx} \bigr)
	\E \exp\Bigl( - \frac{\mu}{2} (\gaussv^{\T} \BB \gaussv - \dimA) \Bigr) .
\end{EQA}
By \eqref{lEemux2p2}
\begin{EQA}
	\log \E \exp\bigl\{ - \mu (\gaussv^{\T} \BB \gaussv - \dimA)/2 \bigr\}
	&=&
	\frac{1}{2} \sum_{j=1}^{\dimp} 
		\bigl\{ \mu \lambda_{j} - \log(1 + \mu \lambda_{j}) \bigr\} .
\label{lEemux2p2m}
\end{EQA}
and 
\begin{EQA}
	\frac{1}{2} \sum_{j=1}^{\dimp} \bigl\{ \mu \lambda_{j} - \log(1 + \mu \lambda_{j}) \bigr\}
	&=&
	\frac{1}{2} \sum_{j=1}^{\dimp} \sum_{k=2}^{\infty} \frac{(- \mu \lambda_{j})^{k}}{k}
	\leq 
	\sum_{j=1}^{\dimp}\frac{(\mu \lambda_{j})^{2}}{4} 
	=
	\frac{\mu^{2} \vA^{2}}{4} .
\label{jmu2v221mum}
\end{EQA}
Here the choice \( \mu = 2 \sqrt{\xx} / \vA \) yields \eqref{Pxiv2dimAvp12m}.

One can put together the arguments used for obtaining the lower and the upper bound 
for getting a bound for a general 
quadratic form \( \gaussv^{\T} \BB \gaussv \), where \( \BB \) is symmetric but not necessarily 
positive.
\end{proof}
}
%
%

Finally we apply this result to weighted sums of centered \( \gauss_{i}^{2} \).
\begin{corollary}
\label{Cuvepsuv}
For any unit vector \( \uv = (u_{i}) \in \R^{n} \) 
and standard normal r.v.'s \( \gauss_{i} \), it holds with 
\( \| \uv \|_{\infty} \eqdef \max_{i} |u_{i}| \)
\begin{EQA}
	\P\biggl( 
		\biggl| \sum_{i=1}^{n} u_{i} (\gauss_{i}^{2} - 1) \biggr| 
		\geq  
		2 \xx^{1/2} + 2 \| \uv \|_{\infty} \xx
	\biggr)
	& \leq &
	2 \ex^{-\xx} .
\label{Puitei2m1v46x}
\label{Puitei2m1v46xx}
\end{EQA}
\end{corollary}

\begin{proof}
The statement follows directly from Theorem~\ref{Cuvepsuv0}.
It suffices to notice \( \vA^{2} = \| \uv \|^{2} = 1 \).
\end{proof}

As a special case, we present a bound for the chi-squared distribution 
corresponding to \( \BB = \Id_{\dimp} \).
Then \( \tr (\BB) = \dimp \), \( \tr(\BB^{2}) = \dimp \) and \( \supA(\BB) = 1 \).

\begin{corollary}
\label{Cchi2p}
Let \( \gaussv \) be a standard normal vector in \( \R^{\dimp} \).
Then
\begin{EQA}[lcl]
\label{Pxi2pm2px}
	\P\bigl( \| \gaussv \|^{2} \geq \dimp + 2 \sqrt{\dimp \xx} + 2 \xx \bigr)
	& \leq &
	\ex^{-\xx},
	\\
	\P\bigl( \| \gaussv \| \,\,  \geq \sqrt{\dimp} + \sqrt{2 \xx} \bigr)
	& \leq &
	\ex^{-\xx} ,
\label{Pxi2pm2px12}
	\\
	\P\bigl( \| \gaussv \|^{2} \leq \dimp - 2 \sqrt{\dimp \xx} \bigr)
	& \leq &
	\ex^{-\xx}	.
\label{Pxi2pm2px22}
\end{EQA}
\end{corollary} 

The previous results are mainly stated for a standard Gaussian vector \( \gaussv \in \R^{n} \).
Now we extend it to the case of a zero mean Gaussian vector \( \xiv \) with the 
\( n \times n \) covariance matrix \( \Covm = (\covm_{ij}) \) with 
\( \lambda_{\max}(\Covm) \leq \supAB \).
Given a unit vector \( \uv = (u_{1},\ldots,u_{n})^{\T} \in \R^{n} \), consider the quadratic form
\begin{EQA}
	\GQF
	&=&
	\sum_{i=1}^{n} u_{i} \xi_{i}^{2} .
\label{S1nuiei2d}
\end{EQA}
We aim at bounding \( \GQF - \E \GQF \).
To apply the result of Theorem~\ref{Cuvepsuv0} represent \( \GQF \)
as \( \gaussv^{\T} \BB \gaussv \) with \( \BB \) depending on 
\( \uv \) and \( \Covm \).
More precisely, let \( \xiv = \Covm^{1/2} \gaussv \) for a standard Gaussian vector 
\( \gaussv \in \R^{n} \).
Then with \( \Uv = \diag(u_{1},\ldots,u_{n}) \), it holds
\begin{EQA}
	S
	&=&
	\tr\bigl( \Uv \xiv \xiv^{\T} \bigr)
	=
	\tr\bigl( \Uv \Covm^{1/2} \gaussv \gaussv^{\T} \Covm^{1/2} \bigr)
	=
	\tr\bigl( \BB \gaussv \gaussv^{\T} \bigr)
	=
	\gaussv^{\T} \BB \gaussv
\label{•}
\end{EQA}
with \( \BB = \Covm^{1/2} \Uv \Covm^{1/2} \).
Therefore, the bound \( \| \Covm \|_{\oper} \leq \supAB \) implies 
\begin{EQA}
	\supA
	&=&
	\supA(\BB)
	=
	\| \Covm^{1/2} \Uv \Covm^{1/2} \|_{\oper}
	\leq 
	\supAB \, \| \uv \|_{\infty} \, ,
	\\
	\vA^{2}
	&=&
	\tr(\BB^{2})
	=
	\tr\bigl( \Covm^{1/2} \Uv \Covm \Uv \Covm^{1/2} \bigr)
	\leq 
	\supAB \tr\bigl( \Uv \Covm \Uv \bigr)
	\leq 
	{\supAB}^{2} \| \uv \|^{2} = {\supAB}^{2}.
\label{supAvp2BB2}
\end{EQA}
Now the general results of Theorem~\ref{Cuvepsuv0} implies the result 
similar to Corollary~\ref{Cuvepsuv}.
 
\begin{corollary}
\label{CuvepsuvnG}
For any unit vector \( \uv = (u_{i}) \in \R^{n} \), \( \| \uv \| = 1 \),
and normal zero mean vector \( \xiv \sim \ND(0,\Covm) \) in \( \R^{n} \) with 
\( \| \Covm \|_{\oper} \leq \supAB \), it holds 
\begin{EQA}
	\P\biggl( 
		\biggl| \sum_{i=1}^{n} u_{i} (\xi_{i}^{2} - \E \xi_{i}^{2}) \biggr| 
		\geq  
		2 \supAB \, \xx^{1/2} + 2 \supAB \, \| \uv \|_{\infty} \xx
	\biggr)
	& \leq &
	2 \ex^{-\xx} .
\label{Puitei2m1v46xnG}
\end{EQA}
\end{corollary}

It is worth noting that the identity \( \| \uv \| = 1 \) implies 
\( \| \uv \|_{\infty} \leq 1 \).
Moreover, in typical situations,
\( \| \uv \|_{\infty} \asymp n^{-1/2} \), and the leading term in the bounds of 
Corollaries~\ref{Cuvepsuv} and \ref{CuvepsuvnG} is \( 2 \supAB \, \xx^{1/2} \).

\Section{Deviation bounds for non-Gaussian quadratic forms}
\label{SdevboundnonGauss}
This section presents an extension of the results obtained for Gaussian quadratic forms
to the non-Gaussian case.

\Subsection{Deviation bounds for the norm of a standardized non-Gaussian vector}
The bounds of Corollary~\ref{Cchi2p} heavily use normality of the vector \( \xiv \).
This section extends the upper bound \eqref{Pxi2pm2px} to the case when  
\( \xiv \) has some exponential moments.
More exactly, suppose  for some fixed \( \gm > 0 \) that 
\begin{EQA}[c]
    \log \E \exp\bigl( \gammav^{\T} \xiv \bigr)
    \le
    \| \gammav \|^{2}/2,
    \qquad
    \gammav \in \R^{\dimp}, \, \| \gammav \| \le \gm .
\label{expgamgm}
\end{EQA}
For ease of presentation, assume below that \( \gm \) is sufficiently large, namely, 
\( 0.3 \gm \ge \sqrt{\dimp} \).
In typical examples of an i.i.d. sample, \( \gm \asymp \sqrt{n} \).
Define
\begin{EQA}
	\xxc
	& \eqdef &
	\gm^{2}/4,
	\\
	\zqc^{2}
	& \eqdef &
	\dimp + \sqrt{\dimp \gm^{2}} + \gm^{2}/2 
	=
	\gm^{2} \bigl( 1/2 + \sqrt{\dimp/\gm^{2}} + \dimp/\gm^{2} \bigr),
	\\
	\gmc
	& \eqdef &
	\frac{\gm \, \bigl( 1/2 + \sqrt{\dimp/\gm^{2}} + \dimp/\gm^{2} \bigr)^{1/2}}
		 {1 + \sqrt{\dimp/\gm^{2}}} .
\label{xxcgm24mucyyc2}
\end{EQA}
Note that with \( \alp = \sqrt{\dimp / \gm^{2}} \leq 0.3 \), one has
\begin{EQA}
	\zqc^{2}
	&=&
	\gm^{2} \bigl( 1/2 + \alp + \alp^{2} \bigr),
	\\
	\gmc
	&=&
	\gm \,\, \frac{\bigl( 1/2 + \alp + \alp^{2} \bigr)^{1/2}}{1 + \alp}
\label{zqcgmcalp2}
\end{EQA}
so that \( \zqc^{2} / \gm^{2} \in [1/2,1] \) and \( \gmc^{2} / \gm^{2} \in [1/2,1] \).

\begin{theorem}
\label{LLbrevelocro}   
Let \eqref{expgamgm} hold and 
\( 0.3 \gm \ge \sqrt{\dimp} \).
Then
for each \( \xx > 0 \)
\begin{EQA}
    \P\bigl( \| \xiv \| \ge \zq(\dimp,\xx) \bigr)
    & \le &
    2 \ex^{-\xx} + 8.4 \ex^{-\xxc } \Ind(\xx < \xxc) ,
\label{PxivbzzBBro}
\end{EQA}    
where \( \zq(\dimp,\xx) \) is defined by
\begin{EQA}
\label{PzzxxpBro}
    \zq(\dimp,\xx)
    & \eqdef &
    \begin{cases}
        \bigl( \dimp + 2 \sqrt{\dimp \xx} + 2 \xx\bigr)^{1/2}, &  \xx \le \xxc  , \\
        \zqc + 2 \gmc^{-1} (\xx - \xxc)   , & \xx > \xxc .
    \end{cases}
\label{zzxxppdBlro}
\end{EQA}    
\end{theorem}

Depending on the value \( \xx \), we have two types of tail behavior of the 
quadratic form \( \| \xiv \|^{2} \). 
For \( \xx \le \xxc = \gm^{2}/4 \), we have the same deviation bounds as in the Gaussian case
with the extra-factor two in the deviation probability.
Remind that one can use a simplified expression 
\( \bigl( \dimp + 2 \sqrt{\dimp \xx} + 2 \xx\bigr)^{1/2} \leq \sqrt{\dimp} + \sqrt{2 \xx} \).
For \( \xx > \xxc \), we switch to the special regime driven by the exponential moment
condition \eqref{expgamgm}.
Usually \( \gm^{2} \) is a large number (of order \( n \) in the i.i.d. setup) and 
the second term in \eqref{PxivbzzBBro} can be simply ignored. 

The main step of the proof is the following exponential bound.
\begin{lemma}
\label{Lexpxiv}
Suppose \eqref{expgamgm}.
For any \( \mu < 1 \) with
\( \gm^{2} > \dimp \mu \), it holds 
\begin{EQA}
\label{Eexp2xi}
    \E \exp\Bigl( \frac{\mu \| \xiv \|^{2}}{2} \Bigr)
        \Ind\Bigl( \| \xiv \| \le \gm/\mu - \sqrt{\dimp/\mu} \Bigr)
    & \le &
    2 (1 - \mu)^{-\dimp/2} .
\end{EQA}
\end{lemma}

\begin{proof}
Let \( \varepsilonv \) be a standard normal vector in \( \R^{\dimp} \) and
\( \uv \in \R^{\dimp} \).
The bound \( \P\bigl( \| \varepsilonv \|^{2} > \dimp \bigr) \le 1/2 \) and
the triangle inequality imply
for any vector \( \uv \) and any \( \rr \) with \( \rr \ge \| \uv \| + \dimp^{1/2} \) that
\( \P\bigl( \| \uv + \varepsilonv \| \le \rr \bigr) \ge 1/2 \).
Let us fix some \( \xiv \) with  \( \| \xiv \| \le \gm/\mu - \sqrt{\dimp/\mu} \)
and denote by \( \P_{\xiv} \) the conditional probability given \( \xiv \).
The previous arguments yield:
\begin{EQA}[c]
\P_{\xiv}\bigl( \| \varepsilonv + \mu^{1/2} \xiv \| \le \mu^{-1/2} \gm \bigr) \ge 0.5.
\end{EQA}
It holds with \( c_{p} = (2\pi)^{-\dimp/2} \)
\begin{EQA}
    && \nquad
    c_{p} \int \exp\Bigl( \gammav^{\T} \xiv - \frac{\| \gammav \|^{2}}{2 \mu}  \Bigr)
        \Ind(\| \gammav \| \le \gm) d\gammav
    \\
    &=&
    c_{p} \exp\bigl( \mu \| \xiv \|^{2} / 2 \bigr)
    \int \exp\Bigl(
        - \frac{1}{2}  \bigl\| \mu^{-1/2} \gammav - \mu^{1/2} \xiv \bigr\|^{2}
    \Bigr) \Ind(\mu^{-1/2} \| \gammav \| \le \mu^{-1/2} \gm) d\gammav
    \\
    & = &
    \mu^{\dimp/2} \exp\bigl( \mu \| \xiv \|^{2} / 2 \bigr)
    \P_{\xiv}\bigl( \| \varepsilonv + \mu^{1/2} \xiv \| \le \mu^{-1/2} \gm \bigr)
    \\
    & \ge &
    0.5 \mu^{\dimp/2} \exp\bigl( \mu \| \xiv \|^{2} / 2 \bigr) ,
\label{intggvv}
\end{EQA}
because \( \| \mu^{1/2} \xiv \| + \dimp^{1/2} \le \mu^{-1/2} \gm \).
This implies in view of \( \dimp < \gm^{2}/\mu \) that
\begin{EQA}
    && \nquad
    \exp\bigl( {\mu \| \xiv \|^{2}}/{2} \bigr)
        \Ind\bigl( \| \xiv \|^{2} \le \gm/\mu - \sqrt{\dimp/\mu} \bigr)
    \\
    & \le &
    2 \mu^{-\dimp/2} c_{p}
    \int \exp\Bigl( \gammav^{\T} \xiv - \frac{\| \gammav \|^{2}}{2\mu} \Bigr)
        \Ind(\| \gammav \| \le \gm) d\gammav .
\label{expxiv1cp}
\end{EQA}
Further, by \eqref{expgamgm}
\begin{EQA}
    &&
    \nquad
    c_{p} \E \int \exp\Bigl( \gammav^{\T} \xiv - \frac{1}{2\mu} \| \gammav \|^{2} \Bigr)
        \Ind(\| \gammav \| \le \gm) d\gammav
    \\
    & \le &
    c_{p} \int \exp\Bigl( - \frac{\mu^{-1} - 1}{2} \| \gammav \|^{2} \Bigr)
        \Ind(\| \gammav \| \le \gm) d\gammav
    \\
    & \le &
    c_{p} \int \exp\Bigl( - \frac{\mu^{-1} - 1}{2} \| \gammav \|^{2} \Bigr) d \gammav
    \\
    & \le &
    (\mu^{-1} - 1)^{- \dimp/2}
\label{nununu}
\end{EQA}
and \eqref{Eexp2xi} follows.
\end{proof}

Due to this result, the scaled squared norm \( \mu \| \xiv \|^{2}/2 \) after a proper
truncation possesses the same exponential moments as in the Gaussian case.
A straightforward implication is the probability bound
\( \P\bigl( \| \xiv \|^{2} > \dimp + u \bigr) \) 
with \( u = 2 \sqrt{\dimp \xx} + 2 \xx \).
Namely, given \( \xx \), define 
\begin{EQA}
	\mu 
	&=& 
	\mu(\xx)
	=
	\frac{1}{1 + 0.5\sqrt{\dimp/\xx}} \, . 
\label{mumuxxgm2}
\end{EQA}
Also define for \( \xxc = \gm^{2}/4 \)
\begin{EQA}
	\muc
	& \eqdef &
	\mu(\xxc)
	=
	\frac{1}{1 + \sqrt{\dimp/\gm^{2}}} \, .
\label{mucdef1alp}
\end{EQA}
Obviously, \( \mu \leq \muc \) for \( \xx \leq \xxc \). 
Now we obtain similarly to the Gaussian case in Lemma~\ref{Lmuvpxx}  
for \( u = 2 \sqrt{\dimp \xx} + 2 \xx \)
\begin{EQA}
    && \nquad
    \P\Bigl( \| \xiv \|^{2} > \dimp + u, \, \| \xiv \| \le \gm/\mu - \sqrt{\dimp/\mu} \Bigr)
    \\
    & \le &
    \exp\Bigl\{ - \frac{\mu (\dimp + u)}{2} \Bigr\}
    \E\exp \Bigl( \frac{\mu \| \xiv \|^{2}}{2} \Bigr)
    \Ind\Bigl( \| \xiv \| \le \gm/\mu - \sqrt{\dimp/\mu} \Bigr)
    \\
    & \le &
    2 \exp\Bigl\{ - \frac{1}{2} \bigl[ \mu (\dimp + u) + \dimp \log( 1 - \mu ) \bigr]
    \Bigr\} 
\label{logPmunu0}
\end{EQA}
and by \eqref{mux2xv4x12} with \( \vp^{2} = \dimp \), it holds for \( \mu \) from \eqref{mumuxxgm2}
\begin{EQA}
	\mu (\dimp + 2 \sqrt{\dimp \xx} + 2 \xx) + \dimp \log( 1 - \mu )
	& \geq &
	2 \xx .
\label{mudp2a}
\end{EQA}
Now we show that the constraint \( \| \xiv \| \le \gm/\mu - \sqrt{\dimp/\mu} \) in \eqref{logPmunu0}
can be replaced by the inequality \( \| \xiv \| \leq \zqc \). 

\begin{lemma}
\label{Lmucyyc}
Let \( 0.3 \gm \ge \sqrt{\dimp} \), \( \xx \leq \xxc = \gm^{2}/4 \), 
and \( \mu = 1/(1 + 0.5\sqrt{\dimp/\xx}) \).
Then 
\begin{EQA}
	\dimp + 2 \sqrt{\dimp \xx} + 2 \xx
	& \leq &
	\dimp + 2 \sqrt{\dimp \xxc} + 2 \xxc ,
	\\ 
	\gm/\mu - \sqrt{\dimp/\mu} 
	& \geq & 
	\gm/\muc - \sqrt{\dimp/\muc} ,
	\\
	\dimp + 2 \sqrt{\dimp \xxc} + 2 \xxc
	& \leq &
	\bigl( \gm/\muc - \sqrt{\dimp/\muc} \bigr)^{2} .
\label{dimp2x2xgm2}
\end{EQA}
\end{lemma}
\begin{proof}
The definition implies \( \mu \leq \muc \) for \( \xx \leq \xxc \) and thus 
the first two inequalities of the lemma are obvious.
Therefore, it remains to check \eqref{dimp2x2xgm2}.
Denote \( \alp^{2} = \dimp / \gm^{2} \). 
Then \( \muc^{-1} = 1 + \alp \) and
\begin{EQA}
	\gm/\muc - \sqrt{\dimp/\muc}
	&=&
	\muc^{-1} \gm \bigl( 1 - \sqrt{\muc \alp^{2}} \bigr)
	=
	\gm \, (1 + \alp) \, \bigl\{  1 - \sqrt{\alp^{2}/(1 + \alp)} \bigr\} .
\label{gmcmcm1a}
\end{EQA}
For  \( \xxc = \gm^{2}/4 \), it holds 
\begin{EQA}
	\dimp + 2 \sqrt{\dimp \xxc} + 2 \xxc
	& = &
	\dimp + \sqrt{\dimp \gm^{2}} + \gm^{2}/2
	= 
	\gm^{2} \bigl( \alp^{2} + \alp + 1/2 \bigr) .
\label{d22xdsdg2a}
\end{EQA}
Direct calculus shows that for \( \alp \leq 0.3 \) one can bound  
\begin{EQA}
	\alp^{2} + \alp + 1/2
	& \leq &
	(1 + \alp)^{2} \Bigl\{ 1 - \sqrt{\alp^{2}/(1 + \alp)} \Bigr\}^{2}
\label{a21a21a21a}
\end{EQA}
and this proves \eqref{dimp2x2xgm2}.
\end{proof}

We conclude from this lemma, \eqref{logPmunu0} and \eqref{mudp2a} that  
\begin{EQA}
    \P\bigl( \| \xiv \|^{2} > \dimp + 2 \sqrt{\dimp \xx} + 2 \xx, \| \xiv \| \le \zqc \bigr)
    & \leq &
    2 \ex^{-\xx} .
\label{logPmunu}
\end{EQA}
If \eqref{expgamgm} holds with \( \gm = \infty \), then we are back in the (sub-)Gaussian case 
with \( \zqc = \infty \).
In the non-Gaussian case with a finite \( \gm \), we have to accompany the moderate
deviation bound with a large deviation bound
\( \P\bigl( \| \xiv \| > \zq \bigr) \) for \( \zq \ge \zqc \).
This is done by combining the bound \eqref{Eexp2xi} with
the standard slicing arguments.

\begin{lemma}
\label{Lexpxi2}
Define \( \gmc = \muc \zqc \); see \eqref{mucdef1alp}.
It holds for \( \zq \ge \zqc \)
\begin{EQA}
\label{Pexp2xi}
    \P\bigl( \| \xiv \| > \zq \bigr)
    & \le &
    8.4 (1 - \gmc/\zq)^{-\dimp/2} \exp\bigl( - \gmc \zq/2 \bigr)
    \\
    & \le &
    8.4 \exp\bigl\{ - \xxc -  \gmc (\zq - \zqc)/2 \bigr\}.
\label{Pexp2xiy}
\end{EQA}
\end{lemma}

\begin{proof}
For a fixed \( \zq \geq \zqc \),
consider the growing sequence \( (\yy_{k}) \) with  \( \yy_{1} = \zq \)
and 
\begin{EQA}
	\yy_{k+1} 
	&=& 
	\zq + k / \gmc.
\label{ykp1ykg}
\end{EQA}
Define also \( \mu_{k} = \gmc/ \yy_{k} \).
Then the sequence \( (\mu_{k}) \) is decreasing,
in particular, \( \mu_{k} \le \mu_{1} = \gmc / \zq \leq \muc \).
Obviously
\begin{EQA}
    \P\bigl( \| \xiv \| > \zq \bigr)
    &=&
    \sum_{k=1}^{\infty}
        \P\bigl( \| \xiv \| > \yy_{k}, \| \xiv \| \le \yy_{k+1}
        \bigr).
\label{zzkkp1}
\end{EQA}
Now we try to evaluate every slicing probability in this expression.
We use that
\begin{EQA}
    \mu_{k+1} \yy_{k}^{2}
    &=&
    \frac{(\gmc \zq  + k - 1)^{2}}{\gmc \zq  + k}
    \ge
    \gmc \zq + k - 2 .
\label{nukk1}
\end{EQA}
Lemma~\ref{Lmucyyc} implies \( \gm - \sqrt{\muc \dimp} \geq \muc \zqc = \gmc \).
This yields
\( \gm/\mu_{k} - \sqrt{\dimp / \mu_{k}} \ge \yy_{k} \) because
\begin{EQA}[c]
    \gm/\mu_{k} - \sqrt{\dimp / \mu_{k}} - \yy_{k}
    =
    \mu_{k}^{-1} (\gm - \sqrt{\mu_{k} \dimp} - \gmc)
    \geq 
    \mu_{k}^{-1} (\gm - \sqrt{\muc \dimp} - \gmc)
    \ge 0 .
\label{gmgmmyy}
\end{EQA}
Hence by \eqref{Eexp2xi}
\begin{EQA}
    \P\Bigl( \| \xiv \| > \zq \Bigr)
    & = &
    \sum_{k=1}^{\infty}
     \P\Bigl(
        \| \xiv \| > \yy_{k},
        \| \xiv \| \le \yy_{k+1}
     \Bigr)
    \\
    & \le &
    \sum_{k=1}^{\infty} \exp\Bigl( - \frac{\mu_{k+1} \yy_{k}^{2}}{2} \Bigr)
    \E \exp\Bigl( \frac{\mu_{k+1} \| \xiv \|^{2}}{2} \Bigr)
     \Ind\biggl( \| \xiv \| \le \frac{\gm}{\mu_{k+1}} - \sqrt{\frac{\dimp}{\mu_{k+1}}} \biggr)
    \\
    & \le &
    \sum_{k=1}^{\infty} 2 \bigl( 1 - \mu_{k+1} \bigr)^{-\dimp/2}
    \exp\Bigl( - \frac{\mu_{k+1} \yy_{k}^{2}}{2} \Bigr)
    \\
    & \le &
    2 \bigl( 1 - \mu_{1} \bigr)^{- \dimp/2}
    \sum_{k=1}^{\infty} \exp\Bigl( - \frac{\gmc \zq + k - 2}{2} \Bigr)
    \\
    & = &
    2 \ex^{1/2} (1 - \ex^{-1/2})^{-1}
    (1 - \mu_{1})^{-\dimp/2} \exp\bigl( - \gmc \zq/2 \bigr)
    \\
    & \le &
    8.4 (1 - \gmc/\zq)^{-\dimp/2} \exp\bigl( - \gmc \zq/2 \bigr)
\label{Psumnuk}
\end{EQA}
and the assertion \eqref{Pexp2xi} follows.
For  \( \zq = \zqc \), it holds by \eqref{mudp2a}
\begin{EQA}[c]
    \gmc \zqc + \dimp \log(1 - \muc)
    =
    \muc \zqc^{2} + \dimp \log(1 - \muc)
    \geq 
    2 \xxc
\label{gmcyyc2xxc}
\end{EQA}
and \eqref{Pexp2xi} implies
\( \P\bigl( \| \xiv \| > \zqc \bigr) \le 8.4 \exp(- \xxc) \).
Now observe that the function
\( f(\zq) = \gmc \zq/2 + (\dimp/2) \log \bigl( 1 - \gmc/\zq \bigr) \)
fulfills \( f(\zqc) = \xxc \) and \( f'(\zq) \ge \gmc/2 \) yielding
\( f(\zq) \ge \xxc + \gmc (\zq - \yyd)/2 \).
This implies \eqref{Pexp2xiy}.
\end{proof}

Now we can conclude that for \( \xx \geq \xxc \), the choice
\begin{EQA}
	\zq 
	=
	\zq(\xx)
	& = &
	2 \gmc^{-1} (\xx - \xxc) + \zqc 
\label{yy2gm1xmxd}
\end{EQA}
implies 
\begin{EQA}
	\P\bigl( \| \xiv \| > \zq(\xx) \bigr)
	& \leq &
	8.4 \ex^{-\xx} .
\label{Pxiy2exf}
\end{EQA}
The statement of the theorem is obtained by a simple combination of \eqref{mudp2a} and 
\eqref{Pxiy2exf}.


\Subsection{A deviation bound for a general non-Gaussian quadratic form }
This section presents a bound for a quadratic form \( \xiv^{\T} \BB \xiv \), where 
\( \xiv \) satisfies \eqref{expgamgm} and \( \BB \) is a given symmetric positive 
\( \dimp \times \dimp \) matrix.
%
Define 
\begin{EQA}[c]
    \dimB
    \eqdef
    \tr \bigl( \BB \bigr) ,
    \qquad 
    \vpB^{2}
    \eqdef
    \tr(\BB^{2}) ,
    \qquad
    \lambdaB \eqdef \lambda_{\max}\bigl( \BB \bigr). 
\label{BBrddB}
\end{EQA}   
For ease of presentation, 
suppose that \( 0.3 \gm \ge \sqrt{\dimA} \) so that 
\( \alp = \sqrt{\dimA / \gmb^{2}} \leq 0.3 \).
The other case only changes the constants in the inequalities. 
Define also
\begin{EQA}
	\xxc
	& \eqdef &
	\gm^{2}/4,
	\\
	\zqc^{2}
	& \eqdef &
	\dimA + \vp \gm + \supA \gmb^{2}/2 ,
	\\
	\gmc
	& \eqdef &
	\frac{ \sqrt{\dimA/\supA + \gm \vp / \supA + \gm^{2}/2}}{1 + \vp / (\supA \gmb)} .
\label{xxcgm24mucyyc2B}
\end{EQA}

\begin{theorem}
\label{LLbrevelocroB}   
Let \eqref{expgamgm} hold and 
\( 0.3 \gm \ge \sqrt{\dimA/\supA} \).
Then
for each \( \xx > 0 \)
\begin{EQA}
    \P\bigl( \| \BB^{1/2} \xiv \| \ge \zq(\BB,\xx) \bigr)
    & \le &
    2 \ex^{-\xx} + 8.4 \ex^{-\xxc} \Ind(\xx < \xxc) ,
\label{PxivbzzBBroB}
\end{EQA}    
where \( \zq(\BB,\xx) \) is defined by
\begin{EQA}
\label{PzzxxpBroB}
    \zq(\BB,\xx)
    & \eqdef &
    \begin{cases}
        \sqrt{ \dimA + 2 \vp \xx^{1/2} + 2 \supA \xx }, &  \xx \le \xxc, \\
        \zqc + 2 \lambdaB (\xx - \xxc)/\gmc , & \xx > \xxc.
    \end{cases}
\label{zzxxppdBlroB}
\end{EQA}    
\end{theorem}

Similarly to the Gaussian case, the upper quantile 
\( \zq(\BB,\xx) = \sqrt{ \dimA + 2 \vp \xx^{1/2} + 2 \supA \xx } \) can be upper bounded 
by \( \sqrt{\dimA} + \sqrt{2 \supA \xx} \):
\begin{EQA}
\label{PzzxxpBroBu}
    \zq(\BB,\xx)
    & \leq &
    \begin{cases}
        \sqrt{\dimA} + \sqrt{2 \supA \xx}, &  \xx \le \xxc, \\
        \zqc + 2 \lambdaB (\xx - \xxc)/\gmc , & \xx > \xxc.
    \end{cases}
\end{EQA}    

The main steps of the proof are similar to the proof of Theorem~\ref{LLbrevelocro}.
Normalization by \( \supA \) reduces the statement to the case \( \supA = 1 \)
which we assume below.
Moreover, the standard change-of-basis arguments allow us 
to reduce the problem to the case of a diagonal matrix 
\( \BB = \diag\bigl( a_{1},\ldots,a_{\dimp} \bigr) \), where 
\( 1 = a_{1} \ge a_{2} \ge \ldots\ge a_{\dimp} > 0 \).
Note that \( \dimA = a_{1} + \ldots + a_{\dimp} \) and 
\( \vp^{2} = a_{1}^{2} + \ldots + a_{\dimp}^{2} \).

\begin{lemma}
\label{Lexpxig} 
Suppose \eqref{expgamgm} and \( \| \BB \|_{\oper} = 1 \).
For any \( \mu < 1 \) with
\( \gm^{2}/\mu \ge \dimA \), it holds 
\begin{EQA}
\label{Eexp2xig}
    \E \exp\bigl( {\mu \| \BB^{1/2} \xiv \|^{2}}/{2} \bigr)
        \Ind\bigl( \| \BB \xiv \| \le \gm/\mu - \sqrt{\dimA/\mu} \bigr)
    & \le &
    2 {\det(\Id_{\dimp} - \mu \BB)^{-1/2}} .
\end{EQA}    
\end{lemma}

\begin{proof}
With \( c_{\dimp}(\BB) = \bigl( 2 \pi \bigr)^{-\dimp/2} \det(\BB^{-1/2}) \)
\begin{EQA}
    && \nquad
    c_{p}(\BB) \int \exp\Bigl( \gammav^{\T} \xiv - \frac{1}{2 \mu} \| \BB^{-1/2} \gammav \|^{2} \Bigr)
        \Ind(\| \gammav \| \le \gm) d\gammav
    \\
    &=&
    c_{p}(\BB) \exp\Bigl( \frac{\mu \| \BB^{1/2} \xiv \|^{2}}{2} \Bigr)
    \int \exp\Bigl( 
        - \frac{1}{2}  \bigl\| \mu^{1/2} \BB^{1/2} \xiv - \mu^{-1/2} \BB^{-1/2} \gammav \bigr\|^{2} 
    \Bigr) \Ind(\| \gammav \| \le \gm) d\gammav    
    \\
    &=&
    \mu^{\dimp/2} \exp\Bigl( \frac{\mu \| \BB^{1/2} \xiv \|^{2}}{2} \Bigr) 
    \P_{\xiv}\bigl( 
        \| \mu^{-1/2} \BB^{1/2} \varepsilonv + \BB^{1/2} \xiv \| \le \gm / \mu
    \bigr),
\label{intggvvg}
\end{EQA}
where \( \varepsilonv \) denotes a standard normal vector in \( \R^{\dimp} \)
and \( \P_{\xiv} \) means the conditional probability given \( \xiv \).
Moreover, for any \( \uv \in \R^{\dimp} \) and \( \rr \ge \dimA^{1/2} + \| \uv \| \), 
it holds in view of \( \P \bigl( \| \BB^{1/2} \varepsilonv \|^{2} > \dimA \bigr) \le 1/2 \)
\begin{EQA}
    \P\bigl( \| \BB^{1/2} \varepsilonv - \uv \| \le \rr \bigr)
    & \ge &
    \P\bigl( \| \BB^{1/2} \varepsilonv \| \le \sqrt{\dimA} \bigr)
    \ge 
    1/2 .
\label{Arepsv}
\end{EQA}    
This implies
\begin{EQA}
    && 
    \nquad
    \exp\Bigl( \mu \| \BB^{1/2} \xiv \|^{2} / 2 \Bigr) 
        \Ind\bigl( \| \BB \xiv \| \le \gm / \mu - \sqrt{\dimA/\mu} \bigr)
    \\
    & \le &
    2 \mu^{- \dimp/2} c_{p}(\BB)
    \int \exp\Bigl( \gammav^{\T} \xiv - \frac{1}{2 \mu} \| \BB^{-1/2} \gammav \|^{2} \Bigr)
        \Ind(\| \gammav \| \le \gm) d\gammav .
\label{expxiv1cpg}
\end{EQA}    
Further, by \eqref{expgamgm}
\begin{EQA}
    && 
    \nquad
    c_{p}(\BB) \E \int \exp\Bigl( 
    	\gammav^{\T} \xiv - \frac{1}{2\mu} \| \BB^{-1/2} \gammav \|^{2} 
    \Bigr)
        \Ind(\| \gammav \| \le \gm) d\gammav
    \\
    & \le &
    c_{p}(\BB) \int \exp\Bigl( 
        \frac{\| \gammav \|^{2}}{2} - \frac{1}{2\mu} \| \BB^{-1/2} \gammav \|^{2} 
    \Bigr) d\gammav
    \\
    & \le &
    \det(\BB^{-1/2}) \det(\mu^{-1} \BB^{-1} - \Id_{\dimp})^{-1/2}
    =
    \mu^{p/2} \det( \Id_{\dimp} - \mu \BB)^{-1/2}
\label{nununug}
\end{EQA}    
and \eqref{Eexp2xig} follows.
\end{proof}

Now we evaluate the probability \( \P\bigl( \| \BB^{1/2} \xiv \| > \yy \bigr) \) for moderate 
values of \( \yy \).
Given \( \xx \leq \xxc = \gm^{2}/4 \), define 
\begin{EQA}
	\mu 
	&=& 
	\mu(\xx)
	=
	\frac{1}{1 + 0.5 \vp \xx^{-1/2}} \, ,
\label{mumuxxgm2B}
	\\
	\muc
	& \eqdef &
	\frac{1}{1 + 0.5 \vp \, \xxc^{-1/2}} 
	=
	\frac{1}{1 + \vp / \gmb} \, .
\label{mucdef1alpB}
\end{EQA}
Obviously \( \mu \leq \muc \).
Now we obtain similarly to the Gaussian case in Lemma~\ref{Lmuvpxx}  
for \( u = 2 \vp \sqrt{\xx} + 2 \xx \)
\begin{EQA}
    && \nquad
    \P\Bigl( \| \BB^{1/2} \xiv \|^{2} > \dimA + u, \, \| \xiv \| \le \gm/\mu - \sqrt{\dimA/\mu} \Bigr)
    \\
    & \le &
    \exp\Bigl\{ - \frac{\mu (\dimA + u)}{2} \Bigr\}
    \E\exp \Bigl( \frac{\mu \| \xiv \|^{2}}{2} \Bigr)
    \Ind\Bigl( \| \xiv \| \le \gm/\mu - \sqrt{\dimA/\mu} \Bigr)
    \\
    & \le &
    2 \exp\Bigl\{ - \frac{1}{2} \bigl[ \mu (\dimA + u) 
    	- \log\det( \Id_{\dimp} - \mu \BB) \bigr]
    \Bigr\} 
\label{logPmunu0B}
\end{EQA}
and by \eqref{mux2xv4x12}, it holds for \( \mu \) from \eqref{mumuxxgm2B}
\begin{EQA}
	\mu (\dimA + 2 \vp \sqrt{\xx} + 2 \xx) + \log\det( \Id_{\dimp} - \mu \BB)
	& \geq &
	2 \xx .
\label{mudp2aB}
\end{EQA}
Now we show that the constraint \( \| \xiv \| \le \gm/\mu - \sqrt{\dimA/\mu} \) in \eqref{logPmunu0B}
can be replaced by the inequality \( \| \xiv \| \leq \zqc \). 
Indeed, the definition implies \( \mu \leq \muc \) for \( \xx \leq \xxc \) and
\begin{EQA}
	\dimA + 2 \vp \, \sqrt{\xx} + 2 \xx
	& \leq &
	\dimA + 2 \vp \sqrt{\xxc} + 2 \xxc ,
	\\ 
	\gm/\mu - \sqrt{\dimA/\mu} 
	& \geq & 
	\gm/\muc - \sqrt{\dimA/\muc} .
\label{d2vx2cc}
\end{EQA}
It remains to show that 
%
\begin{EQA}
	\dimA + 2 \vp \sqrt{\xxc} + 2 \xxc
	& \leq &
	\bigl( \gm/\muc - \sqrt{\dimA/\muc} \bigr)^{2} .
\label{dimp2x2xgm2B}
\end{EQA}
Denote \( \alp^{2} = \dimA / \gm^{2} \). 
By \( \vp^{2} \leq \dimA \) and \( \xxc = \gm^{2}/4 \), it holds
\( \muc^{-1} = 1 + 0.5 \vp \xxc^{-1/2} \leq 1 + \alp \) and
\begin{EQA}
	\gm/\muc - \sqrt{\dimA/\muc}
	&=&
	\muc^{-1} \gm \bigl( 1 - \sqrt{\muc \alp^{2}} \bigr)
	\geq 
	\gm \, (1 + \alp) \, \bigl\{  1 - \sqrt{\alp^{2}/(1 + \alp)} \bigr\} .
\label{gmcmcm1aB}
\end{EQA}
Also in a similar way
\begin{EQA}
	\dimA + 2 \vp \sqrt{\xxc} + 2 \xxc
	& \leq &
	\dimA + \sqrt{\dimA \gm^{2}} + \gm^{2}/2
	=  
	\gm^{2} \bigl( \alp^{2} + \alp + 1/2 \bigr) .
\label{d22xdsdg2aB}
\end{EQA}
This and \eqref{a21a21a21a} prove \eqref{dimp2x2xgm2B} yielding
\begin{EQA}
    \P\bigl( \| \BB^{1/2} \xiv \|^{2} > \dimA + 2 \vp \sqrt{\xx} + 2 \xx, \| \xiv \| \le \zqc \bigr)
    & \leq &
    2 \ex^{-\xx} .
\label{logPmunuB}
\end{EQA}
The large deviation probability \( \P\bigl( \| \BB^{1/2} \xiv \| > \yy \bigr) \) for 
\( \yy > \zqc \) can be bounded as in the case \( \BB = \Id_{\dimp} \).

\begin{lemma}
\label{Lexpxi2B}
Define \( \gmc = \muc \zqc \); see \eqref{mucdef1alpB}.
It holds for \( \zq \ge \zqc \)
\begin{EQA}
\label{Pexp2xiB}
    \P\bigl( \| \BB^{1/2} \xiv \| > \zq \bigr)
    & \le &
    8.4 (1 - \gmc/\zq)^{-\dimp/2} \exp\bigl( - \gmc \zq/2 \bigr)
    \\
    & \le &
    8.4 \exp\bigl\{ - \xxc -  \gmc (\zq - \zqc)/2 \bigr\}.
\label{Pexp2xiyB}
\end{EQA}
\end{lemma}

\begin{proof}
The arguments from the case \( \BB \equiv \Id_{\dimp} \) apply without changes.
\end{proof}


\Chapter{Deviation bounds for random processes}
\label{Chgempir}
\label{SlocalBern}
This chapter presents some general results of the theory of empirical processes.
We assume some exponential moment conditions on the increments of the process which 
allow to apply the well developed chaining 
arguments in Orlicz spaces; see e.g. \cite{VW1996}, Chapter~2.2.
\ifbook{}{
We, however, follow the more recent approach inspired by the notions of generic chaining and 
majorizing measures due to M. Talagrand;
The chaining arguments are replaced by the \emph{pilling} device; 
see e.g. \cite{Ta1996,Ta2001,Ta2005}. 
The results are close to that of \cite{Be2006}. 
}
We state the results in a slightly different form and present an independent 
and self-contained proof.

The first result states a bound for local fluctuations of the process 
\( \UP(\ups) \) given on a metric space \( \Ups \).
Then this result will be used for bounding the maximum of the negatively drifted 
process \( \UP(\ups,\upss) \eqdef \UP(\ups) - \UP(\upss) \) over a  
vicinity \( \Upss(\rups) \) of the central point \( \upss \).
The behavior of \( \UP(\ups) \) outside of the local central set \( \Upss(\rups) \) is 
described using the \emph{upper function} method. 
Namely, we construct a deterministic function \( f(\rr,\rups) \) ensuring that 
with probability at least \( 1 - \ex^{-\xx} \) it holds on a dominating set of probability
at least \( 1 - \ex^{-\xx} \) that
\( \UP(\ups,\upss) - f\bigl( \dist(\ups,\upss),\rups \bigr) < 0 \) 
for all \( \ups \not\in \Upss(\rups) \).

%
\Section{Chaining and covering numbers}

An important step in the whole construction is an exponential bound on the  
maximum of a random process \( \UP(\ups) \) under the exponential moment conditions on its 
increments. 
Let \( \dist(\ups,\upsc) \) be a semi-distance on \( \Ups \).
We suppose the following condition to hold: 

\begin{description}
\item[\( \bb{(\CS{d})} \)\label{CSdref}]
    \textit{
    There exist \( \gmb > 0 \), \( \rups > 0 \), 
    \( \nunu \ge 1 \), 
    such that for any  \( \lambda \le \gmb \) and \( \ups,\upsc \in \Ups \)
    with \( \dist(\ups,\upsc) \le \rups \)
    }
\begin{EQA}[c]
\label{ExpboundUP}
    \log \E \exp \biggl\{ 
        \lambda \frac{\UP(\ups) - \UP(\upsc)}{\dist(\ups,\upsc)} 
    \biggr\}
    \le  
    \nunu^{2} \lambda^{2}/2 .
\end{EQA}
\end{description}

By \( \B_{\rr}(\ups) \) we denote the \( \dist \)-ball centered at \( \ups \) of radius 
\( \rr \):
\begin{EQA}
	\B_{\rr}(\ups) 
	& \eqdef &
	\{ \upsc \in \Ups \colon \dist(\ups,\upsc) \le \rr \} .
\label{Bkupsd}
\end{EQA}
Let \( \Upsd \) be a subset of a ball in \( \Ups \) with center at \( \upss \) 
and radius \( \rups \),
and let  
a sequence \( \rr_{k} \) be fixed with \( \rr_{k} = \rr_{0} 2^{-k} \).

For each \( k \), by \( \MM_{k} \) we denote a \( \rr_{k} \)-net in \( \Upsd \), so that 
\begin{EQA}
	\Upsd
	& \subseteq &
	\bigcup_{\ups \in \MM_{k}} \B_{\rr_{k}}(\ups) .
\label{UpsdkBku}
\end{EQA}
Let also \( \Pi_{k} \ups \) be the closest to \( \ups \) point from \( \MM_{k} \), 
so that \( \dist(\ups,\Pi_{k} \ups) \leq \rr_{k} \). 
We assume that \( \MM_{0} \) consists of one point \( \upss \), that is,
\( \Pi_{0} \ups = \upss \). 
Let \( \NN_{k} \eqdef |\MM_{k}| \) denote the cardinality of \( \MM_{k} \).
Finally set
\( c_{k} = 2^{-k} \) for \( k \ge 1 \), and define
the values \( \entrlq(\Upsd) \) and \( \entrlg(\Upsd) \) by 
\begin{EQ}[rcl]
    \entrlq(\Upsd)
    & \eqdef &
    \sum_{k=1}^{\infty} c_{k} \sqrt{2 \log(2 \NN_{k})}
    =
    \sum_{k=1}^{\infty} 2^{-k} \sqrt{2 \log(2 \NN_{k})} ,
    \\
    \entrlg(\Upsd)
    & \eqdef &
    \sum_{k=1}^{\infty} 2 c_{k} \log(2 \NN_{k})
    =
    \sum_{k=1}^{\infty} 2^{-k+1} \log(2 \NN_{k})    .
\label{entrldefch}
\end{EQ}    
By the Cauchy-Schwartz inequality
\( \entrlq^{2}(\Upsd) \leq \entrlg(\Upsd) \).
The inverse relation is not generally true and one can build some examples with 
\( \entrlq(\Upsd) \) finite and \( \entrlg(\Upsd) \) infinite.
If the process \( \UP(\ups) \) is sub-Gaussian and \nameref{CSdref} is fulfilled with 
\( \gm = \infty \), then one can only operate with \( \entrlq(\Upsd) \) which is 
equivalent to the Dudley integral\ifbook{.}{; see \eqref{Dudlyintch} below.}

\begin{theorem}
\label{TUPUpsdch}
Let \( \UP \) be a separable process  
and \( \Upsd \) be a ball in \( \Ups \) with center \( \upsd \) and radius 
\( \rups \) for the distance \( \dist(\cdot,\cdot) \), 
i.e. \( \dist(\ups,\upsd) \le \rups \) for all \( \ups \in \Upsd \).
If \nameref{CSdref} holds with \( \gm = \infty \) then 
for any \( \xx \geq 1/2 \), it holds with \( \entrlq = \entrlq(\Upsd) \) and 
\( \entrlg = \entrlg(\Upsd) \)
\begin{EQA}
	\P\biggl( 
		\frac{1}{\nunu \rups} \sup_{\ups \in \Upsd} \UP(\ups,\upss)  
		\geq 
		\zzQ(\xx) 
	\biggr)
	& \leq &
	\ex^{-\xx} 
\label{Pczentrl1ch}
\end{EQA}
with 
\begin{EQA}
	\zzQ(\xx)
	& \eqdef &
	2 \entrlq + \sqrt{8 \xx} .
\label{zzQxxgminf}
\end{EQA}
If \nameref{CSdref} holds with \( \gm \leq \infty \), then 
\begin{EQA}
	\P\biggl\{ 
		\frac{1}{\nunu \rups} \sup_{\ups \in \Upsd} \UP(\ups,\upss)
		\geq 
		\zzQ(\xx)  
	\biggr\}
	& \leq &
	\ex^{-\xx} ,
\label{P13nuUPxxch}
\end{EQA}
where \( \zzQ(\xx) \) is given by one of the following rules:
\begin{EQ}[rcl]
	\zzQ(\xx)
	& = &
	2 \entrlq + \sqrt{8 \xx} + 2 \gm^{-1}(\gm^{-2} \xx + 1) \entrlg,
	\\
	\zzQ(\xx)
	&=&
	\begin{cases}
		2 \sqrt{\entrlg + 2\xx} , & 
		\text{ if } {\entrlg + 2 \xx} \leq \gm^{2}, \\
		2\gm^{-1} \xx + \gm^{-1} \entrlg + \gm , & 
		\text{ if } {\entrlg + 2 \xx} > \gm^{2} .		
	\end{cases}
\label{zzxxgfinch}
\label{zzxxgfin}
\end{EQ}
Moreover, the r.v. \( \UPb(\rups) \eqdef \sup_{\ups \in \Upsd} \UP(\ups,\upss) \) fulfills
\begin{EQA}
	\E \UPb(\rups)
	& \leq &
	2 \nunu \rups \, (\entrlq + \entrlg/\gm + 3),
	\\
	\bigl\{ \E |\UPb(\rups)|^{2} \bigr\}^{1/2}
	& \leq &
	2 \nunu \rups \, (\entrlq + \entrlg/\gm + 4) .
\label{ESES234ch}
\end{EQA}
\end{theorem}

\begin{proof}
We start the proof by stating some general facts 
for a convex combinations of sub-exponential r.v.'s \( \zeta_{k} \) such that 
\begin{EQA}[c]
	\log \E \exp (\lambda \zeta_{k}) 
	\leq
	\frac{q_{k}^{2} + \lambda^{2}}{2} \, ,
	\quad
	|\lambda| \leq \gm,
	\quad
	k=0,1,2, \ldots, 
\label{logEexplazek}
\end{EQA}
where \( q_{k} \geq 1 \) are fixed numbers, and \( \gm \) is some positive value or 
infinity.
We aim at bounding a sum \( S \) of the form \( S = \sum_{k} c_{k} \zeta_{k} \) for a sequence of positive weights \( c_{k} \) satisfying \( \sum_{k} c_{k} = 1 \).
We implicitly assume that the numbers \( q_{k} \) grow with \( k \) in a way that 
\( \sum_{k} \exp( - q_{k}) \leq 1 \).
Define
\begin{EQA}
	\QQq \eqdef \sum_{k} c_{k} q_{k} \, ,
	&&
	\quad 
	\QQg
	\eqdef 
	\sum_{k} c_{k} q_{k}^{2} \, .
\label{Q1Q2def}
\end{EQA}

\begin{lemma}
\label{LPsupQ1Q2}
Suppose that random variables \( \zeta_{k} \) follow \eqref{logEexplazek} with \( \gm = \infty \) and \( \sum_{k} \exp( - q_{k}) \leq 1 \). 
Let also \( \sum_{k} c_{k} = 1 \). 
Then it holds for the sum \( S = \sum_{k} c_{k} \zeta_{k} \) 
\begin{EQA}[c]
	\log \E \exp \bigl( S \bigr)
	\leq
	\QQq 
\label{entrlqdef}
\end{EQA}
and for any \( \xx \geq 1/2 \), 
\begin{EQA}
	\P\bigl( S \geq \QQq + \sqrt{2\xx} \bigr)
	& \leq &
	\ex^{-\xx} .
\label{Pczentrl1}
\end{EQA}
If \eqref{logEexplazek} holds for \( \gm < \infty \), then for each \( \lambda > 0 \) 
with \( |\lambda| \leq \gm \)
\begin{EQA}[c]
	\log \E \exp \bigl( \lambda S \bigr)
	\leq 
	\bigl( \QQg + \lambda^{2} \bigr) /2 \, ,
\label{entrl2def}
\end{EQA}
and it holds for \( \xx \geq 1/2 \)
\begin{EQA}
	\P\bigl\{ 
		S \geq \zzQ(\xx)
	\bigr\}
	& \leq &
	\ex^{-\xx} ,
\label{Pczentrl1gml}
\end{EQA}
where \( \zzQ(\xx) \) is given by \eqref{zzxxgfin}.
Moreover, if \( \gm^{2} \geq \QQg + 1 \), then
\begin{EQA}
	\E S
	& \leq &
	\QQq + \QQg/\gm + 3,
	\qquad
	\bigl\{ \E S^{2} \bigr\}^{1/2}
	\leq 
	\QQq + \QQg/\gm + 4.
\label{EsES2QQg}
\end{EQA}
\end{lemma}

\begin{proof}
Consider first the sub-Gaussian case with \( \gm = \infty \).
Define \( \alpha_{k} = c_{k} / q_{k} \).
Obviously \( \sum_{k} \alpha_{k} \leq \sum_{k} c_{k} = 1 \).
By the H\"older inequality and \eqref{logEexplazek}, it holds  
\begin{EQA}
	\log \E \exp \biggl( \sum_{k} c_{k} \zeta_{k} \biggr)
	&=&
	\log \E \exp \biggl( \sum_{k} \alpha_{k} q_{k} \zeta_{k} \biggr)
	\leq
	\sum_{k} \alpha_{k} \log \E \exp\bigl( q_{k} \zeta_{k} \bigr)
	\\
	& \leq & 
	\frac{1}{2} \sum_{k} \alpha_{k} \bigl( q_{k}^{2} + q_{k}^{2} \bigr)
	\leq
	\sum_{k} c_{k} q_{k} \, .
\end{EQA}
Further, by the same arguments, it holds  
\begin{EQA}
	\log \E \exp \bigl( \lambda S \bigr)
	& \leq &
	\sum_{k} c_{k} \log \E \exp\bigl( \lambda \zeta_{k} \bigr)
	\leq 
	\frac{1}{2} \sum_{k} c_{k} \bigl( q_{k}^{2} + \lambda^{2} \bigr)
\end{EQA}
and the assertion \eqref{entrl2def} follows as well.

Let \( \xx \geq 1/2 \) be fixed.
With \( z_{k} = q_{k} + \sqrt{2\xx} \), it follows by \eqref{logEexplazek} for 
\( \lambda_{k} = z_{k} \) in view of 
\( \sum_{k} \ex^{-q_{k}} \leq 1 \) 
\begin{EQA}
	&& \nquad
	\P\biggl( \sum_{k} c_{k} (\zeta_{k} - z_{k}) \geq 0 \biggr)
	\leq 
	\sum_{k} \P\bigl( \zeta_{k} - z_{k} \geq 0 \bigr)
	\leq 
	\sum_{k} \E \exp\bigl\{ \lambda_{k}( \zeta_{k} - z_{k}) \bigr\}
	\\
	& \leq & 
	\sum_{k} \exp\bigl( - \lambda_{k} z_{k} + \lambda_{k}^{2}/2 + q_{k}^{2}/2 \bigr)
	=
	\sum_{k} \exp\bigl( - z_{k}^{2}/2 + q_{k}^{2}/2 \bigr)
	\\
	&=&
	\sum_{k} \exp\bigl(  -\xx - q_{k} \sqrt{2\xx} \bigr)
	\leq
	\ex^{-\xx} .
\label{Psumkckzk}
\end{EQA}
This implies 
\begin{EQA}[c]
	\sum_{k} c_{k} z_{k}
	=
	\sum_{k} c_{k} \bigl( q_{k} + \sqrt{2\xx} \bigr)
	= 
	\QQq + \sqrt{2\xx} \, 
\label{sumckzkxxen}
\end{EQA}
and the assertion \eqref{Pczentrl1} follows.

Now we briefly discuss how the condition \eqref{logEexplazek} can be relaxed 
to the case of a finite \( \gm \).
Suppose that \eqref{logEexplazek} holds for all \( \lambda \leq \gm < \infty \).
Define \( k(\xx) \) as the largest index \( k \), for which 
\( \lambda_{k} = q_{k} + \sqrt{2\xx} \leq \gm \).
For \( k > k(\xx) \), define \( \lambda_{k} = \gm \) and 
\begin{EQA}[c]
	z_{k}
	=
	\frac{\xx + q_{k}}{\gm} + \frac{\gm}{2}  + \frac{q_{k}^{2}}{2\gm} \, .
\label{zkxxqkgm}
\end{EQA}
The above arguments yield for \( k > k(\xx) \)
\begin{EQA}[c]
	\P\bigl( \zeta_{k} \geq z_{k} \bigr)
	\leq
	\exp\biggl( 
		- \gm z_{k} + \frac{1}{2} \bigl( q_{k}^{2} + \gm^{2} \bigr) 
	\biggr)
	=
	\exp( - \xx - q_{k}) .
\end{EQA}
This and \eqref{Psumkckzk} yield
\begin{EQA}
	\sum_{k} \P\bigl( \zeta_{k} \geq z_{k} \bigr)
	& \leq &
	\sum_{k \leq k(\xx)} \exp\bigl(  -\xx - q_{k} \sqrt{2\xx} \bigr)
	+ \sum_{k > k(\xx)} \exp ( - \xx - q_{k})
	\\
	& \leq &
	\sum_{k} \exp ( - \xx - q_{k})
	\leq
	\ex^{-\xx} .
\end{EQA}
Further, as \( q_{k} > \gm \) for \( k > k(\xx) \), it follows from the definition \eqref{zkxxqkgm}
\begin{EQA}
	\sum_{k > k(\xx)} c_{k} z_{k}
	&=&
	\frac{1}{\gm} \sum_{k > k(\xx)} c_{k} (\xx + q_{k}) 
	+ \frac{\gm}{2} \sum_{k > k(\xx)} c_{k} 
	+ \frac{1}{2\gm} \sum_{k > k(\xx)} c_{k} q_{k}^{2}
	\\
	& \leq &
	\frac{1}{\gm} \sum_{k > k(\xx)} c_{k} q_{k} 
	+ \biggl( \frac{\xx}{\gm^{3}} + \frac{1}{\gm} \biggr) 
		\sum_{k > k(\xx)} c_{k} q_{k}^{2}  .
\end{EQA}
This and \eqref{sumckzkxxen} imply due to \( \gm \geq 1 \)
\begin{EQA}[c]
	\sum_{k} c_{k} z_{k}
	\leq
	\sum_{k} c_{k} q_{k} + \biggl( \frac{\xx}{\gm^{3}} + \frac{1}{\gm} \biggr) 
		\sum_{k} c_{k} q_{k}^{2} + \sqrt{2\xx}
	\leq
	\QQq + \biggl( \frac{\xx}{\gm^{3}} + \frac{1}{\gm} \biggr) \QQg + \sqrt{2\xx} \, .
\end{EQA}
In particular, if \( \xx \leq \gm^{2} \), then
\begin{EQA}[c]
	\sum_{k} c_{k} z_{k}
	\leq
	\QQq + \frac{2}{\gm} \QQg + \sqrt{2\xx} \, .
\end{EQA}
Now \eqref{Pczentrl1gml} with \( \zz(\xx) = \QQq + \sqrt{2\xx} + \gm^{-1}(\gm^{-2} \xx + 1) \QQg \) follows similarly to \eqref{Pczentrl1}.
Further, if \( \zz(\xx) = \sqrt{\QQg + 2 \xx} \leq \gm \), then \eqref{entrl2def} with
\( \lambda = \zz(\xx) \) and the exponential Chebyshev inequality implies again
\begin{EQA}
	\P\bigl( S \geq \zz(\xx) \bigr)
	& \leq &
	\exp\Bigl( - \lambda \zz(\xx) + \frac{\QQg + \lambda^{2}}{2} \Bigr)
	=
	\exp\Bigl( \frac{- \zz^{2}(\xx) + \QQg}{2} \Bigr)
	=
	\exp( - \xx) .
\label{PSzzxxQQq}
\end{EQA}
Similarly one can check the case with \( \lambda = \gm \) and 
\( \zz(\xx) = \xx/\gm + \bigl( \QQg/\gm + \gm \bigr)/2 > \gm \).

To bound the moments of \( S \), we apply the following technical result:
if 
\begin{EQA}
	\P\bigl( S \geq \zz(\xx) \bigr)
	& \leq &
	\ex^{-\xx}
\label{¥}
\end{EQA}
for all \( \xx \geq \xx_{0} \) and 
if \( \zz(\cdot) \) is absolutely continuous, then 
\begin{EQA}
\label{EXzzdxx}
	\E S 
	& \leq &
	\zz(\xx_{0}) + \int_{\xx_{0}}^{\infty} \zz'(\xx) \ex^{-\xx} d\xx,
	\\
	\E S^{2}
	& \leq &
	\zz^{2}(\xx_{0}) + 2 \int_{\xx_{0}}^{\infty} \zz(\xx) \zz'(\xx) \ex^{-\xx} d\xx .
\label{EX2zzdxx}
\end{EQA}
For \( \zz(\xx) = \QQq + \sqrt{2\xx} + \gm^{-1}(\gm^{-2} \xx + 1) \QQg \), it holds
\( \zz'(\xx) \leq 1 + \gm^{-3} \).
In view of \( \gm^{2} \geq \QQg + 1 \)   
\begin{EQA}
	\E S
	& \leq &
	\QQq + 1 + (\QQg + 1/2)/\gm + \int_{1/2}^{\infty} (1 + \gm^{-3}) \ex^{-\xx} d\xx
	\leq 
	\QQq + \QQg/\gm + 3 .
\label{ESQQg13}
\end{EQA}
Similarly one can bound 
\begin{EQA}
	\E S^{2}
	& \leq &
	(\QQq + \QQg/\gm + 3/2)^{2} 
	+ 2 \int_{1/2}^{\infty} \bigl( \frac{1}{\sqrt{2\xx}} + \gm^{-3} \bigr) \zz(\xx) \ex^{-\xx} d\xx
	\leq  
	(\QQq + \QQg/\gm + 4)^{2}
\label{ES2QQg3}
\end{EQA}
as required.
\end{proof}

Now we show how the statement of the theorem can be reduced to the bounds of 
Lemma~\ref{LPsupQ1Q2}.
Denote for \( i < k \) by \( \Pi_{i}^{k} \) the product 
\( \Pi_{i}^{k} = \Pi_{i} \Pi_{i+1} \ldots \Pi_{k} \).
As \( \Pi_{0} \ups \equiv \upss \),
the telescopic sum devices yields
\begin{EQA}
	\bigl| \UP(\Pi_{k} \ups) - \UP(\upss) \bigr|
	& \leq &
	\sum_{i=1}^{k} \bigl| 
		\UP\bigl( \Pi_{i-1}^{k} \ups \bigr) - \UP\bigl( \Pi_{i}^{k} \ups \bigr) 
	\bigr| \, .
\label{UPuUPuss1k}
\end{EQA}
Separability of \( \UP(\cdot) \) implies that 
\begin{EQA}
	\lim_{k \to \infty} \UP(\Pi_{k} \ups)
	&=&
	\UP(\ups). 
\label{limkPiups}
\end{EQA}
Therefore, it holds for any \( \ups \in \Upsd \)
\begin{EQA}
	\bigl| \UP(\ups) - \UP(\upss) \bigr|
	&=&
	\lim_{k \to \infty} \bigl| \UP(\Pi_{k} \ups) - \UP(\upss) \bigr|
	\leq 
	\sum_{k=1}^{\infty} \xis_{k} \, ,
\label{UPuUPusch}
\end{EQA}
where
\begin{EQA}
	\xis_{k}
	& \eqdef &
	\max_{\ups \in \MM_{k}} \bigl| \UP(\ups) - \UP(\Pi_{k-1} \ups) \bigr| .
\label{xisidefch}
\end{EQA}
For each \( \ups \in \MM_{k} \), it holds
\( \dist(\ups,\Pi_{k-1} \ups) \le \rr_{k-1} \) and
\begin{EQA}[c]
    \bigl| \UP(\ups) - \UP(\Pi_{k-1} \ups) \bigr|
    \le 
    \rr_{k-1} \frac{\bigl| \UP(\ups) - \UP(\Pi_{k-1} \ups) \bigr|}{\dist(\ups,\Pi_{k-1} \ups)} .    
\label{UPscd}
\end{EQA}    
This implies 
by the Jensen inequality and \nameref{CSdref} in view of \( e^{|x|} \le e^{x} + e^{-x} \)
for each \( k \geq 1 \) and \( |\lambda| \leq \gmb \)
\begin{EQA}
	\E \exp \Bigl( \frac{\lambda}{\rr_{k-1}} \xis_{k} \Bigr)
    & \le &
    2 \sum_{\ups \in \MM_{k}} \E \exp \Bigl(
		\lambda	\frac{\bigl| \UP(\ups) - \UP(\Pi_{k-1} \ups) \bigr|}{\dist(\ups,\Pi_{k-1} \ups)} 
	\Bigr)
	\leq 
    2 \NN_{k} \exp(\lambda^{2}/2)  .
    \qquad
\label{liupsdsupch}
\end{EQA}
For \( k\geq 1 \),
define \( q_{k}^{2}/2 = \log(2 \NN_{k}) \), \( c_{k} = 2^{-k} \), and 
\( \zeta_{k} = \xis_{k}/\rr_{k-1} = c_{k}^{-1} \xis_{k} / (2 \rups) \).
Then \eqref{liupsdsupch} implies by \( \rr_{k-1} = 2^{-k+1} \rups \)
\begin{EQA}
	\log \E \exp \bigl( \lambda \zeta_{k} \bigr)
	& \leq &
	\log(2 \NN_{k}) + \lambda^{2}/2 
	=
	\frac{q_{k}^{2} + \lambda^{2}}{2} \, .
\label{logEelzkch}
\end{EQA}
Now we apply Lemma~\ref{LPsupQ1Q2} with \( c_{k} = 2^{-k} \).
By construction 
\begin{EQA}
	\sum_{k=1}^{\infty} c_{k} \zeta_{k}
	&=&
	\frac{1}{2 \rups} \sum_{k=1}^{\infty} \xis_{k}
\label{sumk1intyxisk}
\end{EQA}
and the results follow with \( \QQq = \entrlq(\Upsd) \), \( \QQg = \entrlg(\Upsd) \).
\end{proof}

\ifbook{}{
\Section{Entropy and Dudley's integral}
The quantities \( \entrlq(\Upsd) \) and \( \entrlg(\Upsd) \) from \eqref{entrldefch}
can be upper bounded by integrals over the interval \( [0,1] \).
Let \( \eps \in (0,1) \).
Denote by \( \MM(\rups, \eps \rups) \) an \( \rr_{\eps} \)-net 
in the \( \rups \)-ball \( \Upsd \) for \( \rr_{\eps} = \rups \eps \). 
Obviously \( \rr_{k} = \rr_{\eps_{k}} \) for \( \eps_{k} = 2^{-k} \),
and the cardinality \( \NN(\eps) = |\MM(\eps)| \) monotonously increases with \( \eps \).
This allows to rewrite the deviation bound \eqref{P13nuUPxxch} in a form which only involves 
the cardinality \( \NN(\eps) \).

\begin{theorem}
It holds for the values \( \entrlq(\Upsd) \) and \( \entrlg(\Upsd) \) from \eqref{entrldefch}
\begin{EQA}
	\entrlq(\Upsd)
    & = &
    \sum_{k=1}^{\infty} 2^{-k} \sqrt{2 \log(2 \NN_{k})}
    \leq 
    1.2 + \sqrt{8} \int_{0}^{1} \sqrt{\log \NN(\rups,\eps \rups)} d\eps .
	\\
	\entrlg(\Upsd)
    & = &
    \sum_{k=1}^{\infty} 2^{-k+1} \log(2 \NN_{k})
    \leq 
    1.4 + 4 \int_{0}^{1} \log \NN(\rups,\eps \rups) d\eps .
\label{Dudly2intch}
\label{Dudlyintch}
\end{EQA}
\end{theorem} 

\begin{proof}
\tobedone{}
\end{proof}

The integral in the right hand-side is usually called \emph{Dudley's integral}.
}

\ifbook{}{
\Section{A local bound with generic chaining}
Here we present a slightly different technique which is often 
called the \emph{majorizing measure} and used in the \emph{generic chaining} device; see 
\cite{Ta2005}.
Formulation of the result involves a sigma-finite measure \( \mes \) on the space 
\( \Ups \).
A typical example of choosing \( \mes \) is the Lebesgue measure on \( \R^{\dimp} \).
Let \( \Upsd \) be a subset of \( \Ups \), 
a sequence \( \rr_{k} \) be fixed with \( 2\rr_{0} = \diam(\Upsd) \)
and \( \rr_{k} = \rr_{0} 2^{-k} \).
Let also \( \B_{k}(\ups) \) mean \( \B_{\rr_{k}}(\ups) \), that is, 
the \( \dist \)-ball centered at \( \ups \) of radius 
\( \rr_{k} \) and \( \mes_{k}(\ups) \) denote its \( \mes \)-measure:
\begin{EQA}[c]
    \mes_{k}(\ups)
    \eqdef
    \int_{\B_{k}(\ups)} \mes(d\upsc) 
    =
    \int_{\Upsd} \Ind\bigl( \dist(\ups,\upsc) \le \rr_{k} \bigr) \mes(d\upsc).
\label{meskups}
\end{EQA}  
Denote also 
\begin{EQA}[c]
    \NN_{k}
    \eqdef
    \int_{\Upsd}
    \frac{\mes(d\ups)}{\mesd_{k}(\ups)} ,
    \qquad 
    k \ge 0.
\label{MkUps}
\end{EQA}    
Finally set
\( c_{0} = 1/3 \), \( c_{k} = 2^{-k+1}/3 \) for \( k \ge 1 \), and define
the values \( \entrlq(\Upsd) \) and \( \entrlg(\Upsd) \) by 
\begin{EQ}[rcl]
    \entrgq(\Upsd)
    & \eqdef &
    \sum_{k=0}^{\infty} c_{k} \sqrt{2 \log(2 \NN_{k})}
    =
    \frac{1}{3} \sqrt{2\log(2 \NN_{0})} 
    	+ \frac{2}{3} \sum_{k=1}^{\infty} 2^{-k} \sqrt{2\log(2 \NN_{k})} ,
    \\
    \entrgg(\Upsd)
    & \eqdef &
    2 \sum_{k=0}^{\infty} c_{k} \log(2 \NN_{k})
    =
    \frac{2}{3} \log(2 \NN_{0}) + \frac{4}{3} \sum_{k=1}^{\infty} 2^{-k} \log(2 \NN_{k})    .
\label{entrldef}
\end{EQ}    
By the Cauchy-Schwartz inequality
\( \entrgq^{2}(\Upsd) \leq \entrgg(\Upsd) \).
The inverse relation is not generally true and one can build some examples with 
\( \entrgq(\Upsd) \) finite and \( \entrg(\Upsd) \) infinite.

\begin{theorem}
\label{TUPUpsd}
Let \( \UP \) be a separable process following to \nameref{CSdref}. 
If \( \Upsd \) is a \( \dist \)-ball in \( \Ups \) with the center \( \upsd \) and the radius 
\( \rups \), i.e. \( \dist(\ups,\upsd) \le \rups \) for all \( \ups \in \Upsd \), then 
for any \( \xx \geq 1/2 \), it holds with \( \QQq = \entrgq(\Upsd) \) and 
\( \QQg = \entrgg(\Upsd) \)
\begin{EQA}
	\P\biggl( 
		\frac{1}{3 \nunu \rups} \sup_{\ups \in \Upsd} \UP(\ups,\upss)  
		\geq 
		\QQq + \sqrt{2\xx} 
	\biggr)
	& \leq &
	\ex^{-\xx} .
\label{Pczentrl1g}
\end{EQA}
If \( \gm \leq \infty \), then 
\begin{EQA}
	\P\biggl\{ 
		\frac{1}{3 \nunu \rups} \sup_{\ups \in \Upsd} \UP(\ups,\upss)
		\geq 
		\zzQ(\xx)  
	\biggr\}
	& \leq &
	\ex^{-\xx} ,
\label{P13nuUPxx}
\end{EQA}
where \( \zzQ(\xx) \) is given by \eqref{zzxxgfin}.
Moreover, the r.v. \( \UPb(\rups) \eqdef \sup_{\ups \in \Upsd} \UP(\ups,\upss) \) fulfills
\begin{EQA}
	\E \UPb(\rups)
	& \leq &
	3 \nunu \rups \, (\QQq + \QQg/\gm + 3),
	\\
	\bigl\{ \E |\UPb(\rups)|^{2} \bigr\}^{1/2}
	& \leq &
	3 \nunu \rups \, (\QQq + \QQg/\gm + 4) .
\label{ESES234}
\end{EQA}
%
\end{theorem}

\begin{proof}
A simple change \( \UP(\cdot) \) with \( \nunu^{-1} \UP(\cdot) \) and 
\( \gmb \) with \( \gmd = \nunu \gmb \)
allows to reduce the result to the case with 
\( \nunu = 1 \) which we assume below.
Consider for \( k \ge 0 \) the smoothing operator \( \smooths_{k} \) defined as 
\begin{EQA}[c]
    \smooths_{k} f(\upsd)
    =
    \frac{1}{\mes_{k}(\upsd)} \int_{\B_{k}(\upsd)} f(\ups) \mes(d\ups) .
\label{Skfupsd}
\end{EQA}    
Further, define
\begin{EQA}[c]
    \smooths_{-1} \UP(\ups) 
    \equiv 
    \UP(\upsd)
\label{S0UPups}
\end{EQA} 
so that \( \smooths_{-1} \UP \) is a constant function and 
the same holds for \( \smooths_{k} \smooths_{k-1} \ldots \smooths_{-1} \UP \) with any 
\( k \ge 0 \).
If \( f(\cdot) \le g(\cdot) \) for two non-negative functions \( f \) and \( g \), 
then \( \smooths_{k} f(\cdot) \le \smooths_{k} g(\cdot) \). 
Separability of the process \( \UP \) implies that
\( \lim_{k} \smooths_{k} \UP(\ups) = \UP(\ups) \).
We conclude that 
for each \( \ups \in \Upsd \) 
\begin{EQA}
    \bigl| \UP(\ups) - \UP(\upsd) \bigr|
    &=&
    \lim_{k \to \infty} 
    \bigl| \smooths_{k} \UP(\ups) - \smooths_{k} \ldots \smooths_{-1} \UP(\ups) \bigr|
    \\
    & \le &
    \lim_{k \to \infty} 
    \sum_{i=0}^{k} 
        \bigl| \smooths_{k} \ldots \smooths_{i} (I - \smooths_{i-1}) \UP(\ups) \bigr| 
    \le 
    \sum_{i=0}^{\infty} \xis_{i} \, .
\label{SkS0UP}
\end{EQA}
Here \( \xis_{k} \eqdef \sup_{\ups \in \Upsd} \xi_{k}(\ups) \) for \( k \ge 0 \) 
with
\begin{EQA}[c]
    \xi_{0}(\ups) \equiv |\smooths_{0} \UP(\ups) - \UP(\upsd)|,
    \quad 
    \xi_{k}(\ups) \eqdef |\smooths_{k} (I - \smooths_{k-1}) \UP(\ups)|,
    \quad 
    k \ge 1.
\label{xibkxib1}
\end{EQA}    
For a fixed point \( \upsdc \) and \( k \geq 1 \), it holds
\begin{EQA}[c]
    \xi_{k}(\upsdc)
    \le 
    \frac{1}{\mes_{k}(\upsdc)} \int_{\B_{k}(\upsdc)} 
        \frac{1}{\mes_{k-1}(\ups)} \int_{\B_{k-1}(\ups)} 
        \bigl| \UP(\ups) - \UP(\upsc) \bigr| \mes(d\upsc) \mes(d\ups) .
\label{Kk1Kkk}
\end{EQA}    
For each \( \upsc \in \B_{k-1}(\ups) \), it holds
\( \dist(\ups,\upsc) \le \rr_{k-1} = 2 \rr_{k} \) and
\begin{EQA}[c]
    \bigl| \UP(\ups) - \UP(\upsc) \bigr|
    \le 
    \rr_{k-1} \frac{\bigl| \UP(\ups) - \UP(\upsc) \bigr|}{\dist(\ups,\upsc)} .    
\label{UPscd}
\end{EQA}    
This implies 
for each \( \upsdc \in \Upsd \) and \( k \geq 1 \)
by the Jensen inequality and \eqref{MkUps}
\begin{EQA}
    \exp \Bigl\{ \frac{\lambda}{\rr_{k-1}} \xi_{k}(\upsdc) \Bigr\}
    & \le &
    \int_{\B_{k}(\upsdc)} 
    \biggl( \int_{\B_{k-1}(\ups)} 
        \exp \frac{\lambda \bigl| \UP(\ups) - \UP(\upsc) \bigr|}{\dist(\ups,\upsc)} 
        \frac{\mes(d\upsc)}{\mes_{k-1}(\ups)}  
    \biggr)
    \frac{\mes(d\ups)}{\mes_{k}(\upsdc)}
    \\
    & \le &
    \int_{\Upsd} 
    \biggl( \int_{\B_{k-1}(\ups)} 
        \exp \frac{\lambda \bigl| \UP(\ups) - \UP(\upsc) \bigr|}{\dist(\ups,\upsc)} 
        \frac{\mes(d\upsc)}{\mes_{k-1}(\ups)}  
    \biggr)
    \frac{\mes(d\ups)}{\mesd_{k}(\ups)} .
\label{lkupsdc}
\end{EQA}    
As the right hand-side does not depend on \( \upsdc \), this yields
for \( \xis_{k} \eqdef \sup_{\ups \in \Upsd} \xi_{k}(\ups) \) 
by condition \nameref{CSdref} in view of \( e^{|x|} \le e^{x} + e^{-x} \)
\begin{EQA}
    \E \exp \Bigl( \frac{\lambda}{\rr_{k-1}} \xis_{k} \Bigr)
    & \le &
    \int_{\Upsd} 
    \biggl( \int_{\B_{k-1}(\ups)} 
        \E \exp \frac{\lambda \bigl| \UP(\ups) - \UP(\upsc) \bigr|}{\dist(\ups,\upsc)} 
        \frac{\mes(d\upsc)}{\mes_{k-1}(\ups)}  
    \biggr)
    \frac{\mes(d\ups)}{\mesd_{k}(\ups)} 
    \\
    & \le &
    2 \exp (\lambda^{2}/2) 
    \int_{\Upsd} \biggl( \int_{\B_{k-1}(\ups)} 
        \frac{\mes(d\upsc)}{\mes_{k-1}(\ups)}  
    \biggr)
    \frac{\mes(d\ups)}{\mesd_{k}(\ups)} 
    \\
    &=&
    2 \NN_{k} 
    \exp(\lambda^{2}/2) ,
    \quad k \geq 1.
\label{lkupsdsup}
\end{EQA}
Further, the use of 
\( \dist(\ups,\upsd) \le \rups \) for all \( \ups \in \Upsd \) yields
by \nameref{CSdref}
\begin{EQA}
    \E \exp \Bigl\{ \frac{\lambda}{\rups} |\UP(\ups) - \UP(\upsd)| \Bigr\}
    & \le &
    2 \exp \bigl( \lambda^{2} / 2 \bigr)
\label{expdistr0}
\end{EQA}    
and thus
\begin{EQA}
    \E \exp\Bigl\{ \frac{\lambda}{\rups} |\smooths_{0} \UP(\ups) - \UP(\upsd)| \Bigr\}
    & \le &
    \frac{1}{\mes_{0}(\ups)} \int_{\B_{0}(\ups)} \E \exp \Bigl\{ 
        \frac{\lambda}{\rups} |\UP(\upsc) - \UP(\upsd)| 
    \Bigr\} \mes(d\upsc)
    \\
    & \le &
    \frac{M_{0}}{\mes(\Upsd)} \int_{\Upsd} \E\exp \Bigl\{ 
        \frac{\lambda}{\rups} |\UP(\upsc) - \UP(\upsd)| 
    \Bigr\} \mes(d\upsc). 
\label{ExpboundUPm0} 
\end{EQA}    
This implies by \eqref{expdistr0} for 
\( \xis_{0} \equiv \sup_{\ups \in \Upsd} |\smooths_{0} \UP(\ups) - \UP(\upsd)| \)
\begin{EQA}[c]
    \E \exp\Bigl( \frac{\lambda}{\rups} \xis_{0} \Bigr)
    \le 
    2 M_{0} \exp \bigl( \lambda^{2} / 2 \bigr).
\label{Eexpcircd}
\end{EQA}    
Denote 
\( c_{0} = 1/3 \) and \( c_{k} = \rr_{k-1}/(3\rups) = 2^{-k+1}/3 \) for \( k \ge 1 \).
Then \( \sum_{k=0}^{\infty} c_{k} = 1 \) and the results follow from 
Lemma~\ref{LPsupQ1Q2} below with \( \zeta_{k} = \xis_{k} / r_{k-1} \) and \( q_{k} = \sqrt{2 \NN_{k}} \).
\end{proof}

\Section{Generic chaining with partitioning}
This section presents a slightly modified construction based on the notion of \emph{partition}. 
Let \( \rr_{k} \to 0 \) be a given sequence, and let for each \( k \),
the set \( \Upsd \) is split into a collection \( \partition_{k} \) of subsets \( C \).
A collection of partitions \( \partition_{k} \) is called \emph{admissible} if
\( \diam(C) \leq \rr_{k} \) for all \( C \in \partition_{k} \)
and every next partition \( \partition_{k+1} \) is a refinement of the previous one 
\( \partition_{k} \):
\begin{EQA}
	\forall C_{k+1} \in \partition_{k+1} \quad \exists C_{k} \in \partition_{k} 
	&\text{ such that } &
	C_{k+1} \subset C_{k} \, .
\label{forallCkp1Ck}
\end{EQA}
The generic chaining approach can be formulated via such partitions.
Let \( \mes(\cdot) \) be a given measure.
Define 
\begin{EQA}
	\NN_{k}
	& \eqdef &
	\sup_{C \in \partition_{k}} \frac{\mes(\Upsd)}{\mes(C)} \, .
\label{NNkpartit}
\end{EQA} 
Now define \( \entrlq(\Upsd) \) and \( \entrlg(\Upsd) \) by \eqref{entrldefch}.

\begin{theorem}
\label{TUPUpsdg}
Let \( \UP \) be a separable process following to \nameref{CSdref} and 
\( \Upsd \) be a \( \dist \)-ball in \( \Ups \) with the center \( \upsd \) and the radius 
\( \rups \), i.e. \( \dist(\ups,\upsd) \le \rups \) for all \( \ups \in \Upsd \).
Let \( \QQq = \entrlq(\Upsd) \) and \( \QQg = \entrlg(\Upsd) \) be defined 
in \eqref{entrldefch} with \( \NN_{k} \) given by \eqref{NNkpartit} 
for an admissible partitioning \( (\partition_{k}) \). 
Then all the statements of Theorem~\ref{TUPUpsdch} continue to apply.
\end{theorem}
}

\Section{A large deviation bound}
Due to the result of Theorem~\ref{TUPUpsdch}, the bound for the maximum of 
\( \UP(\ups,\upss) \) over \( \ups \in \B_{\rr}(\upss) \) grows linearly in \( \rr \). 
So, its applications to situations with \( \rr \gg \entrlq(\Upsd) \) are limited.
The next result shows that introducing a negative drift helps to
state a uniform in \( \rr \) local probability bound.
Namely, the bound for the process 
\( \UP(\ups,\upss) - f(\dist(\ups,\upss)) \) for some 
function \( f(\rr) \) over a ball \( \B_{\rr}(\upss) \) around 
the point \( \upss \) does not depend on \( \rr \).
Here the generic chaining arguments are accomplished with the slicing technique.
The idea is for a given \( \rrb > 1 \) to split the ball \( \B_{\rrb}(\upss) \) into the slices 
\( \B_{\rr_{k}}(\upss) \setminus \B_{\rr_{k-1}}(\upss) \) 
and to apply Theorem~\ref{TUPUpsdch} to each slice separately.

\begin{theorem}
\label{TsuprUP}
Let \( \rrb \) be such that \nameref{CSdref} holds on \( \B_{\rrb}(\upss) \).
Let also \( \entrlq(\B_{\rr}(\upss)) \le \QQq \) and \( \entrlg(\B_{\rr}(\upss)) \le \QQg \) for 
\( \rr \le \rrb \).
Given \( \rups < \rrb \), let a monotonous function \( f(\rr,\rups) \) fulfill for some 
\( \rho < 1 \) 
\begin{EQA}[c]
	f(\rr,\rups) 
	\geq
	\nunu \rr \, \zzQ\bigl( \xx + \log(\rr/\rups) \bigr),
	\quad 
	\rups \leq \rr \leq \rrb,
\label{frhorrxxrr}
\end{EQA}
where the function \( \zzQ(\cdot) \) is given by \eqref{zzxxgfin}.
Then it holds
\begin{EQA}
    \P\biggl( 
        \sup_{\rups \leq \rr \leq \rrb} \,\,
        \sup_{\ups \in \B_{\rr}(\upss)} 
        \bigl\{ 
            \UP(\ups,\upss) - f\bigl( \rho^{-1} \rr, \rups \bigr)
        \bigr\} 
        \geq 0 
    \biggr)
    & \le &
    \frac{\rho}{1 - \rho} \ex^{-\xx} .
\label{Psuprzz}
\end{EQA}    
\end{theorem}

\begin{remark}
\label{RTsuprUP}
Formally the bound applies even with \( \rrb = \infty \) provided that \nameref{CSdref} 
is fulfilled on the whole set \( \Upsd \).
\end{remark}

\begin{remark}
If \( \gm = \infty \), then \( \zzQ(\xx) = 2 \QQq + \sqrt{8 \xx} \) and the condition \eqref{frhorrxxrr} on the drift simplifies to 
\( (2 \nunu \rr)^{-1} f(\rr,\rups) \geq \QQq + \sqrt{2\xx + 2 \log(\rr/\rups)} \).
\end{remark}

\begin{proof}
By \eqref{frhorrxxrr} and Theorem~\ref{TUPUpsdch} for any \( \rr > \rups \)
\begin{EQA}
	&& \nquad
	\P\biggl(  
		\sup_{\ups \in \B_{\rr}(\upss) \setminus \B_{\rho \rr}(\upss)} 
        \bigl\{ \UP(\ups,\upss) - f\bigl( \rr,\rups) \bigr) \bigr\}
		\geq 0 
	\biggr)
	\\
	& \leq &
	\P\biggl( 
		\frac{1}{\nunu \rr} \sup_{\ups \in \B_{\rr}(\upss)} \UP(\ups,\upss) 
		\geq 
		\zz\bigl(\xx + \log(\rr/\rups)\bigr) 
	\biggr)
	\leq 
	\frac{\rups}{\rr} \ex^{- \xx} .
\label{Puprrrho0}
\end{EQA}
Now defined \( \rr_{k} = \rups \rho^{-k} \) for \( k=0,1,2,\ldots \).
Define also \( \kb \eqdef \log (\rrb/\rups) + 1 \).
It follows from \eqref{Puprrrho0} that 
\begin{EQA}
	&& \nquad
	\P\biggl( 
        \sup_{\ups \in \B_{\rrb}(\upss) \setminus \B_{\rups}(\upss)} 
        \biggl\{ 
            \UP(\ups,\upss) - f\bigl( \rho^{-1} \dist(\ups,\upss), \rups \bigr)
        \biggr\} \geq 0 
    \biggr)
    \\
    & \leq &
	\sum_{k=1}^{k*} \P\biggl(  
		\frac{1}{\rr_{k}} \sup_{\ups \in \B_{\rr_{k}}(\upss) \setminus \B_{\rr_{k-1}}(\upss)} 
        \biggl\{ \UP(\ups,\upss) - f\bigl( \rr_{k},\rups \bigr) \biggr\}
		\geq 0 
	\biggr)
	\\
	& \leq &
	\ex^{-\xx} \sum_{k=1}^{\kb} \rho^{k}
	\leq 
	\frac{\rho}{1 - \rho} \ex^{- \xx} 
\label{PBrrbmBrups}
\end{EQA}
as required.
\end{proof}
%
%
%
%

\Section{Finite-dimensional smooth case}
\label{Sexpsmooth}
Here we discuss the special case when \( \Ups \) is an open subset in 
\( \R^{\dimp} \),
the stochastic process \( \UP(\ups) \) is absolutely continuous and its gradient
\( \nabla \UP(\ups) \eqdef d \UP(\ups) / d \ups \) 
has bounded exponential moments.

\begin{description}
\item[\( \bb{(\CS\! D)} \)\label{CSDref}]\textit{
There exist \( \gmb > 0 \), 
\( \nunu \ge 1 \), and for each \( \ups \in \Ups \),
a symmetric non-negative matrix \( \VV(\ups) \)
such that for any  \( \lambda \le \gmb \) 
and any unit vector \( \gammav \in \R^{\dimp} \), it holds
}
\begin{EQA}[c]
    \log \E \exp \Bigl\{
       \lambda 
       \frac{\gammav^{\T}\nabla \UP(\ups)}
            {\| \VV(\ups) \gamma \|}
    \Bigr\} 
    \le 
    \frac{\nunu^{2} \lambda^{2}}{2} .
\end{EQA}
\end{description}

A natural candidate for \( \VV^{2}(\ups) \) is the covariance matrix 
\( \Var\bigl( \nabla \UP(\ups) \bigr) \) provided that this matrix
is well posed. Then the constant \( \nunu \) can be taken close to one by reducing 
the value \( \gmb \).

In what follows we fix a subset \( \Upsd \) of \( \Ups \) and establish a bound for 
the maximum of the process \( \UP(\ups,\upsd) = \UP(\ups) - \UP(\upsd) \) on 
\( \Upsd \) for a fixed point \( \upsd \). 
We assume existence of a matrix \( \VVb = \VVb(\Upsd) \) such that \( \VV(\ups) \preceq 
\VVb \) for all \( \ups \in \Upsd \). 
We also assume that \( \mes \) is the Lebesgue measure on \( \Ups \).
First we show that the differentiability condition \nameref{CSDref} implies 
\nameref{CSdref}.

\begin{lemma}
\label{Lsmu}
Assume that \nameref{CSDref} holds with some \( \gmb \) and \( \VV(\ups) \preceq \VVb \)
for \( \ups \in \Upsd \). 
Consider any \( \ups,\upsd \in \Upsd \).
Then it holds for \( |\lambda| \le \gmb \)
\begin{EQA}
\label{zsmu}
    \log \E \exp \biggl\{
        \lambda \frac{\UP(\ups,\upsd)}{\| \VVb(\ups - \upsd) \|}
    \biggr\}
    & \le &
    \frac{\nunu^{2} \lambda^{2}}{2} .
\end{EQA}
\end{lemma}

\begin{proof}
Denote \( \normc = \| \ups - \upsd \| \),
\( \gammav = (\ups - \upsd)/\normc \).
Then
\begin{EQA}[c]
    \UP(\ups,\upsd)
    =
    \normc \gammav^{\T}\int_{0}^{1} \nabla \UP(\upsd + t \normc \gammav) dt 
\end{EQA}
and \( \| \VVb(\ups - \upsd) \| = \normc \| \VVb \gammav \| \).
Now the H\"older inequality and \nameref{CSDref} yield
\begin{EQA}
    && \nquad
    \E \exp \biggl\{
        \lambda \frac{\UP(\ups,\upsd)}{\| \VVb(\ups - \upsd) \|} 
        - \frac{\nunu^{2} \lambda^{2}}{2}
    \biggr\}
    \\
    &=& 
    \E \exp \biggl\{
        \int_{0}^{1} \biggl[\lambda
            \frac{\gammav^{\T} \nabla \UP(\upsd + t \normc \gammav)}
                 {\| \VVb \gammav \|}
            - \frac{\nunu^{2} \lambda^{2} }{2}
        \biggr]  dt
    \biggr\}
    \\
    & \le & 
    \int_{0}^{1} \E \exp \biggl\{
            \lambda
            \frac{\gammav^{\T} \nabla \UP(\upsd + t \normc \gammav)}
                 {\| \VVb \gammav \|}
            - \frac{\nunu^{2} \lambda^{2} }{2}
    \biggr\}\, dt
    \le 1
\end{EQA}
as required.
\end{proof}

The result of Lemma~\ref{Lsmu} enables us to define 
\( \dist(\ups,\upsc) = \| \VVb(\ups - \upsc) \| \) so that the corresponding 
\( \dist \)-ball coincides with the following ellipsoidal set \( \BU(\rr,\upsd) \):
\begin{EQA}[c]
	\BU(\rr,\upsd)
	\eqdef
	\bigl\{ \ups \colon \| \VVb (\ups - \upsd) \| \leq \rr \bigr\}.
\end{EQA}

\ifbook{}
{\Subsection{Covering and entropy for Euclidean distance}}
Now we bound the value \( \entrl(\Upsd) \) for \( \Upsd = \BU(\rr,\upsd) \).
Note that by change of variable one can reduce the study to the case \( \VVb = \Id_{\dimp} \)
and consider the entropy of the unit ball in \( \R^{\dimp} \) w.r.t. the Euclidean distance.
We use the following general result which allows to upperbound the covering number 
of a convex set in \( \R^{\dimp} \) for the Euclidean metric.

\begin{lemma}
\label{Lcoveringnumber}
Let \( \Upsd \) be a convex set in \( \R^{\dimp} \), \( \delta > 0 \), and \( B \) be 
the unit ball in \( \R^{\dimp} \).
Then the covering number \( \NN(\Upsd,\delta) \) fulfills
\begin{EQA}
	\NN(\Upsd,\delta)
	& \leq &
	\frac{\vol\bigl( \Upsd + (\delta/2)B \bigr)}{\vol(B)} (2/\delta)^{\dimp} \, .
\label{NNUpsddelt}
\end{EQA}
\end{lemma}

\begin{proof}
Let \( (\ups^{(i)} \, , i=1,\ldots,\NN ) \) be a maximal subset of \( \Upsd \) such that 
\( \| \ups^{(i)} - \ups^{(j)} \| \geq \delta \) for all \( i \ne j \).
By maximality, \( (\ups^{(i)}) \) is a \( \delta \)-net of \( \Upsd \).
Let also \( B \) be the unit ball in \( \R^{\dimp} \).
Note that the balls \( \ups^{(i)} + (\delta/2) B \) are disjoint and included 
in \( \Upsd + (\delta/2) B  \). 
Therefore,
\begin{EQA}
	\sum_{i \leq \NN} \vol\bigl( \ups^{(i)} + \frac{\delta}{2} B \bigr)
	& \leq &
	\vol\Bigl( \Upsd + \frac{\delta}{2} B \Bigr) ,
\label{sumiNNkball}
\end{EQA}
where \( \vol(A) \) means the Lebesgue measure of the set \( A \).
This yields
\begin{EQA}
	\NN \,\, (\delta/2)^{\dimp} \, \vol(B)
	& \leq &
	\vol\bigl( \Upsd + (\delta/2)B \bigr) 
\label{NNkd2pvUball}
\end{EQA}
and the claim of the lemma follows.
\end{proof}

\begin{lemma}[Entropy of a ball]
\label{Lcoveringball}
Let \( \Upsd = \BU(\rupd,\upss) \) and \( \rr_{k} = 2^{-k} \rupd \).
Then the covering numbers \( \NN_{k} \) fulfill with \( \delta = \rr_{k}/\rupd = 2^{-k} \)
\begin{EQA}
	\NN_{k}
	& \leq &
	(1 + 2/\delta)^{\dimp}
	=
	(1 + 2^{k+1})^{\dimp} \, .
\label{NNkball}
\end{EQA}
Moreover, with \( \cdimg = 4.67 \),
\begin{EQ}[rcccc]
	\entrlg(\Upsd) 
	&\le & 
	2 \log 2 + \cdimg \, \dimp
	& \leq & 
	6 \dimp ,
	\\
	\entrlq(\Upsd) 
	& \le & 
	\sqrt{2 \log 2 + \cdimg \dimp} 
	&\leq & 
	\sqrt{6 \dimp}.
\label{entrballb}
\end{EQ}
\end{lemma}

\begin{proof}
A change of variable reduces the statement to the case \( \VVb = \Id_{\dimp} \) and 
\( \rupd = 1 \).
For \( \delta = 2^{-k} \), this implies 
by Lemma~\ref{Lcoveringnumber} in view of \( \Upsd = B \)
\begin{EQA}
	\vol\Bigl( \Upsd + \frac{\delta}{2} B \Bigr) 
	&=&
	\bigl( 1 + {\delta}/{2} \bigr)^{\dimp} \vol(B) ,
\label{volUdB2}
\end{EQA}
that
\( \NN_{k} \leq (1 + 2/\delta)^{\dimp} \) as claimed.
Now we derive
\begin{EQ}[rcl]
	\entrlg(\Upsd)
	& \leq &
	\sum_{k=1}^{\infty} 2^{-k+1} \log(2 \NN_{k})
	\leq 
	\sum_{k=1}^{\infty} 2^{-k+1} \bigl\{ \log 2 + 2 \dimp \log (1 + 2^{k+1}) \bigr\}
	\\
	& \leq &
	{2 \log 2} + {\dimp} \sum_{k=0}^{\infty} 2^{-k+1} \log(1 + 2^{k})
	\\
	& \leq &
	{2 \log 2} + \cdimg {\dimp} 
\label{entrgball}
\end{EQ}
as required.
\end{proof}

Now we specify the local bounds of Theorem~\ref{TUPUpsdch} to the smooth case. 
We consider the local sets of the elliptic form 
\( \Upss(\rr) \eqdef \{ \ups: \| \VVc(\ups - \upss) \| \le \rr \} \), 
where \( \VVc \) dominates \( \VV(\ups) \) on this set: \( \VV(\ups) \preceq \VVc \).

\begin{theorem}
\label{Tsmoothpenlc}
Let \nameref{CSDref} hold with some \( \gmb > 0 \), and matrices \( \VV(\ups) \)
such that \( \VV(\ups) \preceq \VVc \) for all \( \ups \in \Upss(\rr) \)
and a fixed \( \rr \).
For any \( \xx \ge 1/2 \) 
\begin{EQA}[c]
\label{Upsdbounddsm}
    \P\Bigl\{ 
        \frac{1}{\nunu \, \rr} \sup_{\ups \in \Upss(\rr)} \bigl| \UP(\ups,\upss) \bigr| 
    	\geq 
		\zzQ(\xx)
    \Bigr\}
    \le 
    \ex^{-\xx} ,
\end{EQA}
where \( \zzQ(\xx) \) is given by \eqref{zzxxgfin} with 
\( \entrlq(\Upsd) \) and \( \entrlg(\Upsd) \) from \eqref{entrballb}.
\end{theorem}

\begin{proof}
Lemma~\ref{Lcoveringball} implies \nameref{CSdref} with
\( \dist(\ups,\upss) = \| \VVc(\ups - \upss) \| \). 
Now the result follows from Theorem~\ref{TUPUpsdch}.
\end{proof}

\ifbook{}{
\Subsection{Generic chaining}
\begin{lemma}[Entropy for generic chaining]
\label{LBtvtd}
Let \( \Upsd = \BU(\rupd,\upsd) \).
Under the conditions of Lemma~\ref{Lsmu}, it holds
\( \entrgg(\Upsd) \le \cdimb \dimp \), 
where \( \cdimb = 2 \) for \( \dimp \ge 2 \),
and \( \cdimb = 2.7 \) for \( \dimp = 1 \).
\end{lemma}

\begin{proof}
The set \( \Upsd \) coincides with the ellipsoid 
\( \BU(\rupd,\upsd) \) while
the \( \dist \)-ball \( \B_{k}(\ups) \) coincides with the ellipsoid 
\( \BU(\rr_{k},\ups) \) for each \( k \ge 0 \).
By change of variables, the study can be reduced to the case with 
\( \upsd = 0 \), \( \VVb \equiv \Id_{\dimp} \), \( \rups = 1 \),
so that \( \BU(\rr,\ups) \) is the usual Euclidean ball in \( \R^{\dimp} \)
of radius \( \rr \). 
It is obvious that the measure of the overlap of two balls 
\( \BU(1,0) \) and \( \BU(2^{-k},\ups) \) for \( \| \ups \| \le 1 \) is minimized 
when \( \| \ups \| = 1 \), and this value is the same for all such \( \ups \). 
For \( k=0 \), the overlap of \( \BU(1,0) \) and \( \BU(1,\ups) \) contains the ball 
\( \BU(1/2,\ups/2) \) of radius \( 1/2 \), so that \( M_{0} \leq 2^{\dimp} \).
For \( k \geq 1 \), we use the following observation. 
Fix \( \upsdc \) with \( \| \upsdc \| = 1 \).
Let \( \rr \le 1 \),
\( \upsdu = (1 - \rr^{2}/2) \upsdc \) and \( \rrf = \rr - \rr^{2}/2 \).
If \( \ups \in \BU(\rrf,\upsdu) \), then 
\( \ups \in \BU(\rr,\upsdc) \) because 
\begin{EQA}[c]
    \| \upsdc - \ups \| \le \| \upsdc - \upsdu \| + \| \ups - \upsdu \| 
\le \rr^{2}/2 + \rr - \rr^{2}/2 \le \rr .
\label{upsdupsrrr}
\end{EQA}
Moreover, for each \( \ups \in \BU(\rrf,\upsdu) \), it holds 
with \( \uv = \ups - \upsdu \)
\begin{EQA}[c]
    \| \ups \|^{2}
    =
    \| \upsdu \|^{2} + \| \uv \|^{2} + 2 \uv^{\T} \upsdu
    \le 
    (1 - \rr^{2}/2)^{2} + |\rrf|^{2} + 2 \uv^{\T} \upsdu
    \le 
    1 + 2 \uv^{\T} \upsdu .
\label{upsscv}
\end{EQA}    
This means that either \( \ups = \upsdu + \uv \) or 
\( \upsdu - \uv \) belongs to the ball \( \BU(\rups,\upsd) \) and thus, 
\( \mes\bigl( \BU(1,0) \cap \BU(\rr,\ups) \bigr) \ge 
\mes\bigl( \BU(\rrf,\upsdu) \bigr)/2 \).
We conclude that 
\begin{EQA}[c]
    \frac{\mes\bigl( \BU(1,0) \bigr)}{\mes\bigl( \BU(1,0) \cap \BU(\rr,\upsdc) \bigr)}
    \le 
    \frac{2 \mes\bigl( \BU(1,0) \bigr)}{\mes\bigl( \BU(\rrf,0) \bigr)} 
    =
    2 (\rr - \rr^{2}/2)^{- \dimp}.
\label{mesrrrrc}
\end{EQA}    
This implies for \( k \ge 0 \) and \( \rr_{k} = 2^{-k} \) that 
\( 2 \NN_{k} \le 2^{2 + k \dimp}  (1 - 2^{-k-1})^{-\dimp} \).
The quantity \( \entrlg(\Upsd) \) can now be evaluated as
\begin{EQA}
    \frac{1}{2} \entrlg(\Upsd)
    & \le &
    \frac{1}{3} \log(2^{1+\dimp}) 
    + \frac{2}{3} \sum_{k=1}^{\infty} 2^{-k} \log(2^{2+k\dimp})
    - \frac{2\dimp}{3} \sum_{k=1}^{\infty} 2^{-k} \log(1 - 2^{-k-1})
    \\
    & = &
    \frac{\log 2}{3}  
    \Bigl[ 
        1 + \dimp + 2 \sum_{k=1}^{\infty} (2 + k \dimp) 2^{-k} 
    \Bigr]
    - \frac{2\dimp}{3} \sum_{k=1}^{\infty} 2^{-k} \log(1 - 2^{-k-1})        
    \le \cdimb \dimp ,
\label{entrlsmdp}
\end{EQA}    
where \( \cdimb = 2 \) for \( \dimp \ge 2 \),
and \( \cdimb = 2.7 \) for \( \dimp = 1 \),
and the result follows.
\end{proof}
}

\Section{Entropy of an ellipsoid}
\label{SEntrellips}
Let \( \HG \) be a positive self adjoint operator in \( \R^{\infty} \).
We are interested to describe the entropy of the elliptic set 
\begin{EQA}
	\Ellips_{\HG}(\rA)
	& \eqdef &
	\bigl\{ \ups \colon \| \HG (\ups - \upsd) \| \leq \rA \bigr\}
\label{Upsdr0del}
\end{EQA}
for given 
\( \upsd \in \R^{\infty} \) and \( \rA > 0 \)
with respect to the usual Euclidean distance in \( \R^{\infty} \).
Below we evaluate the entropy of this set assuming that \( \| \HG^{-1} \|_{\oper} = 1 \) and
\( \HG^{-2} \) is a trace operator, i.e., \( \hg_{1} = 1 \) and 
\begin{EQA}
\label{effdimHG}
	\dimA_{\HG}
	\eqdef
	\tr(\HG^{-2})
	=
	\sum_{j=1}^{\infty} \hg_{j}^{-2}
	& < &
	\infty ,
\label{sumjHGel}
\end{EQA}
where \( \hg_{1} \leq \hg_{2} \leq \ldots \) are the ordered eigenvalues of \( \HG \).

\begin{theorem}
\label{LHGaklogak}
Suppose that for some \( \alpha > 1 \)
\begin{EQA}
	\dimA_{\HG}(\alpha)
	& \eqdef &
	\sum_{j=1}^{\infty} \hg_{j}^{-2} \log^{\alpha}(\hg_{j}^{2})
	<
	\infty .
\label{dimAHGell}
\end{EQA}
Then for \( \Ellips = \Ellips_{\HG}(\rA) \)
\begin{EQA}
	\entrlq(\Ellips)
	& \leq &
	\CONST (\alpha-1)^{-1/2} \sqrt{\dimA_{\HG}(\alpha)} \, ,
\label{entrlelalp}
\end{EQA}
where \( \CONST \) is an absolute constant.
Furthermore,
\begin{EQA}
	\entrlg(\Ellips)
	& \leq &
	\CONST \dimAg_{\HG} 
	=
	\CONST \sum_{j=1}^{\infty} \hg_{j}^{-1} \, .
\label{entrlgelalp}
\end{EQA}
\end{theorem}

\begin{remark}
The log-factor in the definition of \( \dimA_{\HG}(\alpha) \) can be removed by
using a more advanced generic chaining and majorising measure technique.
However, in most of situations, the bound in terms of \( \dimA_{\HG}(\alpha) \) is also sharp.

The term \( \dimAg_{\HG} \) only appears in the sub-exponential case when \( \gm < \infty \).
In this case we need the condition \( \dimAg_{\HG} < \infty \) which requires 
\( \sum_{j} \hg_{j}^{-1} < \infty \), that is, 
a more rapid 
growth of the values \( \hg_{j} \) is necessary than in \eqref{dimAHGell}.
\end{remark}

\begin{proof}
We begin by a general lemma which bounds the covering numbers for the elliptic set 
\( \Ellips \) for the Euclidean distance.

\begin{lemma}[Entropy of the ellipsoid]
\label{Lentrlellips}
Let \( \Ellips = \Ellips_{\HG}(\rA) \) be an elliptic set from \eqref{Upsdr0del}
with \( \| \HG^{-1} \|_{\oper} = 1 \) and \( \tr(\HG^{-2}) < \infty \).
Let also \( \dist(\ups,\upsc) = \| \ups - \upsc \| \).
Then for \( \rr_{k} = 2^{-k} \rA \), 
the value \( \entrlq(\Ellips) \) from \eqref{entrldefch} satisfies
\begin{EQA}
	\entrlq(\Ellips)
	& \leq &
	\sum_{k=1}^{\infty} 2^{-k} 
		\sqrt{\log 2 + 2 \HGLsum_{\HG}(m_{k})} \, ,
\label{entrellUdk}
\end{EQA}
where \( m_{k} \) is the index \( j \) for which 
\( \hg_{m_{k}}^{2} = 2^{2k+1} \) and hence,
\begin{EQA}
	\hg_{j}^{2}
	& \leq &
	2^{2k+1} ,
	\qquad
	j \leq m_{k} \, ,
\label{hgpkdefel}
\end{EQA}
and 
\begin{EQA}
	\HGLsum_{\HG}(m)
	& \eqdef &
	\sum_{j=1}^{m} \log(3 \hg_{m}/\hg_{j}) .
\label{Lmedell}
\end{EQA}
\end{lemma}

\begin{remark}
For the ease of presentation, we supposed in the lemma that for each \( k \geq 1 \),
there exists some \( m_{k} \) with \( \hg_{m_{k}} = 2^{k+1/2} \).
The results easily extend to the case when this equality is approximate.
\end{remark}
\begin{proof}
Without loss of generality assume \( \upsd = 0 \).
A basis transform reduces the study to the case when \( \HG \) is diagonal:
\begin{EQA}
	\HG
	&=&
	\diag\bigl\{ \hg_{1},\hg_{2},\ldots \bigr\}. 
\label{HGdiagel}
\end{EQA}
We only have to evaluate the covering numbers \( \NN_{k} \).
Let us fix \( k \geq 1 \) and let \( m_{k} \) be given by \eqref{hgpkdefel}.
For any point \( \ups = (\ups_{1},\ups_{2},\ldots)^{\T} \) in \( \Ellips \), it holds
\begin{EQA}
	\sum_{j=m_{k}+1}^{\infty} \ups_{j}^{2}
	& = &
	\sum_{j=m_{k}+1}^{\infty} \hg_{j}^{-2} \, \hg_{j}^{2} \, \ups_{j}^{2}
	\\
	& \leq &
	\hg_{m_{k}+1}^{-2} 
		\sum_{j=m_{k}+1}^{\infty} \hg_{j}^{2} \ups_{j}^{2}
	\\
	& \leq &
	\hg_{m_{k}+1}^{-2} \sum_{j=1}^{\infty} \hg_{j}^{2} \ups_{j}^{2}
	\leq 
	2^{-2k-1} \rA^{2}
	\leq 
	\rr_{k}^{2} / 2 .
\label{sumpkp1infel}
\end{EQA}
Consider the elliptic set \( \Ellips_{k} \) in \( \R^{m_{k}} \) obtained by projection
\( \Proj_{k} \)
of \( \Ellips \) on the first \( m_{k} \) coordinates:
\begin{EQA}
	\Ellips_{k}
	& \eqdef &
	\bigl\{ (\ups_{1},\ldots,\ups_{m_{k}})^{\T} \colon 
	\sum_{j=1}^{m_{k}} \hg_{j}^{2} \ups_{j}^{2} \leq \rA^{2} \bigr\} .
\label{Upsdrpkel}
\end{EQA}
Let \( \MM_{k} \) be a \( \reps_{k} \)-net in \( \Ellips_{k} \) for 
\( \reps_{k}^{2} = \rr_{k}^{2}/2 \).
A \( \rr_{k} \)-net in \( \Ellips \) can be constructed from \( \MM_{k} \) in a simple way: 
just fix to zero the remaining coordinates 
\( \ups_{j} = 0 \) for \( j > m_{k} \).
If \( \upsd \) is constructed in this way, then 
\( \| \HG \upsd \| = \| \HG \Proj_{k} \upsd \| \leq 1 \), 
that is, \( \upsd \in \Ellips \).
Moreover, 
for any other point \( \ups \in \Ellips \), take \( \upsd \) such that 
their projections satisfy \( \| \Proj_{k} (\ups - \upsd) \| \leq \reps_{k} \).
Then by \eqref{sumpkp1infel}
\begin{EQA}
	\| \ups - \upsd \|^{2}
	&=&
	\| \Proj_{k} (\ups - \upsd) \|^{2} + \| (\Id - \Proj_{k}) \ups \|^{2}
	\leq 
	\rr_{k}^{2}/2 + \rr_{k}^{2}/2
	=
	\rr_{k}^{2} .
\label{upsupsd2el}
\end{EQA}
Therefore, the covering number \( \NN(\Ellips,\rr_{k}) \) 
of the infinite dimensional elliptic set \( \Ellips \) does not exceed the covering number 
\( \NN(\Ellips_{k},\reps_{k}) \) for the \( m_{k} \)-dimensional ellipsoid \( \Ellips_{k} \).
By Lemma~\ref{Lcoveringnumber} with \( \delta = \reps_{k} \), 
\begin{EQA}
	\NN(\Ellips_{k},\reps_{k})
	& \leq &
	\frac{\vol\bigl( \Ellips_{k} + (\reps_{k}/2)B_{k} \bigr)}{\vol(B_{k})} 
		(2/\reps_{k})^{m_{k}} \, ,
\label{NNUpsddeltel}
\end{EQA}
where \( B_{k} \) is the unit ball in \( \R^{m_{k}} \).
The bound \( \hg_{j}^{-2} \geq 2^{-2k-1} \) for \( j \leq m_{k} \) 
implies that \( \Ellips_{k} + (\reps_{k}/2) B_{k} \) is contained in the elliptic set 
\( (3/2) \Ellips_{k} \).

The definition implies due to \( \hg_{m_{k}}^{2} = 2^{2k+1} \)
\begin{EQA}
	\NN(\Ellips,\rr_{k})
	& \leq &
	\log \frac{\vol\bigl( (3/2) \Ellips_{k} \bigr)}{(\reps_{k}/2)^{m_{k}} \, \vol(B_{k})} 
	\\
	& \leq &
	\sum_{j=1}^{m_{k}} \log \frac{3 \hg_{j}^{-1} }{\reps_{k}}
	\leq 
	\sum_{j=1}^{m_{k}} \log \bigl( 3 \hg_{m_{k}}/\hg_{j} \bigr) 
	=
	\HGLsum_{\HG}(m_{k}) .
\label{logfrvolelk}
\end{EQA}
Now the result \eqref{entrellUdk} follows by the definition of \( \entrlq(\Ellips) \). 
\end{proof}


Denote \( \NN_{k} = \NN(\Ellips,\rr_{k}) \).
By the Cauchy-Schwartz inequality for \( \alpha > 1 \)
\begin{EQA}
	\entrlq(\Ellips)
	&=&
	\sum_{k=1}^{\infty} 2^{-k} \sqrt{2 \log(2 \NN_{k})}
	\leq 
	\biggl\{ 
		\sum_{k=1}^{\infty} k^{- \alpha} \sum_{k=1}^{\infty} 
			k^{\alpha} \, 2^{-2k} \, 2 \log(2 \NN_{k}) 
	\biggr\}^{1/2} .
\label{entrellalp1}
\end{EQA}
The use of \( \hg_{m_{\ell}}^{-2} = 2^{2 \ell + 1} \) and 
\( \hg_{j}^{2} \geq 4 \, \hg_{m_{\ell-1}}^{2} \) for \( j \in (m_{\ell-1},m_{\ell}] \) yields
by \eqref{logfrvolelk} with \( n_{\ell} \eqdef m_{\ell} - m_{\ell-1} \)
\begin{EQA}
	2 \log(2 \NN_{k})
	&=& 
	\sum_{\ell=1}^{k} \sum_{j=m_{\ell-1}+1}^{m_{\ell}} 
		\log \frac{9 \hg_{m_{k}}^{2}}{\hg_{j}^{2}}  
	\leq 
	\sum_{\ell=1}^{k} \bigl\{ k - \ell  + \log(36) \bigr\} n_{\ell}   
\label{entrellalp2}
\end{EQA}
Further, in view of \( \hg_{m_{k}} = 2^{k} \)
\begin{EQA}
	\sum_{k=1}^{\infty} k^{\alpha} \, 2^{-2k} \, 2\NN_{k}
	& \leq &
	\sum_{k=1}^{\infty} k^{\alpha} \, 2^{-2k} \sum_{\ell=1}^{k} 
		\bigl\{ k - \ell  + \log(36) \bigr\} n_{\ell}
	\\
	&=&
	\sum_{\ell=1}^{\infty} \sum_{k \geq \ell} 
		k^{\alpha} \, 2^{-2k} \bigl\{ k - \ell  + \log(36) \bigr\} n_{\ell}
	\\
	&=&
	\sum_{\ell=1}^{\infty} n_{\ell} \, 2^{-2 \ell} \sum_{k \geq \ell} 
		k^{\alpha} \, 2^{-2 (k - \ell)} \, \bigl\{ k - \ell  + \log(36) \bigr\} 
	\\
	&=&
	\CONST \sum_{\ell=1}^{\infty} n_{\ell} \, 2^{-2 \ell} \, \ell^{\alpha} .
\label{entrellalp3}
\end{EQA}
It remains to note that \( 2^{2\ell-1} \leq \hg_{j}^{2} \leq 2^{2\ell+1} \) for 
\( m_{\ell-1} < j \leq m_{\ell} \) and
\begin{EQA}
	\sum_{\ell=1}^{\infty} n_{\ell} \, 2^{-2 \ell} \, \ell^{\alpha}
	& \leq &
	\sum_{\ell=1}^{\infty} \sum_{j=m_{\ell-1}+1}^{m_{\ell}} 
		\hg_{j}^{-2} \, \log^{\alpha}(\hg_{j}^{2}) 
	=
	\sum_{j=1}^{\infty} \hg_{j}^{-2} \, \log^{\alpha}(\hg_{j}^{2}) 
	=
	\dimA_{\HG}(\alpha).
	\qquad
\label{entrellalp4}
\end{EQA}
The assertion \eqref{entrlelalp} now follows from \eqref{entrellalp1} in view of 
\( \sum_{k \geq 1} k^{-\alpha} \leq \CONST (\alpha - 1)^{-1} \).

The result on \( \entrlg(\Ellips) \) requires to bound the sum of \( 2^{-k} \log \NN_{k} \).
Similarly to the above, one easily derives
\begin{EQA}
	\sum_{k=1}^{\infty} 2^{-k} \, \NN_{k}
	& \leq &
	\sum_{k=1}^{\infty} 2^{-k} \sum_{\ell=1}^{k} 
		\bigl\{ k - \ell  + \log(36) \bigr\} n_{\ell}
	\\
	&=&
	\sum_{\ell=1}^{\infty} \sum_{k \geq \ell} 
		2^{-k} \bigl\{ k - \ell  + \log(36) \bigr\} n_{\ell}
	\\
	&=&
	\sum_{\ell=1}^{\infty} n_{\ell} \, 2^{- \ell} \sum_{k \geq \ell} 
		2^{- (k - \ell)} \, \bigl\{ k - \ell  + \log(36) \bigr\} 
	\\
	&=&
	\CONST \sum_{\ell=1}^{\infty} n_{\ell} \, 2^{- \ell}  
	\leq 
	\CONST \sum_{j=1}^{\infty} \hg_{j}^{-1} 
	= 
	\CONST \dimAg_{\HG} \, .
\label{sk1inf2mkNNk}
\end{EQA}
Theorem is proved.
\end{proof}

Now we present a special case for which the entropy can be bounded via the effective dimension
\( \dimA_{\HG} \) of \( \Upsd \) defined in \eqref{effdimHG}.

\begin{theorem}
\label{HGLsumfF}
Let \( \hg_{j}^{2} = f(j) \) for a monotonously increasing smooth function \( f(x) > 0 \).
If \( x f'(x) / f(x) \leq \beta \), then
\begin{EQ}[rcl]
	\entrlg(\Ellips)
	& \leq &
	\CONST {\beta \, \dimA_{\HG}},
	\\
	\entrlq(\Ellips)
	& \leq &
	\CONST \sqrt{\beta \, \dimA_{\HG}} ,
\label{entrlqElbeta}
\end{EQ}
where the effective dimension \( \dimA_{\HG} \) is defined in \eqref{effdimHG}. 
\end{theorem}

\begin{proof}
Obviously
\begin{EQA}
	\sum_{j=1}^{m} \log\biggl( \frac{\hg_{m}^{2}}{\hg_{j}^{2}} \biggr)
	&\leq &
	\int_{0}^{m} \log\biggl( \frac{f(m)}{f(t)} \biggr) dt .
\label{HGsumFx}
\end{EQA}
Now we note that the function
\begin{EQA}
	F(x)
	& \eqdef &
	\int_{0}^{x} \log\biggl( \frac{f(x)}{f(t)} \biggr) dt
\label{Fx1xfxel}
\end{EQA}
fulfills \( F(0) = 0 \) and
\begin{EQA}
	F'(x)
	&=&
	\frac{x f'(x)}{f(x)} .
\label{Fprimxfxel}
\end{EQA}
yielding
\begin{EQA}
	\sum_{j=1}^{m} \log\biggl( \frac{\hg_{m}^{2}}{\hg_{j}^{2}} \biggr)
	& \leq &
	\int_{0}^{m} \log\biggl( \frac{f(m)}{f(t)} \biggr) dt 
	=
	\int_{0}^{m} \frac{x f'(x)}{f(x)} dx .
\label{HGsumFx}
\end{EQA}
Moreover, 
In particular, if \( F'(x) \leq \beta \), then
\( F(x) \leq \beta x \) and thus, \( \HGLsum_{\HG}(m_{k}) \leq \beta m_{k} \).
Now it holds similarly to \eqref{entrellalp4}
\begin{EQA}
	\sum_{k=1}^{\infty} 2^{-k} m_{k}
	&=&
	\sum_{k=1}^{\infty} 2^{-k} \sum_{\ell=1}^{k} n_{\ell}
	\leq 
	\sum_{\ell=1}^{\infty} n_{\ell} \sum_{k \geq \ell}^{\infty} 2^{-k}
	=
	\sum_{\ell=1}^{\infty} n_{\ell} 2^{-\ell}
	\leq
	2 \sum_{j=1}^{\infty} \hg_{j}^{-2} 
	=
	2 \dimA_{\HG} \, ,
\label{HGsumFx3}
\end{EQA}
and the statement \eqref{entrlelalp} follows.
\end{proof}

Now we evaluate the entropy for the cases when \( \hg_{j} \) grow polynomially.

\begin{theorem}
\label{Tentrlell}
Let \( \hg_{j}^{2} = 1 + \kappa^{2} j^{2\beta} \) for \( \beta > 1/2 \) and some small value 
\( \kappa \).
Then
\begin{EQA}
	\entrlq(\Ellips)
	& \leq &
	\CONST (2\beta-1)^{-1/2} \kappa^{-1/(2\beta)} \, ,
	\\
	\entrlg(\Ellips)
	& \leq &
	\CONST (2\beta-1)^{-1} \kappa^{-1/\beta} \, ,
\label{entrlLbeta}
\end{EQA}
where \( \CONST \) is an absolute constant.
\end{theorem}

\begin{proof}
For \( f(x) = 1 + \kappa^{2} x^{2\beta} \), it holds \( x f'(x) / f(x) \leq 2 \beta \) and 
we can apply the result of Theorem~\ref{HGLsumfF}.
With \( \beta > 1/2 \), the effective dimension \( \dimA_{\HG} \) from \eqref{effdimHG} 
fulfills 
\begin{EQA}
	\dimA_{\HG}
	& \leq &
	\sum_{j=1}^{\infty} \hg_{j}^{-2}
	=
	\sum_{j=1}^{\infty} \frac{1}{1 + \kappa^{2} j^{2\beta}}
	\leq 
	\int_{0}^{\infty} \frac{1}{1 + \kappa^{2} x^{2\beta}} dx
	=
	\CONST \kappa^{-1/\beta} \frac{1}{2\beta-1}
\label{dimAHGelli}
\end{EQA}
and the result follows by \eqref{entrlqElbeta}.
\end{proof}

%
%
%

\Section{Roughness constraints for dimension reduction} 
\label{Sroughexp} 

The local bounds of Theorems~\ref{TUPUpsdch} and \ref{TsuprUP} can be extended in several 
directions. 
Here we briefly discuss one extension related to the use of a smoothness 
condition on the parameter \( \ups \). 
Let 
\( \penr(\ups) \) be a non-negative \emph{penalty} function on \( \Ups \). 
A particular example of such penalty function is the 
\emph{roughness penalty} \( \penr(\ups) = \| \GP \ups \|^{2} \) for a 
given \( \dimp \)-matrix \( \GP^{2} \). 
Let \( \rr \) be fixed.
Consider the intersection of the ball 
\( \B_{\rr}(\upsd) \) with the set \( \Ups \) given by the constraint 
\( \penr(\ups) \le 1 \): 
\begin{EQA}[c]
    \Ups_{\penr}(\rr) 
    = 
    \bigl\{ 
        \ups \in \Ups \colon \dist(\ups,\upsd) \le \rr; \, \penr(\ups) \le 1 
    \bigr\},
\label{rrdistpen} 
\end{EQA}   
for a fixed central point \( \upsd \) and the radius \( \rr \).
Here and below we assume that the central point \( \upsd \) is ``smooth'' in the sense that 
\( \penr(\upsd) < 1 \).
One can easily check that the results of Theorems~\ref{TUPUpsdch} and \ref{TsuprUP} and 
their corollaries extend to this situation without any change. 
The only difference is in the definition of the values \( \entrlq(\Upss) \) and 
\( \entrlg(\Upss) \) for \( \Upss = \Ups_{\penr}(\rr) \). 
Examples below show that the use of the penalization can 
substantially reduce these values relative to the non-penalized case. 

We consider the case of a smooth process \( \UP \) given on a local 
set \( \Ups_{\GP}(\rr) \) of the form
\begin{EQA}[c]
    \Ups_{\GP}(\rr) 
    = 
    \bigl\{ 
        \ups \in \Ups \colon \| \VVc (\ups - \upsd) \| \le \rr; \,\, \| \GP \ups \| \le 1 
    \bigr\},
\label{rrdistGP} 
\end{EQA}   
with the distance \( \dist(\ups,\upsd) = \| \VVc(\ups - \upsd) \| \) and a 
smoothness constraint \( \| \GP \ups \|^{2} \le 1 \).
Then the set \( \Ups_{\GP}(\rr) \) is contained in an elliptic set
\begin{EQA}[c]
	\Upss 
	\eqdef
	\bigl\{ 
		\thetav \colon 
		\| \GP \ups \|^{2} + \| \VVc (\ups - \upsd) \|^{2}  
		\leq 1 + \rr^{2}
	\bigr\} .
\label{Brrro}
\end{EQA}
Define
\begin{EQA}
	\VVGP^{2} 
	&=& 
	\VVc^{2} + \GP^{2},
	\\
	\upsdGP 
	&=& 
	\VVGP^{-2} \VVc^{2} \upsd .
\label{VVGPdef}
\end{EQA}
Then 
\begin{EQA}
	\upsd - \upsdGP 
	&=&
	(\Id_{\dimp} - \VVGP^{-2} \VVc^{2}) \upsd
	=
	\VVGP^{-2} \GP^{2} \upsd ,
\label{upsdupssGP}
\end{EQA}
and one can get by simple algebra
\begin{EQA}
	\| \GP \ups \|^{2} + \| \VVc (\ups - \upsd) \|^{2}
	&=&
	\| \VVGP (\ups - \upsdGP) \|^{2} + \| \GP \upsdGP \|^{2} 
	+ \| \VVc (\upsdGP - \upsd) \|^{2}
	\\
	&=&
	\| \VVGP (\ups - \upsdGP) \|^{2} 
	+
	{\upsd}^{\T} \GP^{2} \VVGP^{-2} \VVc^{2} \upsd
	=	
	\| \VVGP (\ups - \upsd) \|^{2} 
	+ d_{\GP}
\label{GPups2uud}
\end{EQA}
with \( d_{\GP} = {\upsd}^{\T} \GP^{2} \VVGP^{-2} \VVc^{2} \upsd \leq \| \GP \upsd \|^{2} < 1 \).
A change of variables \( \ups \to \VVc (\ups - \upsdGP) \) allows us to reduce the study to the case
of an ellipsoid considered in Section~\ref{SEntrellips}.
For \( \HG \) defined by \( \HG^{-2} = \VVc \, \VVGP^{-2} \, \VVc \), 
the set \( \Upss \) from \eqref{Brrro} is transferred 
into the elliptic set
\begin{EQA}[c]
	\Ups_{\HG}(\rr)
	=
	\bigl\{ 
		\ups \colon 
		\| \HG \ups \|^{2}  
		\leq 1 + \rr^{2} - d_{\GP}
	\bigr\} ,
\label{BrrrodGP}
\end{EQA}
whose entropy for the Euclidean distance is given via the effective dimension 
\( \dimA_{\HG} = \tr (\HG^{-2}) \). 

Now we are prepared to state the penalized bound for the process \( \UP(\cdot) \) over 
\( \Upss \) which naturally generalizes 
the result of Theorem~\ref{Tsmoothpenlc} to the non-penalized case. 

\begin{theorem}
\label{Texpro} 
Let \( \Upss = \Ups_{\penr}(\rr) \) be given by \eqref{rrdistGP} and \( \| \GP \upsd \| \leq 1 \).
Let also \nameref{CSDref} hold with some \( \gmb \) and a matrix 
\( \VV(\ups) \preceq \VVc \) for all \( \ups \in \Upss \).
For \( \HG \) defined by \( \HG^{-2} = \VVc \, \VVGP^{-2} \, \VVc \),
let the entropy values \( \entrlq(\Upsd) \) and \( \entrlg(\Upsd) \) 
for the elliptic set \( \Ups_{\HG}(\rr) \) from \eqref{Brrro} be given in Section~\ref{SEntrellips}.
Then for any \( \xx \ge 1/2 \) 
\begin{EQA}[c]
\label{Upsdboundpro}
    \P\Bigl\{ 
        \frac{1}{\nunu \rr}
        \sup_{\ups \in \Ups_{\penr}(\rr)} \bigl| \UP(\ups,\upsd) \bigr| 
    	\geq 
		\zzQ(\xx)
    \Bigr\}
    \le 
    \ex^{-\xx} ,
\end{EQA}
where \( \zzQ(\xx) \) is from \eqref{zzxxgfin} with 
these values \( \entrlq(\Upsd) \) and \( \entrlg(\Upsd) \).
\end{theorem}

\Section{Bound for a bivariate process}
Consider a smooth bivariate process \( \UP(\ups) = \UP(\ups_{1},\ups_{2}) \) over a product 
set \( \Ups = \Ups_{1} \times \Ups_{2} \),
where \( \Ups_{j} \subseteq \R^{\dimp_{j}} \) for \( j=1,2 \).
We suppose that partial derivatives of \( \UP \) have uniform exponential moments.

\begin{description}
\item[\( \bb{(\CS\! D_{p})} \)\label{CSD12ref}]\textit{
There exist \( \gmb > 0 \), 
\( \nunu \ge 1 \), and for each 
\( \ups = (\ups_{1},\ups_{2}) \in \Ups = \Ups_{1} \times \Ups_{2} \),
symmetric non-negative \( \dimp_{j} \times \dimp_{j} \) matrices \( \VV_{j} \), 
\( j = 1,2 \), such that for any  \( \lambda \le \gmb \) 
and any unit vector \( \gammav \in \R^{\dimp} \), it holds
}
\begin{EQA}[c]
    \log \E \exp \Bigl\{
       \lambda 
       \frac{\gammav^{\T}\nabla_{j} \UP(\ups)}
            {\| \VV_{j} \gamma \|}
    \Bigr\} 
    \le 
    \frac{\nunu^{2} \lambda^{2}}{2} , 
    \qquad j=1,2.
\end{EQA}
Here \( \nabla_{j} \UP \) denotes the partial derivative \( \partial \UP/ \partial \ups_{j} \)
for \( j=1,2 \).
\end{description}

This allows to establish an exponential bound for the process \( \UP(\ups) \).
Let us fix the central point \( \upsd = (\upsd_{1},\upsd_{2}) \) and a radius \( \rr \). 
As usual, 
\begin{EQA}
	\Ups_{j}(\rr) 
	&=& 
	\{ \ups_{j} \in \Ups_{j} \colon \| \VV_{j} (\ups_{j} - \upsd_{j}) \| \leq \rr \}
\label{Upsjr0}
\end{EQA}  
denotes the ball in \( \Ups_{j} \) with this radius.

\begin{theorem}
\label{TUPprod}
Let a bivariate random process \( \UP(\ups) \) on \( \Ups = \Ups_{1} \times \Ups_{2} \)
satisfy \nameref{CSD12ref}.
Then for any \( \rupd \) and \( \xx \geq 1/2 \), it holds on the product set 
\( \Upss = \Ups_{1}(\rupd) \times \Ups_{2}(\rupd) \)
\begin{EQA}[c]
\label{Upsdboundpro}
    \P\Bigl\{ 
        \frac{1}{\sqrt{8} \, \nunu \, \rupd}
        \sup_{\ups \in \Upss} \bigl| \UP(\ups,\upsd) \bigr| 
    	\geq 
		\zzQ(\xx)
    \Bigr\}
    \le 
    \ex^{-\xx} ,
\end{EQA}
with \( \zzQ(\xx) \) from \eqref{zzxxgfin} for 
\( \entrlq(\Upsd) = \entrlq(\Ups_{1}) + \entrlq(\Ups_{2}) \) and 
\( \entrlg(\Upsd) = \entrlg(\Ups_{1}) + \entrlg(\Ups_{2}) \).
\end{theorem}

\begin{proof}
By the H\"older inequality, \eqref{gUUem}, and \eqref{gUUgem}, it holds for 
\( \| \gammav_{1} \| = \| \gammav_{2} \| = 1 \) and \( \ups \in \Upss \)
\begin{EQA}
	&& \nquad
	\log \E \exp\biggl\{ 
		\frac{\lambda}{2} 
        (\gammav_{1}, \gammav_{2})^{\T} \nabla \UP(\ups) 
	\biggr\}
	\\
	& \leq &
	\frac{1}{2} \log \E \exp\biggl\{ 
		{\lambda} \gammav_{1}^{\T} \nabla_{1} \UP(\ups) 
	\biggr\}
	+
	\frac{1}{2} \log \E \exp\biggl\{ 
		{\lambda} \gammav_{2}^{\T} \nabla_{2} \UP(\ups)
	\biggr\}
	\\
	& \leq &
	\frac{1}{2} \log \E \exp\biggl\{ 
		{\lambda} \gammav_{1}^{\T} \nabla_{1} \UP(\ups) 
	\biggr\} 
	+
	\frac{1}{2} \log \E \exp\biggl\{ 
		{\lambda} \gammav_{2}^{\T} \nabla_{2} \UP(\ups)
	\biggr\} 
	\\
	& \leq &
	\frac{\nunu^{2} \lambda^{2}}{2},
	\qquad
	|\lambda| \leq \gm .
\label{lexpuvUpsbi}
\end{EQA}
This means that the bivariate process \( \UP(\ups)/2 \) fulfills the full dimensional 
condition \nameref{CSDref} with \( \VV = \block(\VV_{1},\VV_{2}) \).
Let \( \ups = (\ups_{1},\ups_{2}) \) and \( \upsd = (\upsd_{1},\upsd_{2}) \) be a couple 
of points in \( \Ups \) such that 
\( \| \VV_{j} (\ups_{j} - \upsd_{j}) \| \leq \eps \) for \( j=1,2 \).
Then obviously 
\begin{EQA}
	\| \VV (\ups - \upsd) \|^{2}
	& \leq &
	2 \eps^{2} .
\label{VVuudbi}
\end{EQA}
Therefore, the direct product of two \( \eps \)-nets \( \MM_{j}(\eps) \) in 
\( \Ups_{j} \) for \( j=1,2 \)
yield a \( \sqrt{2} \eps \)-net \( \MM(\eps) = \MM_{1}(\eps) \times \MM_{2}(\eps) \) in 
the product space \( \Ups \).

Due to \eqref{VVuudbi}, the product set 
\( \Upss \eqdef \Ups_{1}(\rupd) \times \Ups_{2}(\rupd) \) 
has the radius \( \rupd \).
Now we can easily bound the entropy of the product set \( \Upss \) via the entropy
of \( \Ups_{1} \) and \( \Ups_{2} \).
Indeed, it holds with \( \rr_{k} = 2^{-k} \rupd \) 
for the cardinality \( \NN_{k} \) of \( \MM_{k} = \MM(\rr_{k}) \)
\begin{EQA}
	\NN_{k}
	&=&
	\NN_{k}(\Ups_{1}) \NN_{k}(\Ups_{2})
\label{NNkNN1NN2}
\end{EQA}
and 
\begin{EQA}
	\entrlg(\Upss)
	& \leq &
	\sum_{k=1}^{\infty} 2^{-k+1} \log(2 \NN_{k})
	\\
	& \leq &
	\sum_{k=1}^{\infty} 2^{-k+1} \log(2 \NN_{k}(\Ups_{1})) 
	+ \sum_{k=1}^{\infty} 2^{-k+1} \log(2 \NN_{k}(\Ups_{2}))
	\leq 
	\entrlg(\Ups_{1}) + \entrlg(\Ups_{2}) .
\label{entrgUU1U2}
\end{EQA}
Similarly
\begin{EQA}
	\entrlq(\Upss)
	& \leq &
	\sum_{k=1}^{\infty} 2^{-k} \sqrt{2 \log(2 \NN_{k})}
	\\
	& \leq &
	\sum_{k=1}^{\infty} 2^{-k} 
		\sqrt{2 \log(2 \NN_{k}(\Ups_{1})) + 2 \log(2 \NN_{k}(\Ups_{2}))}
	\leq 
	\entrlq(\Ups_{1}) + \entrlq(\Ups_{2}) .
\label{entrqUU1U2}
\end{EQA}
Now we just apply the assertion of Theorem~\ref{Tsmoothpenlc} to the process 
\( \UP(\ups)/2 \) and account for the fact that by \eqref{VVuudbi} the radius of \( \Upss \) is 
\( \sqrt{2} \rupd \).
\end{proof}

\Section{A bound for the norm of a vector random process}
Let \( \UU(\ups) \), \( \ups \in \Ups \), be a smooth centered random vector process 
with values in \( \R^{\dimq} \), where \( \Ups \subseteq \R^{\dimp} \). 
Let also \( \UU(\upss) = 0 \) for a fixed point \( \upss \in \Ups \).
Without loss of generality assume \( \upss = 0 \).
We aim to bound the maximum of the norm \( \| \UU(\ups) \| \)
over a vicinity \( \Upss \) of \( \upss \).
By \( \nabla \UP(\ups) \) we denote the \( \dimp \times \dimq \) matrix with entries
\( \nabla_{\ups_{i}} \UP_{j} \), \( i \leq \dimp \), \( j \leq \dimq \).
Suppose that \( \UU(\ups) \) satisfies for each 
\( \gammav_{1} \in \R^{\dimp} \) and \( \gammav_{2} \in \R^{\dimq} \) with 
\( \| \gammav_{1} \| = \| \gammav_{2} \| = 1 \)
\begin{EQA}[c]
    \sup_{\ups \in \Ups} \log \E \exp\Bigl\{ 
		\lambda \gammav_{1}^{\T} \nabla \UU(\ups) \gammav_{2} 
	\Bigr\} 
    \le 
    \frac{\nunu^{2} \lambda^{2}}{2} , 
    \qquad 
    |\lambda| \leq \gm.
\label{gUUgem}
\end{EQA}
Condition \eqref{gUUgem}
implies for any \( \ups \in \Upss \) with \( \| \ups \| \leq \rr \) and
\( \gammav \in \R^{\dimq} \) with
\( \| \gammav \| = 1 \) in view of \( \UU(\upss) = 0 \) by Lemma~\ref{Lsmu}
\begin{EQA}[c]
    \log \E \exp\Bigl\{ 
    	\frac{\lambda}{\rr} \UU(\ups)^{\T} \gammav
	\Bigr\} 
    \le 
    \frac{\nunu^{2} \lambda^{2} \| \ups \|^{2}}{2 \rr^{2}} ,
    \qquad 
    |\lambda| \leq \gm .
\label{gUUem}
\end{EQA}   
In what follows, we use the representation
\begin{EQA}[c]
    \| \UU(\ups) \|
    =
    \sup_{\| \uv \| \leq \rr} \frac{1}{\rr} \uv^{\T} \UU(\ups) .
\label{UU2ups}
\end{EQA}    
This implies for 
\( \Upss(\rr) = \bigl\{ \ups \in \Ups \colon \| \ups - \upss \| \leq \rr \bigr\} \)
\begin{EQA}[c]
    \sup_{\ups \in \Upss(\rr)} \| \UU(\ups) \|
    =
    \sup_{\ups \in \Upss(\rr)} \,\, \sup_{\| \uv \| \leq \rr} 
        \frac{1}{\rr}  \uv^{\T} \UU(\ups) .
\label{UU2vups}
\end{EQA}   
Consider a bivariate process \( \uv^{\T} \UU(\ups) \) of 
\( \uv \in \R^{\dimq} \) and \( \ups \in \Ups \subset \R^{\dimp} \).
By definition \( \E \uv^{\T} \UU(\ups) = 0 \).
Further, 
\( \nabla_{\uv} \bigl[ \uv^{\T} \UU(\ups) \bigr] = \UU(\ups) \) while 
\( \nabla_{\ups} \bigl[ \uv^{\T} \UU(\ups) \bigr] = \uv^{\T} \nabla \UU(\ups) 
= \| \uv \| \gammav^{\T} \nabla \UU(\ups) \) for \( \gammav = \uv / \| \uv \| \).
Suppose that \( \uv \in \R^{\dimq} \) and \( \ups \in \Ups \) are such that 
\( \| \uv \| \leq \rr \) and \( \| \ups \| \leq \rr \).
By \eqref{gUUgem}, 
it holds for \( \gammav \in \R^{\dimp} \) with
\( \| \gammav \| = 1 \) and \( \ups \in \Upss(\rr) \)
\begin{EQA}
	\log \E \exp\biggl\{ 
		\frac{\lambda}{\rr} \nabla_{\ups} \bigl[ \uv^{\T} \UU(\ups) \bigr] \gammav
	\biggr\}
	& \leq &
	\log \E \exp\biggl\{ 
		\frac{\lambda}{\rr} \uv^{\T} \nabla \UU(\ups) \gammav
	\biggr\}
	\leq 
	\frac{\nunu^{2} \lambda^{2}}{2} \, ,
\label{logEexlrnuuTU}
\end{EQA}
and by \eqref{gUUem} for a unit vector \( \gammav \in \R^{\dimq} \)
\begin{EQA}
	\log \E \exp\biggl\{ 
		\frac{\lambda}{\rr} \nabla_{\uv} \bigl[ \uv^{\T} \UU(\ups) \bigr] \gammav
	\biggr\}
	& \leq &
	\log \E \exp\biggl\{ 
		\frac{\lambda}{\rr} \UU(\ups) \gammav
	\biggr\}
	\leq 
	\frac{\nunu^{2} \lambda^{2}}{2} \, .
\label{lEelrnvUg}
\end{EQA}
Therefore, \nameref{CSD12ref} is fulfilled for \( \uv^{\T} \UU(\ups) \) and 
Theorem~\ref{Tsmoothpenlc} applies.
We summarize our findings in the following theorem.

\begin{theorem}
\label{Tsqnorm}
Let a random \( \dimp \)-vector process \( \UU(\ups) \) for 
\( \ups \in \Ups \subseteq \R^{\dimp} \)
fulfill \( \UU(\upss) = 0 \), \( \E \UU(\ups) \equiv 0 \), 
and the condition \eqref{gUUgem} be satisfied. 
Then for each \( \rr \) and any \( \xx \ge 1/2 \), it holds for \( \Upss = \Upss(\rr) \)
\begin{EQA}
\label{Upsdboundno}
    \P \Bigl\{ 
        \sup_{\ups \in \Upss(\rr)} \bigl\| \UU(\ups) \bigr\| 
    	\geq 
		\sqrt{8} \nunu \rr \, \zzQ(\xx)
    \Bigr\}
    & \le & 
    \ex^{-\xx} ,
\end{EQA}
where \( \zzQ(\xx) \) is given by \eqref{zzxxgfin} with 
\( \entrlq = \entrlq(\Upss) + \sqrt{6 \dimq} \) and \( \entrlq = \entrlq(\Upss) + 6 \dimq \).
\end{theorem}

\Section{A bound for a family of quadratic forms}
Now we consider an extension of the previous result with a quadratic form
\( \| \AA \, \UU(\ups) \|^{2} \) to be bounded under
the conditions \eqref{gUUgem} and \eqref{gUUem} on \( \UU(\ups) \) for \( \ups \in \Ups \subset \R^{\dimp} \).
Here \( \UU(\cdot) \) is a vector process with values in \( \R^{\dimq} \) and
\( A \) is  a 
\( \dimq \times \dimq \) matrix with \( \| \AA^{\T} \AA \|_{\oper} \le 1 \).
The idea is to use the representation \eqref{UU2ups} in which we replace
\( \uv \) with \( \AA \uv \).
The bound \eqref{Upsdboundno} implies for any \( \rr \)
\begin{EQA}
	\P\Bigl\{
        \sup_{\ups \in \Upss(\rr), \,\, \| \AA \uv \| \leq \rr} \,\,
        \uv^{\T} \AA \UU(\ups) > \sqrt{8} \, \nunu \, \rr \, \zzQ(\xx)
    \Bigr\}
    & \le &
    \ex^{-\xx} ,
\label{bouuvupsdxA}
\end{EQA}
where \( \zzQ(\xx) \) corresponds to 
\( \entrlq = \sqrt{\entrlg} = \sqrt{ 6 \dimp + \entrlg(\Upss)} \).

Now we discuss how this bound can be refined if \( \AA \) is a smoothing operator.
For simplicity assume that \( \AA \) fulfills the condition of Theorem~\ref{HGLsumfF}.
One can expect that the dimension \( \dimq \) can be replaced by the effective dimension
\( \dimA_{\AA} \).
The arguments similar to the above yield
\begin{EQA}[c]
    \| \AA \, \UU(\ups) \|
    =
    \sup_{\uv \in \R^{\dimq} \colon \| \uv \| \leq \rr}
        \frac{1}{\rr} \uv^{\T} \AA \, \UU(\ups)  ,
\label{UU2upsA}
\end{EQA}
and we again consider a bivariate process \( \uv^{\T} \AA \, \UU(\ups) \) of
\( \uv \in \R^{\dimq} \) and \( \ups \in \Ups \subset \R^{\dimp} \).
The conditions \eqref{gUUgem} and \eqref{gUUem} imply for any two unit vectors
\( \gammav_{1} \in \R^{\dimq} \) 
and \( \gammav_{2} \in \R^{\dimp} \) 
and any points \( \uv \in \R^{\dimq} \) with
\( \| \AA \uv \| \leq \rr \) and \( \ups \in \Upss(\rr) \), it holds 
\begin{EQA}
	\log \E \exp\biggl\{ 
		\frac{\lambda}{\rr} \, \nabla_{\ups} \bigl[ \uv^{\T} \AA \, \UU(\ups) \bigr] \gammav_{2}
	\biggr\}
	& = &
	\log \E \exp\biggl\{ 
		\frac{\lambda}{\rr} \, \uv^{\T} \AA \, \nabla \UU(\ups) \gammav_{2}
	\biggr\}
	\leq 
	\frac{\nunu^{2} \lambda^{2}}{2} \, ,
\label{logEexlrnuuTUAA}
\end{EQA}
and by \eqref{gUUem} with \( \VVc_{1}^{2} = \AA^{\T} \AA \)
\begin{EQA}
	\log \E \exp\biggl\{ 
		\frac{\lambda} {\| \VVc_{1} \gammav_{1} \|} \gammav_{1}^{\T} \nabla_{\uv} 
		\bigl[ \uv^{\T} \AA \, \UU(\ups) \bigr]
	\biggr\}
	& \leq &
	\log \E \exp\biggl\{ 
		\frac{\lambda}{\| \VVc_{1} \gammav_{1} \|} (\AA \, \gammav_{1})^{\T} \UU(\ups)
	\biggr\}
	\leq 
	\frac{\nunu^{2} \lambda^{2}}{2} \, .
\label{lEelrnvUgAA}
\end{EQA}
Therefore, \nameref{CSD12ref} is fulfilled for \( \uv^{\T} \AA \, \UU(\ups) \).
Now we apply the bound from Theorem~\ref{TUPprod} and the entropy bound for the elliptic set 
\( \| \AA \uv \| \leq \rr \) from Theorem~\ref{HGLsumfF}.

\begin{theorem}
\label{TexproA}
Let a random vector process \( \UU(\ups) \in \R^{\dimq} \) for
\( \ups \in \Ups \subseteq \R^{\dimp} \)
fulfill \( \UU(\upss) = 0 \), \( \E \UU(\ups) \equiv 0 \),
and the condition \eqref{gUUgem} be satisfied.
Let \( \AA \) fulfill \( 1/2 \le \| \AA \AA^{\T} \|_{\oper} \le 1 \).
Then for each \( \rr \), it holds
\begin{EQA}
	\P\biggl\{
		\sup_{\ups \in \Upss(\rr)} \,\,
		\|\AA \, \UU(\ups)\| > \sqrt{8} \, \nunu \, \rr \, \zzQ(\xx)			
	\biggr\}
	& \leq &
	\ex^{-\xx} .
\label{bouuvA}
\end{EQA}
where \( \zzQ(\xx) \) is given by \eqref{zzxxgfin} with
\( \entrlg = \CONST \, \dimA_{\AA} + \entrlg(\Upss(\rr)) \) and
\( \entrlq = \CONST \sqrt{\dimA_{\AA}} + \entrlq(\Upss(\rr)) \).
\end{theorem}

\bibliography{exp_ts,listpubm-with-url}
\end{document}